\pdfoutput=1

\newif\ifelsarticle
\elsarticlefalse

\ifelsarticle
	\documentclass[3p]{elsarticle}
\else
	\documentclass[11pt,a4paper]{article}
\fi

\usepackage[T1]{fontenc}
\usepackage[english]{babel}
\usepackage[utf8]{inputenc}
\ifelsarticle\else
	\usepackage[hmargin=1.0in,top=0.975in,bottom=1.95in,footskip=50pt]{geometry}
\fi

\usepackage{microtype}

\ifelsarticle
\else
	\usepackage[backend=biber,style=ieee-alphabetic,sorting=nyt,isbn=false,maxbibnames=99]{biblatex}
\fi

\usepackage{csquotes}
\ifelsarticle
	\usepackage{etoolbox}
	\makeatletter
	\def\@seccntformat#1{\@ifundefined{#1@cntformat}%
		{\csname the#1\endcsname.\hskip0.5em}    
		{\csname #1@cntformat\endcsname}
	}
	\patchcmd{\appendix}{\appendixname}{}{}{}
	\appto{\appendix}{%
		\newcommand{\section@cntformat}{\appendixname\ \thesection.\hskip0.5em}
	}
	\makeatother
\else
	\usepackage{authblk}
\fi

\usepackage{enumitem}
\usepackage[cmex10]{amsmath}
\usepackage{amssymb,amsthm}
\usepackage{mathtools}
\usepackage{dsfont}
\usepackage{bm}
\usepackage{caption}
\usepackage{subcaption}
\usepackage{aliascnt}

\usepackage{tikz}
\usetikzlibrary{fit,calc,positioning}

\usepackage[colorlinks,allcolors=blue!60!black]{hyperref}

\ifelsarticle\else
\fi

\theoremstyle{definition}
\newtheorem{definition}{Definition}[section]

\theoremstyle{plain}
\newaliascnt{theorem}{definition}
\newtheorem{theorem}[theorem]{Theorem}
\aliascntresetthe{theorem}
\newaliascnt{lemma}{definition}
\newtheorem{lemma}[lemma]{Lemma}
\aliascntresetthe{lemma}
\newaliascnt{proposition}{definition}

\aliascntresetthe{proposition}
\newaliascnt{corollary}{definition}
\newtheorem{corollary}[corollary]{Corollary}
\aliascntresetthe{corollary}

\newtheorem{claim}{Claim}
\makeatletter
\@addtoreset{claim}{theorem} 
\makeatother

\theoremstyle{remark}
\newtheorem*{remark}{Remark}

\newenvironment{claimproof}{\proof[Proof of claim]}{\endproof}

\addto\extrasenglish{%

}

\DeclarePairedDelimiter{\abs}{\lvert}{\rvert}

\newcommand{\symdif}{\mathbin{\triangle}}

\newcommand{\N}{\mathds{N}}

\newcommand{\Q}{\mathds{Q}}
\newcommand{\R}{\mathds{R}}

\newcommand{\abstSys}[1]{#1}
\newcommand{\sys}[1]{\mathbf{#1}}
\newcommand{\sysSpace}{\mathfrak{S}}
\newcommand{\sysIntSpace}[2]{\sysSpace_{#1,#2}} 
\newcommand{\tuple}[1]{\mathbf{#1}}

\newcommand{\val}{v}
\newcommand{\tim}{t}
\newcommand{\valsp}{\mathcal{V}}
\newcommand{\timesp}{\mathcal{T}}

\newcommand{\inSet}{\Lambda}
\newcommand{\inInSet}{\mathcal{I}}
\newcommand{\outInSet}{\mathcal{O}}
\newcommand{\inIn}{i}
\newcommand{\outIn}{o}

\newcommand{\fixedPoints}[4]{\Phi_{#2, #3, #4}^{#1}}
\newcommand{\fpChoiceFun}{\phi}
\newcommand{\chosenFixedPoint}[4]{\fpChoiceFun_{#2, #3}^{#1}(#4)}

\newcommand{\genFunSys}[3]{\mathfrak{F}(#1, #2, #3)}

\newcommand{\intSetFunct}{\lambda}

\newcommand{\parac}{\mathbin{\parallel}}
\newcommand{\connectSet}{\Gamma}
\newcommand{\connectFunc}{\gamma}
\newcommand{\connect}[2]{\connectFunc_{\{#1, #2\}}}

\newcommand{\monoSysSpace}{\mathfrak{M}}
\newcommand{\monoSysIntSpace}[2]{\monoSysSpace_{#1,#2}} 

\newcommand{\contSysSpace}{\mathfrak{N}}
\newcommand{\contSysIntSpace}[2]{\contSysSpace_{#1,#2}} 

\newcommand{\causalSysSpace}{\mathfrak{C}}
\newcommand{\causalSysIntSpace}[2]{\causalSysSpace_{#1,#2}} 

\ifelsarticle
	\journal{Journal of Theoretical Computer Science}
\else
	\title{\textbf{Toward an Algebraic Theory of Systems\thanks{
	© 2018. This manuscript version is made available under the CC-BY-NC-ND 4.0 license
	\href{https://creativecommons.org/licenses/by-nc-nd/4.0/}{\nolinkurl{https://creativecommons.org/licenses/by-nc-nd/4.0/}}.
	The version published in \emph{Theoretical Computer Science} is available at \href{https://doi.org/10.1016/j.tcs.2018.06.001}{\nolinkurl{https://doi.org/10.1016/j.tcs.2018.06.001}}.}}}
	
	\author[1]{Christian~Matt\thanks{Present address: Department of Computer Science, University of California, Santa Barbara, CA 93106, USA.}}
	\author[1]{Ueli~Maurer}
	\author[2]{Christopher~Portmann\thanks{Present address: Department of Computer Science, ETH Zurich, 8092 Zurich, Switzerland.}}
	\author[2]{\authorcr Renato~Renner}
	\author[3]{Björn~Tackmann\thanks{Present address: IBM Research -- Zurich, 8803 Rüschlikon, Switzerland.}}
	
	\affil[1]{Department of Computer Science, ETH Zurich, 8092 Zurich, Switzerland. \authorcr \href{mailto:cmatt@cs.ucsb.edu}{\nolinkurl{cmatt@cs.ucsb.edu}}, \href{mailto:maurer@inf.ethz.ch}{\nolinkurl{maurer@inf.ethz.ch}}}
	\affil[2]{Institute for Theoretical Physics, ETH Zurich, 8093 Zurich, Switzerland. \authorcr
	\href{mailto:chportma@inf.ethz.ch}{\nolinkurl{chportma@inf.ethz.ch}}, \href{mailto:renner@phys.ethz.ch}{\nolinkurl{renner@phys.ethz.ch}}}
	\affil[3]{Department of Computer Science and Engineering, University of California, San Diego, La Jolla, CA 92093, USA. \authorcr \href{mailto:bta@zurich.ibm.com}{\nolinkurl{bta@zurich.ibm.com}}}
	
	\date{}
\fi

\hypersetup{
	pdftitle={Toward an Algebraic Theory of Systems},
	pdfauthor={Christian Matt, Ueli Maurer, Christopher Portmann, Renato Renner, and Björn Tackmann}
}

\begin{document}

\ifelsarticle
	\begin{frontmatter}
	
	\title{Toward an Algebraic Theory of Systems}
		
	\author[ETHCS]{Christian~Matt\corref{cor}\fnref{fn1}}
	\ead{cmatt@cs.ucsb.edu}
	
	\author[ETHCS]{Ueli~Maurer}
	\ead{maurer@inf.ethz.ch}
	
	\author[ETHPH]{Christopher~Portmann\fnref{fn2}}
	\ead{chportma@inf.ethz.ch}
	
	\author[ETHPH]{Renato~Renner}
	\ead{renner@phys.ethz.ch}
	
	\author[UCSD]{Björn~Tackmann\fnref{fn3}}
	\ead{bta@zurich.ibm.com}
		
	\cortext[cor]{Corresponding author}
	\fntext[fn1]{Present address: Department of Computer Science, University of California, Santa Barbara, CA 93106, USA}
	\fntext[fn2]{Present address: Department of Computer Science, ETH Zurich, 8092 Zurich, Switzerland}
	\fntext[fn3]{Present address: IBM Research -- Zurich, 8803 Rüschlikon, Switzerland}
		
	\address[ETHCS]{Department of Computer Science, ETH Zurich, 8092 Zurich, Switzerland}
	\address[ETHPH]{Institute for Theoretical Physics, ETH Zurich, 8093 Zurich, Switzerland}
	\address[UCSD]{Department of Computer Science and Engineering, University of California, San Diego, La Jolla, CA 92093, USA}
\else
	\maketitle
\fi

\begin{abstract}
	We propose the concept of a system algebra with a parallel composition operation and an interface connection operation, and formalize composition-order invariance, which postulates that the order of composing and connecting systems is irrelevant, a generalized form of associativity. Composition-order invariance explicitly captures a common property that is implicit in any context where one can draw a figure (hiding the drawing order) of several connected systems, which appears in many scientific contexts. This abstract algebra captures settings where one is interested in the behavior of a composed system in an environment and wants to abstract away anything internal not relevant for the behavior. This may include physical systems, electronic circuits, or interacting distributed systems.

	One specific such setting, of special interest in computer science, are functional system algebras, which capture, in the most general sense, any type of system that takes inputs and produces outputs depending on the inputs, and where the output of a system can be the input to another system. The behavior of such a system is uniquely determined by the function mapping inputs to outputs. We consider several instantiations of this very general concept. In particular, we show that Kahn networks form a functional system algebra and prove their composition-order invariance.

	Moreover, we define a functional system algebra of causal systems, characterized by the property that inputs can only influence future outputs, where an abstract partial order relation captures the notion of ``later''. This system algebra is also shown to be composition-order invariant and appropriate instantiations thereof allow to model and analyze systems that depend on time.
	
	\ifelsarticle\else
	\fi
\end{abstract}

\ifelsarticle
	
	\end{frontmatter}
\fi

\section{Introduction}
\subsection{Motivation}
A universal concept in many disciplines is to characterize the behavior of an object, often called a module or a system, and thereby to intentionally ignore internal aspects considered irrelevant. We will here use the term system. The purpose of a system is to be used or embedded in an environment, which can for example consist of several other systems. A system's behavior is the exact characterization of the effect a system can have when embedded in an environment. In such a consideration, anything internal, not affecting the behavior, is by definition considered irrelevant, and two systems with the same behavior are considered to be equal. Which aspects are considered irrelevant, and hence are not modeled as part of the behavior, depends strongly on the concrete investigated question.
For example, if for a software module computing a function, one is only interested in the input-output behavior, then this function characterizes the behavior of the software module, and aspects like the underlying computational model, the program, the complexity, timing guarantees, etc., are irrelevant.  In another consideration, for example, one may want to include the timing aspects as part of the observable (and hence relevant) behavior. In yet another consideration one may be interested in the memory requirements for the specific computational model, etc.

One can think of a system as connecting to the environment via an \emph{interface}. More generally, if one wants to model a system composed of several (sub-)systems, then one can consider each system to have several interfaces and that systems can be composed by connecting an interface of one system with an interface of another system. The term ``interface'' is here used in the sense of capturing the potential connection to another system; it is indeed used in some contexts to refer to the specification of the behavior of an object, such as a software module, but here we use the term in a more general and more abstract sense. Our notion of connected systems corresponds to drawing a diagram with boxes, each having several lines (interfaces), and where some interfaces of systems are connected (by lines) and some interfaces remain free, i.e., accessible by the environment; see \autoref{fig:motivation}. A composed system appears to the environment as having only the free (not connected) interfaces and its behavior is observed only via these free interfaces; the internal topology becomes irrelevant and is not part of the behavior.

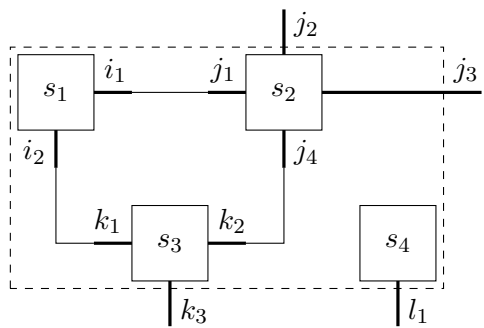
\begin{figure}[b]
	\centering
	\begin{tikzpicture}[interface/.style={draw,very thick},connection/.style={draw}]
		\draw (0, 2) rectangle node[midway] {$\abstSys{s}_1$} (1, 3);
		\draw (3, 2) rectangle node[midway] {$\abstSys{s}_2$} (4, 3);
		\draw (1.5, 0) rectangle node[midway] {$\abstSys{s}_3$} (2.5, 1);
		\draw (4.5, 0) rectangle node[midway] {$\abstSys{s}_4$} (5.5, 1);
		\draw[dashed] (-0.1, -0.1) rectangle (5.6, 3.1);
		
		\draw[interface] (1, 2.5) -- (1.5, 2.5);
		\node[above right] at (1, 2.5) {$i_1$};
		\draw[interface] (0.5, 2) -- (0.5, 1.5);
		\node[below left] at (0.5, 2) {$i_2$};
		
		\draw[interface] (3, 2.5) -- (2.5, 2.5);
		\node[above left] at (3, 2.5) {$j_1$};
		\draw[interface] (3.5, 3) -- (3.5, 3.6);
		\node[above right] at (3.5, 3.1) {$j_2$};
		\draw[interface] (4, 2.5) -- (6.1, 2.5);
		\node[above right] at (5.6, 2.5) {$j_3$};
		\draw[interface] (3.5, 2) -- (3.5, 1.5);
		\node[below right] at (3.5, 2) {$j_4$};
		
		\draw[interface] (1.5, 0.5) -- (1, 0.5);
		\node[above left] at (1.5, 0.5) {$k_1$};
		\draw[interface] (2.5, 0.5) -- (3, 0.5);
		\node[above right] at (2.5, 0.5) {$k_2$};
		\draw[interface] (2, 0) -- (2, -0.6);
		\node[below right] at (2, -0.1) {$k_3$};
		
		\draw[interface] (5, 0) -- (5, -0.6);
		\node[below right] at (5, -0.1) {$l_1$};
		
		\draw[connection] (1.5, 2.5) -- (2.5, 2.5);
		\draw[connection] (0.5, 1.5) |- (1, 0.5);
		\draw[connection] (3.5, 1.5) |- (3, 0.5);
	\end{tikzpicture}
	\caption{A system composed of the four subsystems $\abstSys{s}_1$, $\abstSys{s}_2$, $\abstSys{s}_3$, and $\abstSys{s}_4$. Interfaces are labeled to allow specifying connections, e.g., $i_1$ and $j_1$ are connected, and $\abstSys{s}_4$ is not connected to the other systems. The resulting system has interfaces $j_2$, $j_3$, $k_3$, and $l_1$.}
	\label{fig:motivation}
\end{figure}

In some applications, certain internal details of a system matter. One can then define the relevant internal aspects as being part of the behavior, to make them, by definition, visible from the outside. If too many internal details become relevant, our approach might be less suitable than directly using a model of systems that considers these internals. Indeed, a majority of existing work models systems by defining their internal operations, e.g., via states and transition functions. In many cases, however, such a detailed description of the internal operations is unnecessary and cumbersome, and our abstract approach would be beneficial. We now describe two examples for which ignoring internal details appear particularly useful.

\paragraph{Distributed systems\ifelsarticle\else.\fi}
In distributed systems, where systems are connected to other systems with which they can communicate, one is often interested in certain properties of the composed system. As an example, we present a simple impossibility proof of bit-broadcast for three parties with one dishonest party. This famous result was first proven by Lamport et al.~\cite{LaShPe82,PSL80}; the proof we present here is in the spirit of that given by Fisher et al.~\cite{FLM86}.
This proof only requires that the involved systems can be connected and rearranged as described; the communication between them and how they operate internally is irrelevant. Ignoring these internal details not only simplifies the proof but also makes the result more general, e.g., it also holds if the systems communicate via some sort of analog signals.

The goal of a bit-broadcast protocol is to allow a sender to send a bit such that all honest receivers output the same bit (consistency) and if the sender is honest, they output the bit that was sent (validity). Assume an honest sender uses the system~$\abstSys{s}_a$ for broadcasting a bit~$a\in\{0,1\}$, and honest receivers use systems~$\abstSys{r}_1$ and $\abstSys{r}_2$ that decide on a bit $b_1$, and $b_2$, respectively. If these systems implement a broadcast protocol with the required guarantees, each condition in Figures~\ref{fig:bc-correctness} to \ref{fig:bc-dishonest-r2} must hold for all systems~$\abstSys{e}$, which capture the possible behaviors of a dishonest party. Assume toward a contradiction that this is the case and consider the system in \autoref{fig:bc-proof}. We can view the system in the dotted box composed of $\abstSys{s}_0$ and $\abstSys{s}_1$ as a system~$\abstSys{e}$ in \autoref{fig:bc-dishonest-s} to obtain $b_1 = b_2$. We can also view the system in the densely dotted box composed of $\abstSys{s}_0$ and $\abstSys{r}_1$ as a system~$\abstSys{e}$ in \autoref{fig:bc-dishonest-r1} to obtain~$b_2 = 1$. Finally, the system in the dashed box can be viewed as a system~$\abstSys{e}$ in \autoref{fig:bc-dishonest-r2}, which implies $b_1 = 0$, a contradiction. Hence, there are no systems~$\abstSys{s}_0$, $\abstSys{s}_1$, $\abstSys{r}_1$, and $\abstSys{r}_2$ that satisfy these constraints.
In \autoref{sec:bcImp}, we provide a more formal proof within our theory.


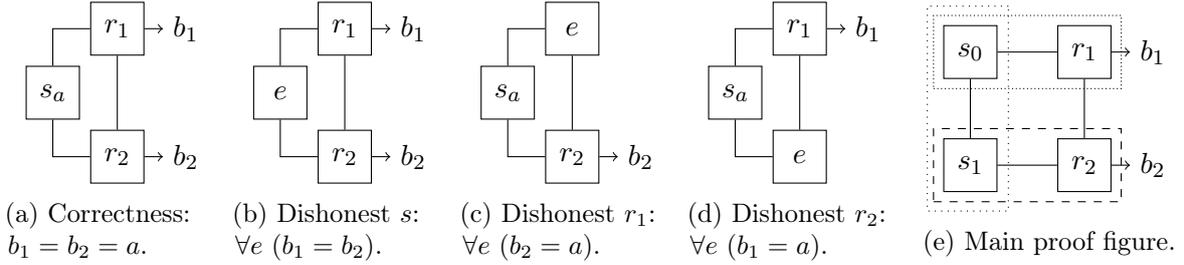
\begin{figure}[t]
	\centering
	\begin{subfigure}{0.18\textwidth}
		\centering
		\begin{tikzpicture}[sys/.style={draw,minimum size=7mm,node distance=12mm},bit/.style={node distance=9mm}]
			\node (s) [sys] {$\abstSys{s}_a$};
			\node (r1) [sys,above right of=s] {$\abstSys{r}_1$};
			\node (r2) [sys,below right of=s] {$\abstSys{r}_2$};
			\node (b1) [bit,right of=r1] {$b_1$};
			\node (b2) [bit,right of=r2] {$b_2$};
			
			\draw (s) |- (r1);
			\draw (s) |- (r2);
			\draw (r1) to (r2);
			\draw [->] (r1) to (b1);
			\draw [->] (r2) to (b2);
		\end{tikzpicture}
		\caption{Correctness:\\ $b_1 = b_2 = a$.}
		\label{fig:bc-correctness}
	\end{subfigure}
	\begin{subfigure}{0.18\textwidth}
		\centering
		\begin{tikzpicture}[sys/.style={draw,minimum size=7mm,node distance=12mm},bit/.style={node distance=9mm}]
			\node (s) [sys] {$\abstSys{e}$};
			\node (r1) [sys,above right of=s] {$\abstSys{r}_1$};
			\node (r2) [sys,below right of=s] {$\abstSys{r}_2$};
			\node (b1) [bit,right of=r1] {$b_1$};
			\node (b2) [bit,right of=r2] {$b_2$};
			
			\draw (s) |- (r1);
			\draw (s) |- (r2);
			\draw (r1) to (r2);
			\draw [->] (r1) to (b1);
			\draw [->] (r2) to (b2);
		\end{tikzpicture}
		\caption{Dishonest $\abstSys{s}$:\\$\forall \abstSys{e} \ (b_1 = b_2)$.}
		\label{fig:bc-dishonest-s}
	\end{subfigure}
	\begin{subfigure}{0.18\textwidth}
		\centering
		\begin{tikzpicture}[sys/.style={draw,minimum size=7mm,node distance=12mm},bit/.style={node distance=9mm}]
			\node (s) [sys] {$\abstSys{s}_a$};
			\node (r1) [sys,above right of=s] {$\abstSys{e}$};
			\node (r2) [sys,below right of=s] {$\abstSys{r}_2$};
			\node (b2) [bit,right of=r2] {$b_2$};
			
			\draw (s) |- (r1);
			\draw (s) |- (r2);
			\draw (r1) to (r2);
			\draw [->] (r2) to (b2);
		\end{tikzpicture}
		\caption{Dishonest $\abstSys{r}_1$: \\$\forall \abstSys{e} \ (b_2 = a)$.}
		\label{fig:bc-dishonest-r1}
	\end{subfigure}
	\begin{subfigure}{0.18\textwidth}
		\centering
		\begin{tikzpicture}[sys/.style={draw,minimum size=7mm,node distance=12mm},bit/.style={node distance=9mm}]
			\node (s) [sys] {$\abstSys{s}_a$};
			\node (r1) [sys,above right of=s] {$\abstSys{r}_1$};
			\node (r2) [sys,below right of=s] {$\abstSys{e}$};
			\node (b1) [bit,right of=r1] {$b_1$};
			
			\draw (s) |- (r1);
			\draw (s) |- (r2);
			\draw (r1) to (r2);
			\draw [->] (r1) to (b1);
		\end{tikzpicture}
		\caption{Dishonest $\abstSys{r}_2$: \\$\forall \abstSys{e} \ (b_1 = a)$.}
		\label{fig:bc-dishonest-r2}
	\end{subfigure}
	\begin{subfigure}{0.22\textwidth}
		\centering
		\begin{tikzpicture}[sys/.style={draw,minimum size=7mm,node distance=15mm},bit/.style={node distance=9mm}]
			\node (s0) [sys] {$\abstSys{s}_0$};
			\node (s1) [sys,below of=s0] {$\abstSys{s}_1$};
			\node (r1)  [sys,right of=s0] {$\abstSys{r}_1$};
			\node (r2)  [sys,below of=r1] {$\abstSys{r}_2$};
			\node (b1)  [bit,right of=r1] {$b_1$};
			\node (b2)  [bit,right of=r2] {$b_2$};
			
			\draw (s0) to (r1) to (r2) to (s1) to (s0);
			\draw [->] (r1) to (b1);
			\draw [->] (r2) to (b2);
			
			\node [rectangle,scale=1.1,draw=black,dotted,fit=(s0)(s1),xshift=-0.3mm] {};
			\node [rectangle,draw=black,densely dotted,fit=(s0)(r1)] {};
			\node [rectangle,draw=black,dashed,fit=(s1)(r2)] {};
		\end{tikzpicture}
		\caption{Main proof figure.}
		\label{fig:bc-proof}
	\end{subfigure}
	\caption{The impossibility of broadcast for three parties with one dishonest party. Figures~\ref{fig:bc-correctness} to \ref{fig:bc-dishonest-r2} correspond to the requirements for a broadcast protocol. \autoref{fig:bc-proof} can be subdivided in different ways and is used in the proof. Interface labels are omitted to increase readability.}
	\label{fig:bc}
\end{figure}

\paragraph{Cryptography\ifelsarticle\else.\fi}
Cryptographic schemes are often defined as some sort of efficient algorithms. While efficiency is of course relevant in practice, one can separate the computational aspects from the functionality to simplify the analysis. \emph{Constructive cryptography} by Maurer and Renner~\cite{MR11,Mau12} allows one to model what cryptographic protocols achieve using a system algebra that abstracts away cumbersome details. To this end, one considers so-called \emph{resource systems}, which provide a certain functionality to the parties connected to their interfaces, and \emph{converter systems}, which can be connected to resources to obtain a new resource. Typical resources with interfaces for two honest parties and an adversary are a shared secret key, which provides a randomly generated key at the interfaces for the honest parties and nothing at the interface for the adversary,\footnote{One could also view a shared secret key as having no interface for the adversary, but as defined in constructive cryptography, all resources involved in a construction have an interface for each party.} and different types of channels with different capabilities for the adversary. The goal of a cryptographic protocol is then to \emph{construct} a resource~$S$ from a resource~$R$. Such a protocol consists of a converter for each honest party, and it achieves the construction if the system obtained from $R$ by connecting the protocol converters to the interfaces of the honest parties is \emph{indistinguishable} from the system obtained from $S$ by connecting some converter system, called \emph{simulator}, to the interface of the adversary. The notion of indistinguishability can be defined in several ways leading to different types of security, e.g., as the systems being identical or via a certain class of distinguishers.

A crucial property of this construction notion is that it is \emph{composable}, i.e., if some protocol constructs a resource~$S$ from a resource~$R$ and another protocol constructs a resource~$T$ from $S$, these protocols can be composed to obtain a construction of $T$ from $R$. Turned around, one can also decompose the construction of $T$ from $R$ into two separate constructions. Since these two constructions can be analyzed independently, this approach provides modularity and simplifies the analysis of complex protocols by breaking them down into smaller parts.

An example of a construction is that of a secure channel, which leaks only the length of the sent messages to the adversary and does not allow modifications of them, from the resource consisting of a shared secret key and an authenticated channel, which leaks the sent messages to the adversary but also does not allow modifications of them, by a symmetric encryption scheme. See \autoref{fig:symEnc} for an illustration of the involved systems. To achieve this construction, the simulator must, knowing only the length of the sent messages, output bit-strings that are indistinguishable from encryptions of these messages. If the used encryption scheme is, e.g., the one-time pad, the two systems in \autoref{fig:symEnc} are identical for an appropriate simulator, i.e., they have the same input-output behavior \cite{Mau12}, which provides the strongest possible security guarantee. Note that internally, these systems are very different, but this is intentionally ignored.

\begin{figure}[t]
	\centering
	\begin{subfigure}{0.45\textwidth}
		\centering
		\begin{tikzpicture}
			\def\delta{0.5}
			\draw (-1.5, 0.5*\delta) rectangle node[midway] {$\mathsf{Key}$} (1.5, 1.25);
			\draw (-1.5, -1.25) rectangle node[midway] {$\mathsf{AuthChannel}$} (1.5, -0.5*\delta);
			\draw (-3, -1.25) rectangle node[midway] {$\mathsf{enc}$} (-1.5-\delta, 1.25);
			\draw (1.5+\delta, -1.25) rectangle node[midway] {$\mathsf{dec}$} (3, 1.25);
			\draw (-3, 0) -- (-3.5, 0);
			\draw (3, 0) -- (3.5, 0);	
			\draw (-1.5, 0.25*\delta + 0.625) -- (-1.5-\delta, 0.25*\delta + 0.625);
			\draw (1.5, 0.25*\delta + 0.625) -- (1.5+\delta, 0.25*\delta + 0.625);
			\draw (-1.5, -0.25*\delta - 0.625) -- (-1.5-\delta, -0.25*\delta - 0.625);
			\draw (1.5, -0.25*\delta - 0.625) -- (1.5+\delta, -0.25*\delta - 0.625);
			\draw (-0.2, 0.5*\delta) -- (-0.2, -0.5*\delta);
			\draw[dotted] (-0.2, -0.5*\delta) -- (-0.2, -1.25);
			\draw (-0.2, -1.25) -- (-0.2, -1.75);
			\draw (0.2, -1.25) -- (0.2, -1.75);
			\draw[dashed] (-3.1, -1.35) rectangle (3.1, 1.35);
		\end{tikzpicture}
	\end{subfigure}
	\qquad
	\begin{subfigure}{0.45\textwidth}
		\centering
		\begin{tikzpicture}
			\def\delta{0.5}
			\draw (-1.5, 0.5*\delta) rectangle node[midway] {$\mathsf{SecChannel}$} (1.5, 1.25);
			\draw (-0.5, -1.25) rectangle node[midway] {$\mathsf{sim}$} (0.5, -0.5*\delta);
			\draw (-1.5, 0.25*\delta + 0.625) -- (-2, 0.25*\delta + 0.625) -- (-2, 0) -- (-3.5, 0);
			\draw (1.5, 0.25*\delta + 0.625) -- (2, 0.25*\delta + 0.625) -- (2, 0) -- (3.5, 0);
			\draw (0, 0.5*\delta) -- (0, -0.5*\delta);
			\draw (-0.2, -1.25) -- (-0.2, -1.75);
			\draw (0.2, -1.25) -- (0.2, -1.75);
			\draw[dashed] (-3.1, -1.35) rectangle (3.1, 1.35);
		\end{tikzpicture}
	\end{subfigure}
	\caption{The systems appearing in the construction of a secure channel from an authenticated channel and a shared secret key via symmetric encryption. The construction notion requires the two systems in the dashed boxes to be indistinguishable.}
	\label{fig:symEnc}
\end{figure}
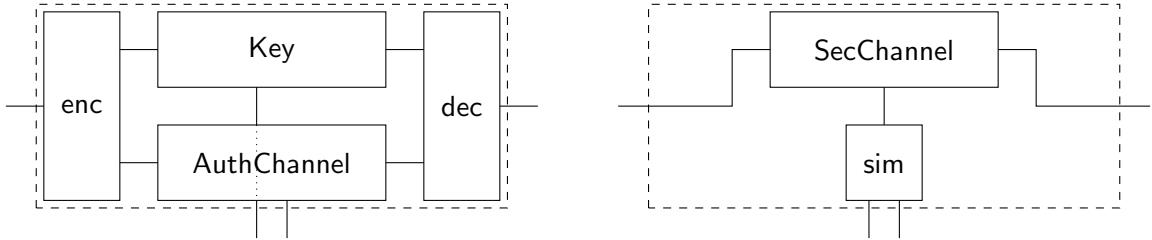

\subsection{Contributions}
We develop a theory of systems with different levels of abstraction. To achieve generality and to strive for simplicity, theorems are proved at the highest level of abstraction at which they hold. See \autoref{fig:syslevels} for an overview of the types of systems we consider. This is not meant to be a complete picture of all systems one can consider; systems we do not consider in this paper but could be treated within our theory include probabilistic systems, physical systems, circuits, etc.

\begin{figure}[htb]
	\centering
	\begin{tikzpicture}[level 1/.style={sibling distance=3.2cm},
	                    level 2/.style={sibling distance=3.6cm},
	                    level distance=1.6cm,
	                    edge from parent path={(\tikzparentnode.south) -- ++(0, -0.25cm) -| (\tikzchildnode.north)},
	                    inp/.style={rectangle, draw, align=center},          
	                    ninp/.style={rectangle, draw, dashed, align=center}] 
		\node[inp]{systems\\(\autoref{sec:abstSys})}
			child{ node[inp]{functional systems\\(\autoref{sec:funSys})}
				child{ node[inp]{monotone systems\\(\autoref{sec:monotoneSys})} }
				child{ node[inp]{continuous systems\\(\autoref{sec:continuousSys})}
					child{ node[inp]{Kahn networks\\(\autoref{sec:Kahn})} } }
				child{ node[inp]{causal systems\\(\autoref{sec:causalSys})} } };
	\end{tikzpicture}
	\caption{The hierarchy of the systems treated in this paper, where systems at lower levels are special cases of their parents.}
	\label{fig:syslevels}
\end{figure}
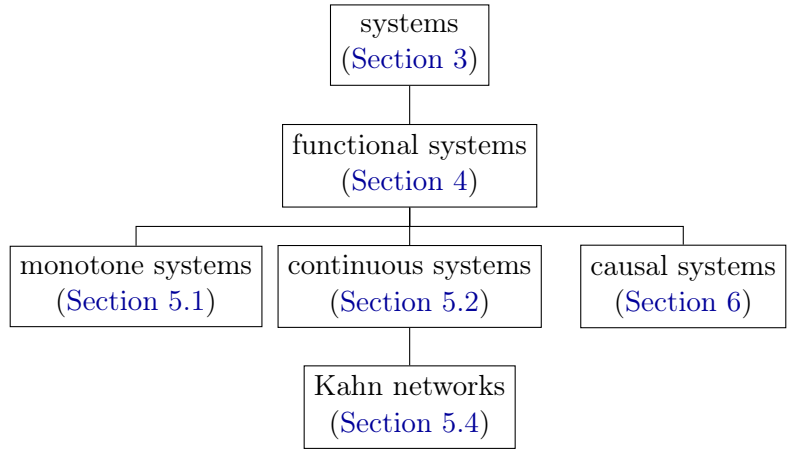

\paragraph{Abstract system algebras and composition-order invariance\ifelsarticle\else.\fi}
At the highest level of abstraction, we do not specify what systems are, but only postulate two operations~$\parac$ and $\connectFunc$, depicted in \autoref{fig:composition}; the former for combining two systems in parallel as in \autoref{fig:parallel-composition} and the latter for connecting interfaces as in \autoref{fig:connection-composed}. Using these operations one can build ``graphs'' such as those depicted in \autoref{fig:motivation}, by first taking the systems in parallel and then connecting the interfaces. We call a set of systems together with these operations, and the specification which interfaces a system has and which of them can be connected, a \emph{system algebra}, see \autoref{sec:abstSys}.

\begin{figure}[htb]
	\centering
	\begin{subfigure}[b]{0.42\textwidth}
		\centering
		\begin{tikzpicture}[system/.style={rectangle,draw,minimum width=16mm,minimum height=1cm}]
			\draw node (sys1) [system] {$\abstSys{s}_1$};
			\draw node (sys3) [system,right of=sys1,node distance=20mm] {$\abstSys{s}_2$};
	
			\draw ($(sys1.north) - (6mm,0)$) -- ($(sys1.north) - (6mm,-3mm)$);
			\draw (sys1.north) -- ($(sys1.north) + (0,3mm)$);
			\draw ($(sys1.north) + (6mm,0)$) -- ($(sys1.north) + (6mm,3mm)$);
			\draw ($(sys1.south) - (3mm,0)$) -- ($(sys1.south) - (3mm,3mm)$);
			\draw ($(sys1.south) + (3mm,0)$) -- ($(sys1.south) + (3mm,-3mm)$);

			\draw ($(sys1.north) - (6mm,-5mm)$) node {$1$};
			\draw ($(sys1.north) + (0,5mm)$) node {$2$};
			\draw ($(sys1.north) + (6mm,5mm)$) node {$3$};
			\draw ($(sys1.south) - (3mm,5mm)$) node {$4$};
			\draw ($(sys1.south) + (3mm,-5mm)$) node {$5$};
	
			\draw ($(sys3.north) - (3mm,0)$) -- ($(sys3.north) - (3mm,-3mm)$);
			\draw ($(sys3.north) + (3mm,0)$) -- ($(sys3.north) + (3mm,3mm)$);
			\draw ($(sys3.south) - (6mm,0)$) -- ($(sys3.south) - (6mm,3mm)$);
			\draw (sys3.south) -- ($(sys3.south) + (0,-3mm)$);
			\draw ($(sys3.south) + (6mm,0)$) -- ($(sys3.south) + (6mm,-3mm)$);

			\draw ($(sys3.north) + (-3mm,5mm)$) node {$A$};
			\draw ($(sys3.north) + (3mm,5mm)$) node {$B$};
			\draw ($(sys3.south) - (6mm,5mm)$) node {$C$};
			\draw ($(sys3.south) - (0cm,5mm)$) node {$D$};
			\draw ($(sys3.south) + (6mm,-5mm)$) node {$E$};
	
			\draw [dashed] (-0.9cm,-6mm) rectangle (2.9cm,6mm);
		\end{tikzpicture}
		\caption{The system~$\abstSys{s}_1 \parac \abstSys{s}_2$ obtained as the parallel composition of $\abstSys{s}_1$ and $\abstSys{s}_2$.}
		\label{fig:parallel-composition}
	\end{subfigure}
	\qquad
	\begin{subfigure}[b]{0.42\textwidth}
		\centering
		\begin{tikzpicture}[system/.style={rectangle,draw,minimum width=3cm,minimum height=1cm}]
			\path[use as bounding box] (-43pt, -36.02286pt) rectangle (43pt, 36.02286pt);
			
			\draw node (sys1) [system] {$\abstSys{s}_1$};
	
			\draw ($(sys1.north) - (1cm,0)$) -- ($(sys1.north) - (1cm,-3mm)$);
			\draw (sys1.north) -- ($(sys1.north) + (0,3mm)$);
			\draw ($(sys1.north) + (1cm,0)$) -- ($(sys1.north) + (1cm,3mm)$);
			\draw ($(sys1.south) - (5mm,0)$) -- ($(sys1.south) - (5mm,3mm)$);
			\draw ($(sys1.south) + (5mm,0)$) -- ($(sys1.south) + (5mm,-3mm)$);
	
			\draw ($(sys1.north) - (1cm,-5mm)$) node {$1$};
			\draw ($(sys1.north) + (0,5mm)$) node {$2$};
			\draw ($(sys1.north) + (13mm,5mm)$) node {$3$};
			\draw ($(sys1.south) - (8mm,5mm)$) node {$4$};
			\draw ($(sys1.south) + (5mm,-5mm)$) node {$5$};

			\draw [dashed] ($(sys1.south) + (-5mm,-3mm)$) .. controls +(0,-2) and +(0,2) .. ($(sys1.north) + (1cm,3mm)$);
		\end{tikzpicture}
		\caption{The system~$\connect{3}{4}(\abstSys{s}_1)$ obtained by connecting the interfaces $3$ and $4$ of system $\abstSys{s}_1$.}
		\label{fig:connection-composed}
	\end{subfigure}
	\caption{Composition operations in a system algebra.}
	\label{fig:composition}
\end{figure}
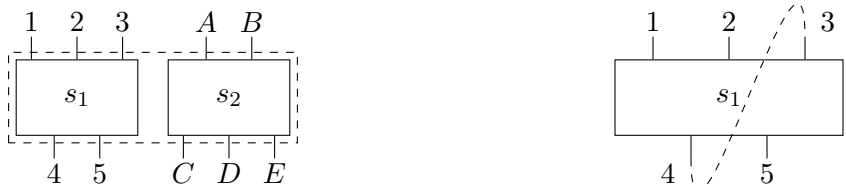

A natural property of such an algebra that can be specified at this level of abstraction is \emph{composition-order invariance}, that is, a composed system is completely described by its ``graph'', and the order in which the operations are applied to build the graph does not matter. This property is not only very natural, but also necessary for many applications. For example, the impossibility proof for broadcast sketched above relies on it since otherwise, rearranging and subdividing systems as in \autoref{fig:bc-proof} would not be allowed. To illustrate this, we formalize the systems occurring in that figure in an abstract system algebra and show how composition-order invariance appears in the proof in \autoref{sec:bcImp}. Composition-order invariance was also used by Maurer and Renner to prove the composition theorem of constructive cryptography \cite{MR11}. While appearing natural and innocent, examples throughout our paper indicate that composition-order invariance is actually a nontrivial property and requires a proof.

\paragraph{Functional system algebras\ifelsarticle\else.\fi}
An important type of system in computer science takes inputs and produces outputs depending on these inputs. The behavior of such a system can be fully described by a function mapping inputs to outputs. We define a type of system algebra, called \emph{functional system algebra}, in \autoref{sec:funSysDefs} where the systems have input interfaces and output interfaces and correspond to such functions. Connecting an input interface to an output interface is understood as setting the value input at the former to be equal to the value output at the latter. Determining the resulting system thus involves finding a fixed point of the underlying function; if multiple fixed points exist, the system algebra has to specify which one to select. By appropriately choosing the domains and functions, various types of systems can be modeled in this way, including interactive systems that take many inputs in different rounds and systems that depend on time.

We prove several basic results at this level of abstraction, i.e., without specifying which functions are considered in \autoref{sec:funSysProps}. For example, we show that not all functional system algebras are composition-order invariant, but if there are always unique fixed points for interface connections and connections can be reordered, composition-order invariance holds.

While this paper focuses on the deterministic case, we point out that functional systems can be used as a basis to model probabilistic systems. For example, one can consider systems that take randomness as an explicit input at a dedicated interface, one can include random variables in the domains of the functions, or one can consider probability distributions over deterministic systems. Systematically understanding probabilistic systems in this way is future work.

\paragraph{Instantiations of functional system algebras\ifelsarticle\else.\fi}
To instantiate the concept of a functional system algebra, we need to specify the domains of the functions and the set of functions to consider. To be able to define the interface connection, we have to ensure that all functions have the required fixed points. One way to guarantee fixed points that is well-studied in mathematics, especially domain theory, is to equip the domain with a partial order such that all chains have a supremum and to consider \emph{monotone functions}. A related concept are \emph{continuous functions}, which are defined as preserving these suprema. In both cases the functions have a least fixed point. While continuity is a stronger requirement than monotonicity, a slightly weaker assumption on the domains is sufficient to guarantee least fixed points. We show in \autoref{sec:monotoneSys} and \autoref{sec:continuousSys} that if least fixed points are chosen for interface connections, monotone and continuous functions form a functional system algebra, respectively. Under the additional assumption that nonempty chains have an infimum, we show in \autoref{sec:monoCompInv} that these system algebras are composition-order invariant.

Monotone and continuous functions are not only a mathematical convenience to obtain fixed points, but they also encompass as a special case an intuitive and useful model known as \emph{Kahn networks}~\cite{Kahn74} (or \emph{Kahn process networks}; in Kahn's paper only defined for continuous functions, but monotone functions can be used as well), which we consider in \autoref{sec:Kahn}. Kahn networks have been developed to provide a semantics for parallel programming languages~\cite{Kahn74}, but they have also been used in other contexts, including embedded systems~\cite{SZTKD04} and signal processing~\cite{LP95}. The domains of the functions there consist of sequences of values and the partial order is defined to be the initial segment (or prefix) relation. An interpretation of a function is that it maps input histories to output histories. Such functions therefore correspond to interactive systems that take one input after the other at each input interface and produce outputs depending on the inputs. Monotonicity means that additional inputs can only yield additional outputs, an output cannot be ``taken back''. Even though it appears to be a very natural question whether the order in which interfaces of Kahn networks are connected matters, we are not aware of any result in this direction. Our proof of composition-order invariance, which indeed turned out to be nontrivial, therefore also provides new insights into this well-studied model.

In \autoref{sec:causalSys}, we finally provide an instantiation of functional system algebras consisting of \emph{causal systems} in which inputs to the system can only influence ``later'' outputs. We formalize this by considering a partially ordered set (where ``less than'' can be interpreted as ``before'') and letting the domains of the functions consist of subsets thereof. As an example, consider the partially ordered set containing pairs~$(\val, \tim)$, which can be interpreted as the value~$\val$ being input (or output) at time~$\tim$, where the order is naturally defined as the one induced by the second component. The domains are then sets of such pairs. This allows us, as for Kahn networks, to model systems that take several inputs at each input interface and produce several outputs. We define causality for such systems and prove that the corresponding functions have unique fixed points. Therefore, we obtain a composition-order invariant functional system algebra. This system algebra can in particular be used to model and analyze systems that depend on time, such as clocks and channels with certain delays.

\subsection{Related Work}\label{sec:related}
There exists a large body of work on modeling certain types of systems mathematically. Some models can be understood as special cases in our theory, but also very general theories and models that do not fit in our theory exist. The work we are aware of, however, only captures partial aspects of our theory. We now describe some of this work and compare it to ours.

\paragraph{Abstract models\ifelsarticle\else.\fi}
The abstract concept of a system algebra in which complex systems are built from components has been informally described in the context of cryptography by Maurer and Renner~\cite{MR11}. They also, again informally, introduced composition-order independence, which corresponds to our composition-order invariance. We provide in this paper a formalization that matches their requirements.

Hardy has developed an abstract theory, in which composition-order invariance (there called order independence) plays an important role \cite{Har13}. That work, however, focuses on physical systems. Lee and Sangiovanni-Vincentelli~\cite{LeeSan98} also introduce an abstract system model, but it is specific to systems that consider some form of time and does not follow an algebraic approach.

Closely related to our abstract system algebras are block algebras as introduced by de Alfaro and Henzinger in the context of interface theories~\cite{dAH01a}. Our systems and interfaces are there called blocks and ports, respectively, and they also define parallel composition and port connection operations. A major difference compared to our system algebras is that port connections do not hide the connected ports. Moreover, while de Alfaro and Henzinger require the parallel composition to be commutative and associative, they do not define a notion that corresponds to our composition-order invariance, i.e., their port connections not necessarily commute with other port connections and parallel composition.

Process algebras allow the modeling of communicating concurrent processes \cite{BBR09}, which correspond to our systems. Examples of process algebras include Milner's Calculus of Communicating Systems (CCS) \cite{Mil80} and Hoare's Communicating Sequential Processes (CSP) \cite{Hoa78}. These theories are similar to ours in that they consider certain operations, including parallel composition, on processes and postulate axioms they need to satisfy. These axioms also include guarantees similar to our composition-order invariance. In contrast to our theory, CCS does not have an explicit operation for connecting interfaces, there called ports. Rather, ports with matching labels are implicitly connected when systems are taken in parallel. As in the work by de Alfaro and Henzinger~\cite{dAH01a} discussed above, ports are not hidden by this connection; hiding them is an explicit operation in CCS. In CSP, processes communicate via special input and output commands. As opposed to our abstract level, these process algebras differentiate between inputs and outputs and allow to specify the behavior of processes. Thus, with respect to the level of abstraction, they belong somewhere between our abstract and functional levels, or even below.

Milner's Flowgraphs~\cite{Mil79} model composed systems as generalized graphs, in which the nodes correspond to subsystems. This is considerably different from our theory since we abstract away all internal details of a system, in particular the subsystems that compose a system.


\paragraph{Functional models\ifelsarticle\else.\fi}
A line of work on system models based on functions has been initiated by Kahn's seminal paper~\cite{Kahn74} on networks of autonomous computing systems. These systems may be sensitive to the order in which messages arrive on one interface, but they are oblivious to the relative order of incoming messages on different interfaces. He shows that least fixed points exist, based on earlier work by Tarski~\cite{Tar55}, Scott~\cite{Scott70}, and Milner~\cite{Milner73},\footnote{Other sources attribute the original theorem to Kleene or Knaster, see~\cite{LaNgSo82}.} and therefore connecting systems is well-defined.
Tackmann~\cite{Tac14} considered the case where systems are fully oblivious to the order of their incoming messages, which can be seen as a special case of Kahn networks where each interface contains at most one message.
Micciancio and Tessaro~\cite{MicTes13} start from the same type as Kahn but extend it to tolerate certain types of order-dependent behavior within complex systems.

\paragraph{Timed models\ifelsarticle\else.\fi}
Several works have defined \emph{causal} system models.
Lee and Sangiovanni-Vincentelli~\cite{LeeSan98} define \emph{delta causality}, which intuitively requires that each output must be provoked by an input that occurred at least a $\delta$-difference earlier.
They show that fixed points exist, based on Banach's theorem.
Cataldo et al.~\cite{CLLMZ06} generalize this to a notion of ``superdense'' time where multiple events may occur simultaneously.
Portmann et al.~\cite{PRMMT16}, in the quantum scenario, describe a type of strict causality based on a \emph{causality function} that can be seen as a generalization of delta causality.
Naundorf~\cite{Nau00} considers \emph{strict} causality without any minimal time distance, and proves that fixed points still exist. Matsikoudis and Lee~\cite{ML15} then show a \emph{constructive} fixed point theorem for the same notion, which they refer to as \emph{strictly contracting}. They show that it is implied by a more natural notion of (strict) causality where outputs can be influenced only by inputs that occur strictly earlier, under the assumption that the ordering of inputs is well-founded\footnote{A partial order on a set~$T$ is \emph{well-founded} if every nonempty subset of~$T$ has one or more minimal elements.}.
We show in \autoref{app:equivalence-causal} that the strict causality notion of~\cite{ML15} is essentially equivalent to the definition we introduce in this work.

Except for the work of Portmann et al., none of the previously mentioned definitions of causal functions explicitly capture systems with multiple interfaces as the work of Kahn~\cite{Kahn74} or our work. Also, the mentioned papers investigating causal functions do not define how to connect systems such that one obtains a system of the same type, and therefore they do not provide a system algebra as we do.
The model by Portmann et al.~\cite{PRMMT16} captures quantum information-processing systems and can be seen as a generalization of our causal systems. Restricting that model to classical, deterministic inputs and outputs yields, however, a more complex and less general model than our causal systems. For example, the causality definition in that paper is more restrictive and in contrast to our causal systems, the systems there are not allowed to produce infinitely many outputs in finite time.

Partial orders have also been used to model causal or temporal dependencies, for example in Pratt's theory of partially ordered multisets (pomsets) \cite{Pra86,Pra85}. As in our causal systems, there is a partial order on the possible events. In contrast to our causal systems, however, systems in that theory are not necessarily functions and can exhibit nondeterminism. Our model is therefore more specific, which allows a more accessible presentation.

\paragraph{Stateful models\ifelsarticle\else.\fi}
Several models of systems have been proposed that model the systems as objects that explicitly contain state. I/O automata initially discussed by Lynch and Tuttle~\cite{LynTut89} and interface automata by de Alfaro and Henzinger~\cite{dAH01} enhance stateful automata by interactive communication.
Timed automata by Alur and Dill~\cite{AD94} and timed I/O automata by Kaynar et al.~\cite{KLSV10} extend them to include a notion of time as our causal systems.
Interactive Turing machines basically equip Turing machines with additional tapes that they share with other machines and have been used widely in (complexity-theoretic) cryptography \cite{Goldre01,Can01}.
All these models are substantially different from ours since we want to hide all internal details of systems, including their state.

\section{Preliminaries} \label{sec:preliminaries}
\subsection{Functions and Notation for Sets and Tuples}
A \emph{function} $f \colon X \to Y$ is a subset of $X \times Y$ such that for every $x \in X$, there is exactly one $y \in Y$ such that $(x,y) \in f$, where we will usually write $f(x) = y$ instead of $(x,y) \in f$. For two sets~$X$ and $Y$, the set of all functions $X \to Y$ is denoted by $Y^X$. A \emph{partial function} $f \colon X \to Y$ is a function $X' \to Y$ for some $X' \subseteq X$. For a subset~$S \subseteq X$, $f(S) \coloneqq \{f(s) \mid s \in S\}$ denotes the \emph{image of $S$ under $f$}. For $X' \subseteq X$, we define the \emph{restriction of $f$ to $X'$} as $f|_{X'} \coloneqq \{(x,y) \in f \mid x \in X'\} \in Y^{X'}$. Note that for $z \notin X$ and $y \in Y$, we have $f \cup \{(z, y)\} \in Y^{X \cup \{z\}}$. An element in $Y^X$ can equivalently be interpreted as a tuple of elements in $Y$ indexed by elements in $X$. In case we interpret a function as a tuple, we usually use a boldface symbol to denote it. For a tuple~$\tuple{x} \in X^I$ with $\tuple{x}(i) = x_i$ for all $i \in I$, we also write $\tuple{x} = (x_i)_{i \in I}$ and if $I = \{i_1, \ldots, i_n\}$, we write $\tuple{x} = (x_{i_1}, \ldots, x_{i_n})$.
The \emph{symmetric difference} of two sets $X$ and $Y$ is defined as
\begin{equation*}
	X \symdif Y \coloneqq (X \setminus Y) \cup (Y \setminus X).
\end{equation*}
Finally, we denote the \emph{power set} of a set~$X$ by $\mathcal{P}(X) \coloneqq \{S \mid S \subseteq X\}$.

\subsection{Order Relations}
We first recall some basic definitions about relations.
\begin{definition}
	Let $X$ be a set. A \emph{(binary) relation on $X$} is a subset of $X \times X$. We write $x \mathrel{R} y$ for $(x, y) \in R$. A relation $R \subseteq X \times X$ is called \emph{reflexive} if $x \mathrel{R} x$ for all $x \in X$. It is called \emph{symmetric} if $x \mathrel{R} y \ \Longrightarrow y \mathrel{R} x$ for all $x,y \in X$ and \emph{antisymmetric} if $(x \mathrel{R} y \ \wedge \ y \mathrel{R} x) \ \Longrightarrow \ x = y$. A relation $R \subseteq X \times X$ is called \emph{transitive} if $(x \mathrel{R} y \ \wedge \ y \mathrel{R} z) \ \Longrightarrow  \ x \mathrel{R} z$ for all $x,y,z \in X$.
\end{definition}

For sets $X$ and $I$ and a binary relation~$\mathrel{R}$ on $X$, we define the relation~$\mathrel{R}$ on $X^I$ as the componentwise relation, i.e., for $\tuple{x}, \tuple{y} \in X^I$,
\begin{equation*}
	\tuple{x} \mathrel{R} \tuple{y} \ \Longleftrightarrow \ \forall i \in I \ \bigl(\tuple{x}(i) \mathrel{R} \tuple{y}(i) \bigr).
\end{equation*}
\begin{definition}
	A \emph{partial order} on $X$ is a binary relation on~$X$ that is reflexive, antisymmetric, and transitive. A \emph{partially ordered set (poset)}~$(X, \preceq)$ is a set~$X$ together with a partial order~$\preceq$ on~$X$.
\end{definition}

We will typically denote partial orders by $\leq$, $\preceq$, or $\sqsubseteq$, and define the relation $<$ by $x < y \ :\Longleftrightarrow \ x \leq y \ \wedge \ x \neq y$, and analogously $\prec$ and $\sqsubset$.

\begin{definition}
	Let $(X, \preceq)$ be a poset. Two elements $x, y \in X$ are \emph{comparable} if $x \preceq y$ or $y \preceq x$, and \emph{incomparable} otherwise. If all $x,y \in X$ are comparable, $X$ is \emph{totally ordered}. A totally ordered subset of a poset is called a \emph{chain}.
\end{definition}

\begin{definition}
	Let $(X, \preceq)$ be a poset. An element $x \in X$ is the \emph{least element} of~$X$ if $x \preceq y$ for all $y \in X$. Similarly, $x \in X$ is the \emph{greatest element} of~$X$ if $y \preceq x$ for all $y \in X$. The least element and greatest element of $X$ are denoted by $\min X$ and $\max X$, respectively. An element~$x \in X$ is a \emph{minimal element} of~$X$ if there is no $y \in X$ with $y \prec x$, and $x$ is a \emph{maximal element} if there is no $y \in X$, $y \succ x$. For a subset $S \subseteq X$, $x \in X$ is a \emph{lower bound} of~$S$ if $x \preceq s$ for all $s \in S$ and an \emph{upper bound} of~$S$ if $x \succeq s$ for all $s \in S$. If the set of lower bounds has a greatest element, it is called the \emph{infimum} of~$S$, denoted $\inf S$; the \emph{supremum} of~$S$, denoted~$\sup S$, is the least upper bound of~$S$.
\end{definition}

\begin{definition}
	A poset~$(X, \preceq)$ is \emph{well-ordered} if every nonempty subset of $X$ has a least element.
\end{definition}
Note that every well-ordered poset~$(X, \preceq)$ is totally ordered because $\{x, y\}$ has a least element.

\begin{definition}
	Let $(X, \leq)$ and $(Y, \preceq)$ be posets. An \emph{order isomorphism} is a bijection~$\psi \colon X \to Y$ such that $\psi(x_1) \preceq \psi(x_2) \Longleftrightarrow x_1 \leq x_2$ for all $x_1, x_2 \in X$. The posets~$(X, \leq)$ and $(Y, \preceq)$ are called \emph{order isomorphic} if such order isomorphism exists.
\end{definition}

\subsection{Ordinals and Transfinite Induction}
We briefly recall some basics of set theory, following Halbeisen \cite{Hal12} and Jech \cite{Jec03}. A \emph{class} is a collection of sets. More formally, a class~$C$ corresponds to a logical formula and we write $x \in C$ if $x$ satisfies that formula. Every set is a class but not all classes are sets; for example, the class of all sets and the class of all ordinals are not sets. A class that is not a set is called a \emph{proper class}.
\begin{definition}
	An \emph{ordinal} is a set~$\alpha$ such that
	 $\forall x \in \alpha \ (x \subseteq \alpha)$,
	 $\forall x_1, x_2 \in \alpha \ (x_1 \in x_2 \ \vee \ x_1 = x_2 \ \vee \ x_2 \in x_1)$, and
	 $\forall S \subseteq \alpha, S \neq \emptyset \ \exists x \in S \ \forall y \in x \ (y \notin S)$.
\end{definition}

For ordinals~$\alpha, \beta$ with $\alpha \in \beta$, we write $\alpha < \beta$. For all ordinals~$\alpha \neq \beta$, we have either $\alpha < \beta$ or $\beta < \alpha$ (but not both) and $\alpha \nless \alpha$. It can be shown that every nonempty class of ordinals has a least element (according to the relation~$<$) \cite[Theorem~3.12]{Hal12}. For an ordinal~$\alpha$, we define
\begin{equation*}
	\alpha + 1 \coloneqq \alpha \cup \{ \alpha \}.
\end{equation*}
We have that $\alpha + 1$ is the least ordinal greater than~$\alpha$ \cite[Corollary~3.13]{Hal12}. An ordinal~$\alpha$ is called \emph{successor ordinal} if $\alpha = \beta + 1$ for some ordinal~$\beta$. A \emph{limit ordinal} is an ordinal that is not a successor ordinal.

Every well-ordered set~$X$ is order isomorphic to exactly one ordinal~\cite[Theorem~2.12]{Jec03}. This ordinal is called the \emph{order type} of~$X$.

We define the natural numbers as $0 \coloneqq \emptyset$, $1 \coloneqq 0 + 1 = \{ 0 \}$, $2 \coloneqq 1 + 1 = \{ 0, 1 \}$, and so on. That is, a number~$n \in \N$ is the set of all numbers less than~$n$. The set of natural numbers~$\N$ is also an ordinal, denoted by $\omega$. Note that $0$ and $\omega = \N$ are limit ordinals and all nonzero natural numbers are successor ordinals. A method for proving a statement about all ordinals is via the transfinite induction theorem \cite[Theorem~2.14]{Jec03}.
\begin{theorem}[Transfinite Induction]\label{thm:transfInd}
	Let $C$ be a class of ordinals such that
	\begin{enumerate}[label=(\roman*), itemsep=0pt]
		\item $0 \in C$,
		\item if $\alpha \in C$, then $\alpha + 1 \in C$, and
		\item if $\alpha$ is a nonzero limit ordinal and $\beta \in C$ for all $\beta < \alpha$, then $\alpha \in C$.
	\end{enumerate}
	Then, $C$ is the class of all ordinals.
\end{theorem}

The following lemma will later be useful.
\begin{lemma}\label{lem:chainFixHartogs}
	Let $(X, \preceq)$ be a poset and assume that for every ordinal~$\alpha$, there is an $x_{\alpha} \in X$ such that $x_{\beta} \preceq x_{\alpha}$ for all $\beta \leq \alpha$.\footnote{More formally, we consider a \emph{class function}~$F$ from the class of all ordinals to~$X$, and write $x_\alpha$ for $F(\alpha)$.} Then, there exists an ordinal $\hat{\alpha}$ such that $x_{\hat{\alpha}} = x_{\hat{\alpha}+1}$.
\end{lemma}
\begin{proof}
	A result by Hartogs implies that for any set~$S$, there exists an ordinal~$\eta(S)$ such that there is no injective function~$\eta(S) \to S$ \cite[Lemma~7.1]{Joh87}. Hence, there exist ordinals~$\hat{\alpha} < \beta$ such that $x_{\hat{\alpha}} = x_{\beta}$, since otherwise the function $\eta(X) \to X, \gamma \mapsto x_{\gamma}$ would be injective. Since $\hat{\alpha} \leq \hat{\alpha}+1 \leq \beta$, we have $x_{\hat{\alpha}} \preceq x_{\hat{\alpha} + 1} \preceq x_{\beta}$. Thus, $x_{\hat{\alpha}} = x_{\beta}$ implies that $x_{\hat{\alpha}} = x_{\hat{\alpha}+1}$.
\end{proof}

\subsection{Complete Posets and Fixed Points of Monotone and Continuous Functions}
A natural requirement for functions between posets is that they preserve order. Order-preserving functions are also called monotone and are defined below.
\begin{definition}
	Let $(X, \leq)$ and $(Y, \preceq)$ be posets. A function~$f \colon X \to Y$ is \emph{monotone} if
	\begin{equation*}
		\forall x_1, x_2 \in X \ \bigl( x_1 \leq x_2 \ \rightarrow \ f(x_1) \preceq f(x_2) \bigr).
	\end{equation*}
\end{definition}

Note that a monotone bijection is not necessarily an order isomorphism: For $X = \{x_0, x_1\}$ with incomparable $x_0$ and $x_1$ and $Y = \{0,1\}$ with $0 \preceq 1$, the bijection $f \colon X \to Y, x_i \mapsto i$ is trivially monotone but not an order isomorphism.

\begin{definition}
	Let $X$ be a set and $f \colon X \to X$ be a function. Then, $x \in X$ is called a \emph{fixed point} of $f$ if $f(x) = x$.
\end{definition}

\begin{definition}
	A \emph{complete partially ordered set (CPO)} is a poset in which every chain has a supremum.
\end{definition}

Note that the empty set is a chain and every element is an upper bound of~$\emptyset$. Therefore, a CPO contains a least element.

\begin{theorem}[{\cite[Theorem~2.5]{Esi09}}]\label{thm:monotoneFP}
	Let $(X, \preceq)$ be a CPO and $f \colon X \to X$ be monotone. Then, $f$ has a least fixed point, which equals~$x_{\hat{\alpha}}$ for some ordinal~$\hat{\alpha}$, where $x_0 = \min(X)$, $x_{\alpha+1} = f(x_{\alpha})$ for any ordinal $\alpha$, and for nonzero limit ordinals $\alpha$, $x_{\alpha} = \sup \{ x_{\beta} \mid \beta < \alpha\}$. We further have $x_{\alpha} \leq x_{\beta}$ for $\alpha \leq \beta$.
\end{theorem}

The above theorem is constructive in the sense that it not only guarantees the existence of a least fixed point, but also provides a procedure to find it. However, this procedure might only terminate after transfinitely many steps. The situation improves if the function is not only monotone but also continuous in the sense that it preserves suprema. In this case, a weaker requirement on the domain of the function is sufficient, namely only chains that correspond to infinite sequences need to have a supremum.

\begin{definition}
	Let $(X, \preceq)$ be a poset. An \emph{$\omega$-chain} in $X$ is a sequence~$(x_n)_{n \in \omega}$ such that $x_i \preceq x_j$ for all $i \leq j$. We say $(X, \preceq)$ is an \emph{$\omega$-chain complete partially ordered set ($\omega$-CPO)} if it has a least element and every $\omega$-chain has a supremum.
\end{definition}

\begin{definition}
	Let $(X, \leq)$ and $(Y, \preceq)$ be $\omega$-CPOs. A function $f \colon X \to Y$ is \emph{$\omega$-continuous} if for every $\omega$-chain $C = (x_n)_{n \in \omega}$ in~$X$, $\sup f(C)$ exists and $f(\sup C) = \sup f(C)$.
\end{definition}


The next lemma shows that $\omega$-continuity implies monotonicity (the converse is not true).

\begin{lemma}\label{lem:contImplMon}
	Let $(X, \leq)$ and $(Y, \preceq)$ be $\omega$-CPOs and $f \colon X \to Y$ an $\omega$-continuous function. Then, $f$ is monotone.
\end{lemma}
\begin{proof}
	Let $x_1, x_2 \in X$ such that $x_1 \leq x_2$. Then, $\{x_1, x_2, x_2, x_2, \ldots\}$ is an $\omega$-chain. Therefore,
	\begin{equation*}
		f(x_1) \preceq \sup f \bigl(\{x_1, x_2, x_2, \ldots\} \bigr) = f \bigl(\sup \{x_1, x_2, x_2, \ldots\} \bigr) = f(x_2). \qedhere
	\end{equation*}
\end{proof}

\begin{theorem}[{\cite[Theorem~8.15]{DB02}}]\label{thm:continuousFP}
	Let $(X, \preceq)$ be an $\omega$-CPO and $f \colon X \to X$ be $\omega$-continuous. Then, $f$ has a least fixed point, which equals~$\sup \{x_n \mid n \in \omega\}$, where $x_0 = \min(X)$ and $x_{n+1} = f(x_n)$ for $n \in \omega$.
\end{theorem}

\begin{remark}
	In the literature, CPOs and continuity are often defined in terms of so-called directed subsets instead of chains \cite{DB02,Esi09}. For our purposes, chains are more intuitive and directly applicable in our proofs. This definitional inconsistency is not an issue since the two types of definitions have been shown to be equivalent \cite{Mar76}.
\end{remark}

\section{Abstract System Algebras}\label{sec:abstSys}
We first define system algebras at an abstract level described, but not formalized, by Maurer and Renner~\cite{MR11}, where systems are objects with interfaces via which they can be connected to other systems. Our motivation for defining system algebras at this level of abstraction is twofold: First, we thereby introduce a common language and notation that can be used for more concrete system algebras in later sections as well as in future papers. Secondly, one can already make meaningful statements and prove theorems at this level of abstraction. As examples, we provide a proof of broadcast impossibility in \autoref{sec:bcImp} and refer to the paper by Maurer and Renner~\cite{MR11}, where the composition theorem of constructive cryptography is proven based on abstract systems. Whenever a proof is possible at this level of abstraction, it is clearly preferable to proofs at lower levels due to increased generality and simplicity.

\subsection{Definitions}
A system algebra is a set of systems with certain operations that allow one to compose several systems to obtain a new system. In this way, complex systems can be decomposed into independent components. At this level of abstraction, we only specify how systems can be composed, but not what systems are or how they interact with other systems. In the same sense as the elements of an algebraic ring are abstract objects without concrete meaning, abstract systems have no particular meaning attached, beyond how they can be composed. In a concrete instantiation, systems could, e.g., communicate via discrete inputs and outputs at their interfaces or via analog signals, et cetera. We define two operations; an operation~$\parac$ for taking two systems in parallel, and an operation~$\connectFunc$ for connecting two interfaces of a system. See \autoref{fig:composition} for a depiction of these two operations. Several systems can be connected by first taking them in parallel and then connecting their interfaces.
A similar definition has been given by Tackmann~\cite[Definition~3.4]{Tac14}.
\begin{definition}\label{def:sysAlg}
	Let $\inSet$ be a set (the set of interface labels). A \emph{$\inSet$-system algebra}~$(\sysSpace, \intSetFunct, \parac, \connectSet, \connectFunc)$ consists of a set~$\sysSpace$ (the set of systems), a function $\intSetFunct \colon \sysSpace \to \mathcal{P}(\inSet)$ (assigning to each system its set of interface labels), a partial function~$\parac \colon \sysSpace \times \sysSpace \to \sysSpace$ (the \emph{parallel composition} operation), a function~$\connectSet \colon \sysSpace \to \mathcal{P}\bigl(\bigl\{ \{\inIn, \inIn'\} \mid \inIn, \inIn' \in \inSet \bigr\} \bigr)$ (specifying for each system the set of interface-label pairs that can be connected), and a partial function~$\connectFunc \colon \bigl\{ \{\inIn, \inIn'\} \mid \inIn, \inIn' \in \inSet \bigr\} \times \sysSpace \to \sysSpace$ (the \emph{interface connection} operation), such that
	\begin{itemize}
		\item for all $\abstSys{s} \in \sysSpace$, $\intSetFunct(\abstSys{s})$ is finite and for all $\{\inIn, \inIn'\} \in \connectSet(\abstSys{s})$, we have $\inIn, \inIn' \in \intSetFunct(\abstSys{s})$,
		\item for $\abstSys{s}_1, \abstSys{s}_2 \in \sysSpace$, $\mathop{\parac}(\abstSys{s}_1, \abstSys{s}_2)$, denoted $\abstSys{s}_1 \parac \abstSys{s}_2$, is defined if and only if $\intSetFunct(\abstSys{s}_1) \cap \intSetFunct(\abstSys{s}_2) = \emptyset$, and in this case, $\intSetFunct(\abstSys{s}_1 \parac \abstSys{s}_2) = \intSetFunct(\abstSys{s}_1) \cup \intSetFunct(\abstSys{s}_2)$ and for $j \in \{1,2\}$ and for all $\inIn, \inIn' \in \intSetFunct(\abstSys{s}_j)$, we have $\{\inIn, \inIn'\} \in \connectSet(\abstSys{s}_1 \parac \abstSys{s}_2) \Longleftrightarrow \{\inIn, \inIn'\} \in \connectSet(\abstSys{s}_j)$, and
		\item for $\inIn, \inIn' \in \inSet$ and $\abstSys{s} \in \sysSpace$, $\connectFunc(\{\inIn, \inIn'\}, \abstSys{s})$, denoted by $\connect{\inIn}{\inIn'}(\abstSys{s})$, is defined if and only if $\{\inIn, \inIn'\} \in \connectSet(\abstSys{s})$, and in this case, $\intSetFunct\bigl(\connect{\inIn}{\inIn'}(\abstSys{s}) \bigr) = \intSetFunct(\abstSys{s}) \setminus \{\inIn, \inIn'\}$.
	\end{itemize}
\end{definition}

We will usually identify a system algebra with the set of systems~$\sysSpace$ and use the same symbols $\intSetFunct$, $\parac$, $\connectSet$, and $\connectFunc$ for different system algebras. The parallel composition of two system is only allowed if they have disjoint interface sets. This means in particular that one cannot consider the parallel composition of a system with itself. One can imagine that each system exists only once and therefore cannot be used twice within another system. This is not an issue because~$\sysSpace$ can contain many ``copies'' of a system with different interface labels, and different systems can have different interface labels. One could also introduce an interface-renaming operation, but we will not formalize this because it is not needed here.
The set~$\connectSet(\abstSys{s})$ determines which interfaces of a system~$\abstSys{s}$ are compatible, i.e., can be connected to each other. It might or might not be possible to connect an interface to itself. Figuratively speaking, one could imagine that interfaces come with different types of plugs and one can only connect interfaces with matching plugs, where $\connectSet(\abstSys{s})$ contains all unordered pairs of matching interfaces. For example, we will later consider system algebras with separate interfaces for inputs and outputs, where one can only connect input interfaces to output interfaces, but not two interfaces of the same type. Since the connection operation is defined for unordered pairs of interfaces, one always connects two interfaces to each other, without a direction. The condition on~$\connectSet$ for the parallel composition ensures that if one can connect two interfaces of a system, one can still do so after taking another system in parallel (i.e., $\{\inIn, \inIn'\} \in \connectSet(\abstSys{s}_j) \Longrightarrow \{\inIn, \inIn'\} \in \connectSet(\abstSys{s}_1 \parac \abstSys{s}_2)$ for $j \in \{1,2\}$), and additional connections are only created between the two systems, not for a single system (i.e., for $j \in \{1,2\}$ and $\inIn, \inIn' \in \intSetFunct(\abstSys{s}_j)$,  $\{\inIn, \inIn'\} \in \connectSet(\abstSys{s}_1 \parac \abstSys{s}_2) \Longrightarrow \{\inIn, \inIn'\} \in \connectSet(\abstSys{s}_j)$). The intuition behind this condition is that the two systems are independent and do not influence what is possible for the other system. After connecting interfaces, however, it is possible that connections that were allowed before become disallowed. For example, one might want to consider a system algebra in which one cannot create cycles by connecting systems, e.g., when modeling systems that correspond to Boolean circuits. Then, certain connections are only allowed as long as other interfaces are not connected.

We restrict ourselves to systems with finitely many interfaces because we are only interested in systems that are composed of finitely many components. Therefore, we can define parallel composition as a binary operation and interface connection for a single pair of interfaces, whereas in general, one would define the parallel composition of potentially infinitely many systems and the connection of potentially infinitely many pairs of interfaces. In our simplified setting, repeated applications of the binary parallel composition and the connection of two interfaces are sufficient.

An important property system algebras can have is \emph{composition-order invariance} (called composition-order independence by Maurer and Renner~\cite{MR11}). Loosely speaking, it guarantees that a system that is composed of several systems is independent of the order in which they have been composed. Put differently, a figure in which several systems are connected by lines uniquely determines the overall system; the order in which the figure was drawn is irrelevant.

\begin{definition}
	For a $\inSet$-system algebra~$(\sysSpace, \intSetFunct, \parac, \connectSet, \connectFunc)$, we say $\connectSet$ \emph{permits reordering} if for all $\abstSys{s} \in \sysSpace$, $\{\inIn, \inIn'\} \in \connectSet(\abstSys{s})$, and $\{j, j'\} \in \connectSet \bigl(\connect{\inIn}{\inIn'}(\abstSys{s}) \bigr)$, we have $\{j, j'\} \in \connectSet(\abstSys{s})$ and $\{\inIn, \inIn'\} \in \connectSet \bigl(\connect{j}{j'}(\abstSys{s}) \bigr)$. 
	If additionally $\connect{j}{j'}\bigl( \connect{\inIn}{\inIn'}(\abstSys{s}) \bigr) = \connect{\inIn}{\inIn'}\bigl( \connect{j}{j'}(\abstSys{s}) \bigr)$, $(\sysSpace, \intSetFunct, \parac, \connectSet, \connectFunc)$ is called \emph{connection-order invariant}.
	A connection-order invariant system algebra is called \emph{composition-order invariant} if the operation~$\parac$ is associative and commutative
	and for all $\abstSys{s}_1, \abstSys{s}_2 \in \sysSpace$ and $\{ \inIn, \inIn' \} \in \connectSet(\abstSys{s}_1)$ such that $\intSetFunct(\abstSys{s}_1) \cap \intSetFunct(\abstSys{s}_2) = \emptyset$, we have $\connect{\inIn}{\inIn'}(\abstSys{s}_1) \parac \abstSys{s}_2 = \connect{\inIn}{\inIn'}(\abstSys{s}_1 \parac \abstSys{s}_2)$.
\end{definition}

All system algebras we consider in this paper have associative and commutative parallel composition. Note however, that one can also imagine system algebras where this is not the case: Consider a set of systems that correspond to software components that are compiled together when two systems are composed in parallel. Depending on compiler optimizations, the efficiency of the resulting program might depend on the order in which components are compiled together.

\subsection{An Abstract Proof of Broadcast Impossibility}\label{sec:bcImp}
In this section, we provide a formal proof at the level of abstract systems of the impossibility of bit-broadcast for three parties with one dishonest party, as sketched in the introduction. This exemplifies how a proof that involves drawing figures can be justified at the level of abstract systems and why composition-order invariance is crucial for doing so. Since there is no notion of outputting a bit for abstract systems, one cannot directly formulate the requirements for broadcast at this level. We therefore first prove a more abstract statement and afterwards argue how this implies the impossibility for more concrete systems. We assume in the following that we have a system algebra that allows to connect systems arbitrarily, i.e., two (different) interfaces of a system can always be connected (formally, $\{\inIn, \inIn' \} \in \connectSet(\abstSys{s})$ for all $\inIn, \inIn' \in \intSetFunct(\abstSys{s})$ with $\inIn \neq \inIn'$). This is a reasonable assumption, e.g., in a setting where the systems communicate with other systems by sending messages over channels.

\begin{theorem}\label{thm:bcImp}
	Let $(\sysSpace, \intSetFunct, \parac, \connectSet, \connectFunc)$ be a composition-order invariant $\inSet$-system algebra such that $\connectSet(\abstSys{s}) \supseteq \bigl\{ \{\inIn, \inIn' \} \mid \inIn, \inIn' \in \intSetFunct(\abstSys{s}) \wedge \inIn \neq \inIn' \bigr\}$ for all $\abstSys{s} \in \sysSpace$.
	Let $\abstSys{s}_0, \abstSys{s}_1, \abstSys{r}_1, \abstSys{r}_2 \in \sysSpace$ such that $\intSetFunct(\abstSys{s}_0)$, $\intSetFunct(\abstSys{s}_1)$, $\intSetFunct(\abstSys{r}_1)$, and $\intSetFunct(\abstSys{r}_2)$ are pairwise disjoint and $s_0^1, s_0^2 \in \intSetFunct(\abstSys{s}_0)$, $s_1^1, s_1^2 \in \intSetFunct(\abstSys{s}_1)$, $r_1^1, r_1^2 \in \intSetFunct(\abstSys{r}_1)$, and $r_2^1, r_2^2 \in \intSetFunct(\abstSys{r}_2)$, where these eight interface labels are distinct. Further let
	\begin{align*}
		\mathcal{S}_0 &\coloneqq \bigl\{ \connect{r_1^2}{e^2}(\connect{s_0^2}{e^1}(\connect{s_0^1}{r_1^1}(\abstSys{s}_0 \parac \abstSys{r}_1 \parac \abstSys{e}))) \mid e \in \sysSpace \wedge \intSetFunct(\abstSys{e}) \cap \intSetFunct(\abstSys{s}_0) = \emptyset = \intSetFunct(\abstSys{e}) \cap \intSetFunct(\abstSys{r}_1) \\ &\hspace{26em} \wedge e^1, e^2 \in \intSetFunct(\abstSys{e}) \wedge e^1 \neq e^2 \bigr\}, \\
		\mathcal{S}_1 &\coloneqq \bigl\{ \connect{e^2}{r_2^2}(\connect{s_1^2}{r_2^1}(\connect{s_1^1}{e^1}(\abstSys{s}_1 \parac \abstSys{e} \parac \abstSys{r}_2))) \mid e \in \sysSpace \wedge \intSetFunct(\abstSys{e}) \cap \intSetFunct(\abstSys{s}_1) = \emptyset = \intSetFunct(\abstSys{e}) \cap \intSetFunct(\abstSys{r}_2) \\ &\hspace{26em} \wedge e^1, e^2 \in \intSetFunct(\abstSys{e}) \wedge e^1 \neq e^2 \bigr\}, \\
		\mathcal{S}_{=} &\coloneqq \bigl\{ \connect{r_1^2}{r_2^2}(\connect{e^2}{r_2^1}(\connect{e^1}{r_1^1}(\abstSys{e} \parac \abstSys{r}_1 \parac \abstSys{r}_2))) \mid e \in \sysSpace \wedge \intSetFunct(\abstSys{e}) \cap \intSetFunct(\abstSys{r}_1) = \emptyset = \intSetFunct(\abstSys{e}) \cap \intSetFunct(\abstSys{r}_2) \\ &\hspace{26em} \wedge e^1, e^2 \in \intSetFunct(\abstSys{e}) \wedge e^1 \neq e^2 \bigr\}.
	\end{align*}
	Then, $\mathcal{S}_0 \cap \mathcal{S}_1 \cap \mathcal{S}_{=} \neq \emptyset$.
\end{theorem}
\begin{proof}
	Note that $\mathcal{S}_0$ corresponds to the set of systems for all possible $\abstSys{e}$ and $a = 0$ in \autoref{fig:bc-dishonest-r2}, $\mathcal{S}_1$ corresponds to \autoref{fig:bc-dishonest-r1} for $a = 1$, and $\mathcal{S}_{=}$ corresponds to the systems in \autoref{fig:bc-dishonest-s}. Let 
	\begin{equation*}
		\abstSys{t} \coloneqq \connect{r_1^2}{r_2^2}(\connect{s_1^2}{r_2^1}(\connect{s_0^1}{r_1^1}(\connect{s_0^2}{s_1^1}(\abstSys{s}_0 \parac \abstSys{s}_1 \parac \abstSys{r}_1 \parac \abstSys{r}_2))))
	\end{equation*}	
	be the system in \autoref{fig:bc-proof}. By composition-order invariance, we have for $\abstSys{e} = \connect{s_0^2}{s_1^1}(\abstSys{s}_0 \parac \abstSys{s}_1)$,
	\begin{equation*}
		\abstSys{e} \parac \abstSys{r}_1 \parac \abstSys{r}_2 = \connect{s_0^2}{s_1^1}(\abstSys{s}_0 \parac \abstSys{s}_1) \parac \abstSys{r}_1 \parac \abstSys{r}_2 = \connect{s_0^2}{s_1^1}(\abstSys{s}_0 \parac \abstSys{s}_1 \parac \abstSys{r}_1 \parac \abstSys{r}_2).
	\end{equation*}
	Hence, we obtain for $e^1 = s_0^1$, and $e^2 = s_1^2$, that $\abstSys{t} = \connect{r_1^2}{r_2^2}(\connect{e^2}{r_2^1}(\connect{e^1}{r_1^1}(\abstSys{e} \parac \abstSys{r}_1 \parac \abstSys{r}_2))) \in \mathcal{S}_{=}$. Again using composition-order invariance, we further have for $\hat{\abstSys{e}} = \connect{s_1^2}{r_2^1}(\abstSys{s}_1 \parac \abstSys{r}_2)$, $\hat{e}^1 = s_1^1$, and $\hat{e}^2 = r_2^2$,
	\begin{align*}
		t &= \connect{r_1^2}{r_2^2}(\connect{s_0^2}{s_1^1}(\connect{s_0^1}{r_1^1}(\connect{s_1^2}{r_2^1}(\abstSys{s}_1 \parac \abstSys{r}_2 \parac \abstSys{s}_0 \parac \abstSys{r}_1)))) \\
		&= \connect{r_1^2}{\hat{e}^2}(\connect{s_0^2}{\hat{e}^1}(\connect{s_0^1}{r_1^1}(\hat{\abstSys{e}} \parac \abstSys{s}_0 \parac \abstSys{r}_1))) \in \mathcal{S}_0.
	\end{align*}
	Finally, we have for $\tilde{\abstSys{e}} = \connect{s_0^1}{r_1^1}(\abstSys{s}_0 \parac \abstSys{r}_1)$, $\tilde{e}^1 = s_0^2$, and $\tilde{e}^2 = r_1^2$,
	\begin{align*}
		t &= \connect{r_1^2}{r_2^2}(\connect{s_1^2}{r_2^1}(\connect{s_0^2}{s_1^1}(\connect{s_0^1}{r_1^1}(\abstSys{s}_0 \parac \abstSys{r}_1 \parac \abstSys{s}_1 \parac \abstSys{r}_2)))) \\
		&= \connect{\tilde{e}^2}{r_2^2}(\connect{s_1^2}{r_2^1}(\connect{\tilde{e}^1}{s_1^1}(\tilde{\abstSys{e}} \parac \abstSys{s}_1 \parac \abstSys{r}_2))) \in \mathcal{S}_1.
	\end{align*}
	Therefore, we have $t \in \mathcal{S}_0 \cap \mathcal{S}_1 \cap \mathcal{S}_{=} \neq \emptyset$.
\end{proof}

To see why this implies the claimed impossibility, assume a protocol for broadcast exists and let $\abstSys{s}_a$ for $a \in \{0,1\}$ be a system that implements the protocol for the sender to broadcast the bit~$a$ and let $\abstSys{r}_1$ and $\abstSys{r}_2$ be systems for the two receivers such that these systems have distinct interface labels.\footnote{Since interface labels are only used for connecting systems and typical protocols do not depend on them, it is reasonable to assume that such systems with distinct interface labels exist.} The validity of the broadcast protocol implies that for all systems in $\mathcal{S}_0$, the subsystem~$\abstSys{r}_1$ decides on the bit~$0$ (say with probability more than $\frac{2}{3}$) and for all systems in $\mathcal{S}_1$, the subsystem~$\abstSys{r}_2$ decides on the bit~$1$ (with probability more than $\frac{2}{3}$). The consistency condition further implies that for all systems in $\mathcal{S}_{=}$, $\abstSys{r}_1$ and $\abstSys{r}_2$ decide on the same bit (with probability more than $\frac{2}{3}$). Now \autoref{thm:bcImp} says that there is a system that satisfies all three constraints, which is impossible. Therefore no such protocol exists.

\section{Functional System Algebras}\label{sec:funSys}
\subsection{Definitions}\label{sec:funSysDefs}
We now introduce special system algebras for functional systems that take inputs at dedicated input interfaces and produce outputs at their output interfaces, where the outputs are computed as a function of the inputs. This not only allows us to model systems that take a single input at each input interface and produce a single output at each output interface, but also much more general systems. For example, to model interactive systems that successively take inputs and produce outputs, one can consider the set of sequences of values as the domain of the functions. The function corresponding to a system then maps an entire input history to the output history and is a compact description of the system.

We define the parallel composition of two systems to be the function that evaluates both systems independently. Interface connection is defined in a way such that after connecting an input interface to an output interface, the input at the former equals the output at the latter. This corresponds to having a fixed point of a certain function determined by the connected interfaces and the system. One therefore has to choose~$\connectSet$ such that fixed points for all allowed connections exist.
Ideally, there is always a unique fixed point, because in this case, the interface connection operation is uniquely determined by this condition. If there are several fixed points, one has to be chosen in each case. A functional system algebra is therefore characterized by a set of functions~$\sysSpace$, a function~$\connectSet$ determining the allowed interface connections, and an appropriate choice of fixed points~$\fpChoiceFun$. We use boldface letters for functional systems to distinguish them from abstract systems.

\begin{definition}\label{def:funSysAlg}
	Let $\inSet$ and $\mathcal{X}$ be sets, let for all finite disjoint $\inInSet, \outInSet \subseteq \inSet$, $\sysIntSpace{\inInSet}{\outInSet} \subseteq \bigl( \mathcal{X}^\outInSet \bigr)^{\mathcal{X}^\inInSet}$ be a set of functions~$\sys{s} \colon \mathcal{X}^\inInSet \to \mathcal{X}^\outInSet$, and let $\sysSpace$ be the union of all $\sysIntSpace{\inInSet}{\outInSet}$. For $\sys{s} \in \sysIntSpace{\inInSet}{\outInSet}$, $\inIn \in \inInSet$, $\outIn \in \outInSet$, and $\tuple{x} \in \mathcal{X}^{\inInSet \setminus \{\inIn\}}$, let
	\begin{equation*}
		\fixedPoints{\sys{s}}{\inIn}{\outIn}{\tuple{x}} \coloneqq \bigl\{x_\inIn \in \mathcal{X} \mid \sys{s}\bigl(\tuple{x} \cup \{(\inIn, x_\inIn)\} \bigr)(\outIn) = x_\inIn \bigr\}
	\end{equation*}
	be the set of fixed points of the function $x_\inIn \mapsto \sys{s}\bigl(\tuple{x} \cup \{(\inIn, x_\inIn)\} \bigr)(\outIn)$. Further let $\connectSet \colon \sysSpace \to \mathcal{P}\bigl(\bigl\{ \{\inIn, \outIn \} \mid \inIn, \outIn \in \inSet \bigr\} \bigr)$ such that for all $\sys{s} \in \sysIntSpace{\inInSet}{\outInSet}$ and $\{\inIn, \outIn\} \in \connectSet(\sys{s})$, we have $\inIn \in \inInSet$, $\outIn \in \outInSet$ (or $\inIn \in \outInSet$, $\outIn \in \inInSet$),\footnote{We formally cannot require $\inIn \in \inInSet$ and $\outIn \in \outInSet$ because $\{\inIn, \outIn\}$ is unordered. To simplify the notation we will, however, always assume that $\inIn \in \inInSet$ and $\outIn \in \outInSet$ when we write $\{\inIn, \outIn\} \in \connectSet(\sys{s})$, and similarly for $\{\inIn', \outIn'\}$, $\{\inIn_1, \outIn_1\}$, etc.} and for all $\tuple{x} \in \mathcal{X}^{\inInSet \setminus \{\inIn\}}$, we have $\fixedPoints{\sys{s}}{\inIn}{\outIn}{\tuple{x}} \neq \emptyset$. Finally let $\fpChoiceFun^{\sys{s}}_{\inIn, \outIn} \colon \mathcal{X}^{\inInSet \setminus \{\inIn\}} \to \mathcal{X}$ for $\sys{s} \in \sysIntSpace{\inInSet}{\outInSet}$ and $\{\inIn, \outIn\} \in \connectSet(\sys{s})$ be a function such that for all $\tuple{x} \in \mathcal{X}^{\inInSet \setminus \{\inIn\}}$, we have $\chosenFixedPoint{\sys{s}}{\inIn}{\outIn}{\tuple{x}} \in \fixedPoints{\sys{s}}{\inIn}{\outIn}{\tuple{x}}$. Then, we define
	\begin{equation*}
		\genFunSys{\sysSpace}{\connectSet}{\fpChoiceFun} \coloneqq (\sysSpace, \intSetFunct, \parac, \connectSet, \connectFunc),
	\end{equation*}
	where $\fpChoiceFun$ is the set of all $\fpChoiceFun^{\sys{s}}_{\inIn, \outIn}$ and
	\begin{itemize}
		\item for $\sys{s} \in \sysIntSpace{\inInSet}{\outInSet}$, we have $\lambda(\sys{s}) = \inInSet \cup \outInSet$,
		
		\item for pairwise disjoint $\inInSet_1, \inInSet_2, \outInSet_1, \outInSet_2 \subseteq \inSet$, $\sys{s}_1 \in \sysIntSpace{\inInSet_1}{\outInSet_1}$, and $\sys{s}_2 \in \sysIntSpace{\inInSet_2}{\outInSet_2}$, we have $\sys{s}_1 \parac \sys{s}_2 \colon \mathcal{X}^{\inInSet_1 \cup \inInSet_2} \to \mathcal{X}^{\outInSet_1 \cup \outInSet_2}$ with
		\begin{equation}\label{eq:parallelCond}
			\forall \tuple{x} \in \mathcal{X}^{\inInSet_1 \cup \inInSet_2} \ \forall j \in \{1,2\} \ \forall \outIn_j \in \outInSet_j \ \ (\sys{s}_1 \parac \sys{s}_2)(\tuple{x})(\outIn_j) = \sys{s}_j \bigl(\tuple{x}|_{\inInSet_j}\bigr)(\outIn_j),
		\end{equation}
		
		\item and for $\sys{s} \in \sysIntSpace{\inInSet}{\outInSet}$ and $\{\inIn, \outIn\} \in \connectSet(\sys{s})$, we have $\connect{\inIn}{\outIn}(\sys{s}) \colon \mathcal{X}^{\inInSet \setminus \{\inIn\}} \to \mathcal{X}^{\outInSet \setminus \{\outIn\}}$ with
		\begin{equation*}
			\forall \tuple{x} \in \mathcal{X}^{\inInSet \setminus \{\inIn\}} \ \ \connect{\inIn}{\outIn}(\sys{s})(\tuple{x}) = \sys{s}\bigl(\tuple{x} \cup \bigl\{ \bigl(\inIn, \chosenFixedPoint{\sys{s}}{\inIn}{\outIn}{\tuple{x}} \bigr) \bigr\} \bigr)|_{\outInSet \setminus \{\outIn\}}.
		\end{equation*}
	\end{itemize}
	
	If $\sysSpace$ is closed under $\parac$ and $\connectFunc$,\footnote{That is, for $\sys{s}_1, \sys{s}_2 \in \sysSpace$ for which $\parac$ is defined, $\sys{s}_1 \parac \sys{s}_2 \in \sysSpace$, and for $\sys{s} \in \sysSpace$ and $\{\inIn, \outIn\} \in \connectSet(\sys{s})$, $\connect{\inIn}{\outIn}(\sys{s}) \in \sysSpace$.} and if for all $\sys{s}_1, \sys{s}_2 \in \sysSpace$ with $\intSetFunct(\sys{s}_1) \cap \intSetFunct(\sys{s}_2) = \emptyset$, for $j \in \{1,2\}$, and for all $\inIn, \outIn \in \intSetFunct(\sys{s}_j)$, we have $\{\inIn, \outIn\} \in \connectSet(\sys{s}_1 \parac \sys{s}_2) \Longleftrightarrow \{\inIn, \outIn\} \in \connectSet(\sys{s}_j)$, then $\genFunSys{\sysSpace}{\connectSet}{\fpChoiceFun}$ is a $\inSet$-system algebra. In this case, we say $\genFunSys{\sysSpace}{\connectSet}{\fpChoiceFun}$ is a \emph{functional $\inSet$-system algebra over~$\mathcal{X}$}.
\end{definition}

\begin{remark}
	Our definition of interface connections via fixed points implies that the ``line'' connecting two systems has no effect on the values. This means in particular that it does not introduce delays or transmission errors. If, e.g., delays are required (and the domains and functions are defined at a level such that delays can be specified, see \autoref{sec:causalSys}), one can allow the connection of systems only via dedicated channel systems that introduce the desired delays.
\end{remark}

The choice of fixed points is crucial for obtaining a system algebra with desired properties. For example, if the system algebra is supposed to model a class of real-world systems, the chosen fixed point should correspond to the value generated by the real system. The choice of the fixed point can also influence whether connection-order invariance holds. If fixed points are not unique, a reasonable requirement is that they are consistently chosen in the sense that whenever two systems have the same set of fixed points for a specific interface connection, the same fixed point is chosen for both systems.

\begin{definition}
	Let $(\sysSpace, \intSetFunct, \parac, \connectSet, \connectFunc)$ be a functional $\inSet$-system algebra over~$\mathcal{X}$. We say it has \emph{unique fixed points} if $\bigl\lvert \fixedPoints{\sys{s}}{\inIn}{\outIn}{\tuple{x}} \bigr\rvert = 1$ for all $\sys{s} \in \sysIntSpace{\inInSet}{\outInSet}$, $\{\inIn, \outIn\} \in \connectSet(\sys{s})$, and $\tuple{x} \in \mathcal{X}^{\inInSet \setminus \{\inIn\}}$.
	If there exists\footnote{While $\fpChoiceFun$ uniquely determines $\connectFunc$, the converse is not true. For example, if $\sys{s}$, $\inIn$, and $\outIn$ are such that the input at interface~$\inIn$ does not influence the outputs of~$\sys{s}$ at interfaces different from~$\outIn$, the choice of $\chosenFixedPoint{\sys{s}}{\inIn}{\outIn}{\cdot}$ is irrelevant. Hence, we only require for consistently chosen fixed points that $\connectFunc$ can be explained consistently.} $\fpChoiceFun$ such that $(\sysSpace, \intSetFunct, \parac, \connectSet, \connectFunc) = \genFunSys{\sysSpace}{\connectSet}{\fpChoiceFun}$ and for all $\sys{s} \in \sysIntSpace{\inInSet}{\outInSet}, \sys{s}' \in \sysIntSpace{\inInSet'}{\outInSet'}$, $\{\inIn, \outIn\} \in \connectSet(\sys{s}) \cap \connectSet(\sys{s}')$, $\tuple{x} \in \mathcal{X}^{\inInSet \setminus \{\inIn\}}$, and $\tuple{x}' \in \mathcal{X}^{\inInSet' \setminus \{\inIn\}}$,
	\begin{equation*}
		\fixedPoints{\sys{s}}{\inIn}{\outIn}{\tuple{x}} = \fixedPoints{\sys{s}'}{\inIn}{\outIn}{\tuple{x}'} \ \ \Longrightarrow \ \ \chosenFixedPoint{\sys{s}}{\inIn}{\outIn}{\tuple{x}} = \chosenFixedPoint{\sys{s}'}{\inIn}{\outIn}{\tuple{x}'},
	\end{equation*}
	we say $(\sysSpace, \intSetFunct, \parac, \connectSet, \connectFunc)$ has \emph{consistently chosen fixed points}.
\end{definition}

\subsection{Basic Properties}\label{sec:funSysProps}
It is easy to find functional system algebras that are not connection-order invariant if fixed points are not chosen consistently. Consider for example a system~$\sys{s}$ with two fixed points~$x, x'$ for connecting interfaces $\inIn$ and $\outIn$ regardless of inputs at other interfaces. It is possible that $x$ is chosen for $\connect{\inIn}{\outIn}(\sys{s})$ but for the system $\sys{s}' \coloneqq \connect{\inIn'}{\outIn'}(\sys{s})$, which has the same two fixed points for connecting $\inIn$ and $\outIn$, $x'$ is chosen. In this case, one can have $\connect{\inIn'}{\outIn'}\bigl( \connect{\inIn}{\outIn}(\sys{s}) \bigr) \neq \connect{\inIn}{\outIn}\bigl( \connect{\inIn'}{\outIn'}(\sys{s}) \bigr)$. As the following example shows, consistently chosen fixed points are, however, not sufficient to guarantee connection-order invariance.

\paragraph{Consistently chosen fixed points do not imply connection-order invariance\ifelsarticle\else.\fi} Let $\mathcal{X} = \{0, 1\}$ and consider a system algebra where $0$ is preferred to $1$ as a fixed point, i.e., $\fixedPoints{\sys{s}}{\inIn}{\outIn}{\tuple{x}} = \{0,1\} \Rightarrow \chosenFixedPoint{\sys{s}}{\inIn}{\outIn}{\tuple{x}} = 0$. This clearly guarantees consistently chosen fixed points. Now consider the system $\sys{s} \in \sysIntSpace{\{\inIn_1, \inIn_2\}}{\{\outIn_1, \outIn_2, \outIn_3\}}$, $\sys{s}\bigl( \{ (\inIn_1, x_1), (\inIn_2, x_2) \} \bigr) = \{ (\outIn_1, 1-x_1), (\outIn_2, 1-x_2), (\outIn_3, x_1) \}$. If we connect $\inIn_2$ and $\outIn_1$, there is exactly one fixed point for each input at $\inIn_1$ because we essentially compose two functions, which both invert their input bit. Thus, $\connect{\inIn_2}{\outIn_1}(\sys{s}) \bigl(\{(\inIn_1, x_1)\} \bigr) = \{(\outIn_2, x_1), (\outIn_3, x_1)\}$. If we now connect $\inIn_1$ and $\outIn_2$, both $0$ and $1$ are fixed points and $0$ is chosen, i.e., we get $\connect{\inIn_1}{\outIn_2} \bigl(\connect{\inIn_2}{\outIn_1}(\sys{s}) \bigr) (\emptyset) = \{(\outIn_3, 0)\}$. However, first connecting $\inIn_1$ and $\outIn_2$ yields $\connect{\inIn_1}{\outIn_2}(\sys{s}) \bigl(\{(\inIn_2, x_2)\} \bigr) = \{(\outIn_1, x_2), (\outIn_3, 1-x_2)\}$. If we then connect $\inIn_2$ and $\outIn_1$, $0$ is again the preferred fixed point and therefore $\connect{\inIn_2}{\outIn_1} \bigl(\connect{\inIn_1}{\outIn_2}(\sys{s}) \bigr) (\emptyset) = \{(\outIn_3, 1)\}$, which implies $\connect{\inIn_2}{\outIn_1} \bigl(\connect{\inIn_1}{\outIn_2}(\sys{s}) \bigr) \neq \connect{\inIn_1}{\outIn_2} \bigl(\connect{\inIn_2}{\outIn_1}(\sys{s}) \bigr)$.

While consistently chosen fixed points are not sufficient, the next lemma shows that connection-order invariance follows from unique fixed points if $\connectSet$ permits reordering.
\begin{lemma}\label{lem:unifpConinv}
	Functional system algebras with unique fixed points and $\connectSet$ that permits reordering are connection-order invariant.
\end{lemma}
\begin{proof}
	Let $\inSet$ and $\mathcal{X}$ be sets, and let $(\sysSpace, \intSetFunct, \parac, \connectSet, \connectFunc) = \genFunSys{\sysSpace}{\connectSet}{\fpChoiceFun}$ be a functional $\inSet$-system algebra over $\mathcal{X}$ with unique fixed points such that $\connectSet$ permits reordering. Further let $\sys{s} \in \sysSpace$, $\{\inIn, \outIn\} \in \connectSet(\sys{s})$, and $\{\inIn', \outIn'\} \in \connectSet\bigl(\connect{\inIn}{\outIn}(\sys{s}) \bigr)$. We have to show that $\connect{\inIn'}{\outIn'}\bigl( \connect{\inIn}{\outIn}(\sys{s}) \bigr) = \connect{\inIn}{\outIn}\bigl( \connect{\inIn'}{\outIn'}(\sys{s}) \bigr)$. To this end, let $\tuple{x} \in \mathcal{X}^{\inInSet \setminus \{\inIn, \inIn'\}}$. By \autoref{def:funSysAlg}, there exist $x_{\inIn}, x_{\inIn'}, x'_{\inIn}, x'_{\inIn'} \in \mathcal{X}$ such that
	\begin{align*}
		\connect{\inIn'}{\outIn'}\bigl( \connect{\inIn}{\outIn}(\sys{s})\bigr) (\tuple{x}) &= \sys{s}\bigl(\tuple{x} \cup \bigl\{ \bigl(\inIn, x_{\inIn} \bigr), \bigl(\inIn', x_{\inIn'} \bigr) \bigr\} \bigr)|_{\outInSet \setminus \{\outIn, \outIn'\}}, \\
		\connect{\inIn}{\outIn}\bigl( \connect{\inIn'}{\outIn'}(\sys{s})\bigr) (\tuple{x}) &= \sys{s}\bigl(\tuple{x} \cup \bigl\{ \bigl(\inIn, x'_{\inIn} \bigr), \bigl(\inIn', x'_{\inIn'} \bigr) \bigr\} \bigr)|_{\outInSet \setminus \{\outIn, \outIn'\}}.
	\end{align*}
	We further have $x_{\inIn} \in \fixedPoints{\sys{s}}{\inIn}{\outIn}{\tuple{x} \cup \{ (\inIn', x_{\inIn'}) \}}$ and since fixed points are unique, $x_{\inIn}$ is chosen for connecting $\inIn$ and $\outIn$ on input $\tuple{x} \cup \{ (\inIn', x_{\inIn'}) \}$. This implies $x_{\inIn'} \in \fixedPoints{\connect{\inIn}{\outIn}(\sys{s})}{\inIn'}{\outIn'}{\tuple{x}}$. Moreover, $x'_{\inIn} \in \fixedPoints{\sys{s}}{\inIn}{\outIn}{\tuple{x} \cup \{ (\inIn', x'_{\inIn'} ) \}}$ and thus $x'_{\inIn'} \in \fixedPoints{\connect{\inIn}{\outIn}(\sys{s})}{\inIn'}{\outIn'}{\tuple{x}}$. Because $\abs[\Big]{\fixedPoints{\connect{\inIn}{\outIn}(\sys{s})}{\inIn'}{\outIn'}{\tuple{x}}} = 1$, this implies $x_{\inIn'} = x'_{\inIn'}$. Therefore, $\fixedPoints{\sys{s}}{\inIn}{\outIn}{\tuple{x} \cup \{ (\inIn', x_{\inIn'} ) \}} = \fixedPoints{\sys{s}}{\inIn}{\outIn}{\tuple{x} \cup \{ (\inIn', x'_{\inIn'} ) \}}$ and we obtain $x_{\inIn} = x'_{\inIn}$. We conclude $\connect{\inIn'}{\outIn'}\bigl( \connect{\inIn}{\outIn}(\sys{s})\bigr) (\tuple{x}) = \connect{\inIn}{\outIn}\bigl( \connect{\inIn'}{\outIn'}(\sys{s})\bigr) (\tuple{x})$.
\end{proof}

We next prove that the parallel composition operation of functional system algebras is always associative and commutative.
\begin{lemma}\label{lem:funasoccom}
	For any functional system algebra, parallel composition is associative and commutative.
\end{lemma}
\begin{proof}
	Let $\inSet$ and $\mathcal{X}$ be sets, and let $(\sysSpace, \intSetFunct, \parac, \connectSet, \connectFunc) = \genFunSys{\sysSpace}{\connectSet}{\fpChoiceFun}$ be a functional $\inSet$-system algebra over $\mathcal{X}$. Further let $\inInSet_1, \inInSet_2, \inInSet_3, \outInSet_1, \outInSet_2, \outInSet_3 \subseteq \inSet$ be pairwise disjoint, let $\sys{s}_j \in \sysIntSpace{\inInSet_j}{\outInSet_j}$ for $j \in \{1,2,3\}$, and let $\tuple{x} \in \mathcal{X}^{\inInSet_1 \cup \inInSet_2 \cup \inInSet_3}$ and $\outIn \in \outInSet_k$ for some $k \in \{1,2,3\}$. If $k = 3$, we have by equation~\eqref{eq:parallelCond}, $\bigl((\sys{s}_1 \parac \sys{s}_2) \parac \sys{s}_3 \bigr)(\tuple{x})(\outIn) = \sys{s}_3 \bigl(\tuple{x}|_{\inInSet_3} \bigr)(\outIn)$. If $k \neq 3$, we have $\bigl((\sys{s}_1 \parac \sys{s}_2) \parac \sys{s}_3 \bigr)(\tuple{x})(\outIn) = (\sys{s}_1 \parac \sys{s}_2) \bigl(\tuple{x}|_{\inInSet_1 \cup \inInSet_2} \bigr)(\outIn) = \sys{s}_k \bigl(\tuple{x}|_{\inInSet_k} \bigr)(\outIn)$. Similarly, $\bigl(\sys{s}_1 \parac (\sys{s}_2 \parac \sys{s}_3) \bigr)(\tuple{x})(\outIn) = \sys{s}_k \bigl(\tuple{x}|_{\inInSet_k} \bigr)(\outIn)$. Hence, $(\sys{s}_1 \parac \sys{s}_2) \parac \sys{s}_3 = \sys{s}_1 \parac (\sys{s}_2 \parac \sys{s}_3)$, i.e., $\parac$ is associative. Commutativity follows directly from equation~\eqref{eq:parallelCond}: for $\tuple{x} \in \mathcal{X}^{\inInSet_1 \cup \inInSet_2}$ and $\outIn \in \outInSet_k$ for some $k \in \{1,2\}$, we have $(\sys{s}_1 \parac \sys{s}_2)(\tuple{x})(\outIn) = \sys{s}_k \bigl(\tuple{x}|_{\inInSet_k} \bigr)(\outIn) = (\sys{s}_2 \parac \sys{s}_1)(\tuple{x})(\outIn)$.
\end{proof}

We further show that connection-order invariance implies composition-order invariance for functional system algebras with consistently chosen fixed points.
\begin{lemma}\label{lem:coninvCompinvCons}
	A functional system algebra with consistently chosen fixed points is composition-order invariant if and only if it is connection-order invariant.
\end{lemma}
\begin{proof}
	By definition, a functional system algebra that is not connection-order invariant is not composition-order invariant. For the other direction, let $(\sysSpace, \intSetFunct, \parac, \connectSet, \connectFunc) = \genFunSys{\sysSpace}{\connectSet}{\fpChoiceFun}$ be a connection-order invariant functional $\inSet$-system algebra over~$\mathcal{X}$ with consistently chosen fixed points. Associativity and commutativity of~$\parac$ follows from \autoref{lem:funasoccom}. It remains to show that $\connect{\inIn}{\outIn}(\sys{s}_1) \parac \sys{s}_2 = \connect{\inIn}{\outIn}(\sys{s}_1 \parac \sys{s}_2)$ for all $\sys{s}_1 \in \sysIntSpace{\inInSet_1}{\outInSet_1}$, $\sys{s}_2 \in \sysIntSpace{\inInSet_2}{\outInSet_2}$, and $\{ \inIn, \outIn \} \in \connectSet(\sys{s}_1)$ such that $\intSetFunct(\sys{s}_1) \cap \intSetFunct(\sys{s}_2) = \emptyset$. To this end, let $\tuple{x} \in \mathcal{X}^{(\inInSet_1 \cup \inInSet_2) \setminus \{\inIn\}}$, $\outIn_1 \in \outInSet_1 \setminus \{\outIn\}$, and $\outIn_2 \in \outInSet_2$. Because inputs to~$\sys{s}_1$ do not affect the outputs of~$\sys{s}_2$, we have
	\begin{equation*}
		\bigl(\connect{\inIn}{\outIn}(\sys{s}_1) \parac \sys{s}_2 \bigr) (\tuple{x}) (\outIn_2) = \sys{s}_2 (\tuple{x}|_{\inInSet_2})(\outIn_2) = \bigl(\connect{\inIn}{\outIn}(\sys{s}_1 \parac \sys{s}_2) \bigr) (\tuple{x}) (\outIn_2).
	\end{equation*}
	By the same reasoning, the set of fixed points for connecting interfaces~$\inIn$ and $\outIn$ is not influenced by~$\sys{s}_2$, i.e., $\fixedPoints{\sys{s}_1 \parac \sys{s}_2}{\inIn}{\outIn}{\tuple{x}} = \fixedPoints{\sys{s}_1}{\inIn}{\outIn}{\tuple{x}|_{\inInSet_1}}$. Consistently chosen fixed points imply $ \chosenFixedPoint{\sys{s}_1 \parac \sys{s}_2}{\inIn}{\outIn}{\tuple{x}} = \chosenFixedPoint{\sys{s}_1}{\inIn}{\outIn}{\tuple{x}|_{\inInSet_1}}$ and hence
	\begin{align*}
		\bigl(\connect{\inIn}{\outIn}(\sys{s}_1) \parac \sys{s}_2 \bigr) (\tuple{x}) (\outIn_1) &= \connect{\inIn}{\outIn}(\sys{s}_1) (\tuple{x}|_{\inInSet_1}) (\outIn_1)
		= \sys{s}_1 \Bigl(\tuple{x}|_{\inInSet_1} \cup \bigl\{ \bigl(\inIn, \chosenFixedPoint{\sys{s}_1}{\inIn}{\outIn}{\tuple{x}|_{\inInSet_1}} \bigr) \bigr\} \Bigr) (\outIn_1) \\
		&= \sys{s}_1 \Bigl(\tuple{x}|_{\inInSet_1} \cup \Bigl\{ \Bigl(\inIn, \chosenFixedPoint{\sys{s}_1 \parac \sys{s}_2}{\inIn}{\outIn}{\tuple{x}} \Bigr) \Bigr\} \Bigr) (\outIn_1) \\
		&= (\sys{s}_1 \parac \sys{s}_2) \Bigl(\tuple{x} \cup \Bigl\{ \Bigl(\inIn, \chosenFixedPoint{\sys{s}_1 \parac \sys{s}_2}{\inIn}{\outIn}{\tuple{x}} \Bigr) \Bigr\} \Bigr) (\outIn_1) \\
		&= \bigl(\connect{\inIn}{\outIn}(\sys{s}_1 \parac \sys{s}_2) \bigr) (\tuple{x}) (\outIn_1).
	\end{align*}
	Therefore, $\connect{\inIn}{\outIn}(\sys{s}_1) \parac \sys{s}_2 = \connect{\inIn}{\outIn}(\sys{s}_1 \parac \sys{s}_2)$ and hence composition-order invariance holds.
\end{proof}

Since unique fixed points imply consistently chosen fixed points, we can combine \autoref{lem:unifpConinv} and \autoref{lem:coninvCompinvCons} to obtain the following result.
\begin{theorem}\label{lem:unifpCompinv}
	Functional system algebras with unique fixed points and $\connectSet$ that permits reordering are composition-order invariant.
\end{theorem}

\subsection{Merging Input and Output Interfaces}
It is often desirable to have a system algebra that allows arbitrary interface connections, as we have assumed in \autoref{sec:bcImp}. A functional system algebra does not provide this because one can only connect input interfaces to output interfaces. A single interface connection in a functional system algebra also allows only communication in one direction, while one might prefer to connect systems to each other such that they can communicate in both directions. These drawbacks can be mitigated easily by considering an interface as a pair of an input interface and an output interface and by defining interface connection as connecting both interfaces of the corresponding pair: Let $(\sysSpace, \intSetFunct, \parac, \connectSet, \connectFunc)$ be a functional $(\inSet \times \{0,1\})$-system algebra that allows connecting arbitrary input interface to arbitrary output interfaces (i.e., $\connectSet(\sys{s}) = \bigl\{ \{\inIn, \outIn \} \mid \inIn \in \inInSet, \outIn \in \outInSet \bigr\}$ for $\sys{s} \in \sysIntSpace{\inInSet}{\outInSet}$). We then define the $\inSet$-system algebra $(\sysSpace', \intSetFunct', \parac', \connectSet', \connectFunc')$ via $\sysSpace' \coloneqq \{\sys{s} \in \sysIntSpace{\mathcal{L} \times \{1\}}{\mathcal{L} \times \{0\}} \mid \mathcal{L} \subseteq \inSet \}$, $\intSetFunct'(\sys{s}) = \mathcal{L}$ for $\sys{s} \in \sysIntSpace{\mathcal{L} \times \{1\}}{\mathcal{L} \times \{0\}}$, ${\parac'} = {\parac}$, $\connectSet(\sys{s}) = \bigl\{ \{\inIn, \inIn' \} \mid \inIn, \inIn' \in \intSetFunct(\sys{s}) \bigr\}$, $\connect{\inIn}{\inIn'}'(\sys{s}) = \connect{(\inIn', 1)}{(\inIn, 0)} \bigl(\connect{(\inIn, 1)}{(\inIn', 0)}(\sys{s}) \bigr)$. See \autoref{fig:mergedConnect} for a depiction of a system and the connection operation in $(\sysSpace', \intSetFunct', \parac', \connectSet', \connectFunc')$. Note that the system algebra $(\sysSpace', \intSetFunct', \parac', \connectSet', \connectFunc')$ is composition-order invariant if $(\sysSpace, \intSetFunct, \parac, \connectSet, \connectFunc)$ is.

\begin{figure}[t]
	\centering
	\begin{tikzpicture}[system/.style={rectangle,draw,minimum width=5cm,minimum height=1.5cm}]
		\def\delta{0.75mm}
		\def\dasdelta{2.5*\delta}
		
		\draw node (sys) [system] {$\sys{s}$};

		\draw[->,thick] ($(sys.north) - (12mm-\delta,0)$) -- ($(sys.north) - (12mm-\delta,-4mm)$);
		\draw[<-,thick] ($(sys.north) - (12mm+\delta,0)$) -- ($(sys.north) - (12mm+\delta,-4mm)$);

		\draw[<-,thick] ($(sys.north) + (12mm-\delta,0)$) -- ($(sys.north) + (12mm-\delta,4mm)$);
		\draw[->,thick] ($(sys.north) + (12mm+\delta,0)$) -- ($(sys.north) + (12mm+\delta,4mm)$);
		
		\draw[->,thick] ($(sys.east) + (0,\delta)$) -- ($(sys.east) + (15mm,\delta)$);
		\draw[<-,thick] ($(sys.east) + (0,-\delta)$) -- ($(sys.east) + (15mm,-\delta)$);
		
		\draw ($(sys.north) - (12mm-\delta,-4mm)$) -- ($(sys.north) - (12mm-\delta,-10mm+\delta)$) -- ($(sys.north) + (12mm-\delta,10mm-\delta)$) -- ($(sys.north) + (12mm-\delta,4mm)$);
		\draw ($(sys.north) - (12mm+\delta,-4mm)$) -- ($(sys.north) - (12mm+\delta,-10mm-\delta)$) -- ($(sys.north) + (12mm+\delta,10mm+\delta)$) -- ($(sys.north) + (12mm+\delta,4mm)$);
		\draw[dashed] ($(sys.north) - (2mm,0)$) -- ($(sys.north) - (2mm,-4.5mm)$) -- ($(sys.north) - (12mm-\dasdelta,-4.5mm)$) -- ($(sys.north) - (12mm-\dasdelta,-10mm+\dasdelta)$) -- ($(sys.north) + (12mm-\dasdelta,10mm-\dasdelta)$) -- ($(sys.north) + (12mm-\dasdelta,4.5mm)$) -- ($(sys.north) + (2mm,4.5mm)$) -- ($(sys.north) + (2mm,0)$);
		\draw[dashed] ($(sys.north) - (22mm,0)$) -- ($(sys.north) - (22mm,-4.5mm)$) -- ($(sys.north) - (12mm+\dasdelta,-4.5mm)$) -- ($(sys.north) - (12mm+\dasdelta,-10mm-\dasdelta)$) -- ($(sys.north) + (12mm+\dasdelta,10mm+\dasdelta)$) -- ($(sys.north) + (12mm+\dasdelta,4.5mm)$) -- ($(sys.north) + (22mm,4.5mm)$) -- ($(sys.north) + (22mm,0)$);
		
		\draw[dashed] ($(sys.east) + (0,5mm)$) -- ($(sys.east) + (9mm,5mm)$) -- ($(sys.east) + (9mm,\dasdelta)$) -- ($(sys.east) + (15mm,\dasdelta)$) -- ($(sys.east) + (15mm,-\dasdelta)$) -- ($(sys.east) + (9mm,-\dasdelta)$) -- ($(sys.east) + (9mm,-5mm)$) -- ($(sys.east) + (0,-5mm)$);

		\node[left] at ($(sys.north) - (12mm,-2mm)$) {\scriptsize$(\inIn_1, 1)$};
		\node[right] at ($(sys.north) - (12mm,-2mm)$) {\scriptsize$(\inIn_1, 0)$};
		\node[left] at ($(sys.north) - (13.5mm,-9mm)$) {$\inIn_1$};
		
		\node[left] at ($(sys.north) + (12mm,2mm)$) {\scriptsize$(\inIn_2, 1)$};
		\node[right] at ($(sys.north) + (12mm,2mm)$) {\scriptsize$(\inIn_2, 0)$};
		\node[right] at ($(sys.north) + (14mm,9mm)$) {$\inIn_2$};
		
		\node[below right] at ($(sys.east) - (\delta,0)$) {\scriptsize$(\inIn_3, 1)$};
		\node[above right] at ($(sys.east) - (\delta,0)$) {\scriptsize$(\inIn_3, 0)$};
		\node[right] at ($(sys.east) + (15mm,0)$) {$\inIn_3$};
	\end{tikzpicture}
	\caption{The system~$\connect{\inIn_1}{\inIn_2}'(\sys{s})$ in the system algebra with merged input and output interfaces.}
	\label{fig:mergedConnect}
\end{figure}
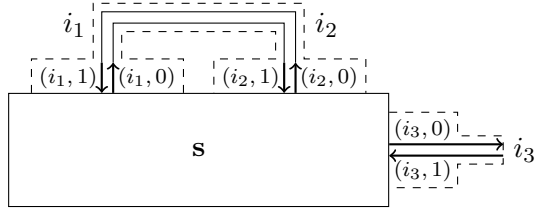

By doing this, one can use functional system algebras as the ones we develop in later sections to instantiate system algebras with undirected interfaces.

\begin{remark}
	As defined above, $\connectSet'$ allows connecting an interface to itself. While this is well-defined, one might want to exclude it since it is typically not needed.
\end{remark}

\section{The System Algebras of Monotone and Continuous Systems}
By \autoref{thm:monotoneFP} and \autoref{thm:continuousFP}, monotone functions on CPOs and $\omega$-continuous functions on $\omega$-CPOs have least fixed points. We can use this to define a functional system algebra consisting of such functions. While fixed points of monotone functions are well-studied in domain theory, we still need to show that choosing the least fixed point yields a functional system algebra with the desired properties. To obtain a system algebra, we need to prove that the set of system is closed under interface connections defined via choosing the least fixed point. We prove that this is the case both for monotone and continuous functions. Furthermore, we show that composition-order invariance holds for both system algebras under the additional assumption on the CPO that every nonempty chain has an infimum. Finally, we describe Kahn networks as a special case of the system algebra of continuous systems.

\subsection{Monotone Systems}\label{sec:monotoneSys}
Let $(\mathcal{X}, \preceq)$ be a CPO. Recall that for a set $\inInSet$, we also write~$\preceq$ for the partial order on $\mathcal{X}^{\inInSet}$ where for $\tuple{x}, \tuple{y} \in \mathcal{X}^{\inInSet}$, $\tuple{x} \preceq \tuple{y}$ if $\tuple{x}(\inIn) \preceq \tuple{y}(\inIn)$ for all $\inIn \in \inInSet$. Monotonicity is therefore also defined for functions on such tuples.
Let $\inSet$ be a set and define for finite disjoint~$\inInSet, \outInSet \subseteq \inSet$,
\begin{equation*}
	\monoSysIntSpace{\inInSet}{\outInSet} \coloneqq \bigl\{ \sys{s} \colon \mathcal{X}^{\inInSet} \to \mathcal{X}^{\outInSet} \mid \sys{s} \text{ is monotone} \bigr\}.
\end{equation*}
We can use the existence of least fixed points to define a system algebra. Before we do so, we need the following simple lemma.
\begin{lemma} \label{lem:monoSysFixpoint}
	Let $\inInSet, \outInSet \subseteq \inSet$ be finite disjoint sets, let $\sys{s} \in \monoSysIntSpace{\inInSet}{\outInSet}$, and let $\inIn \in \inInSet$, $\outIn \in \outInSet$. Then, for all $\tuple{x} \in \mathcal{X}^{\inInSet \setminus \{\inIn\}}$, there exists a least $x_{\inIn} \in \mathcal{X}$ such that $\sys{s}\bigl(\tuple{x} \cup \{(\inIn, x_{\inIn})\} \bigr)(\outIn) = x_{\inIn}$.
\end{lemma}
\begin{proof}
	For $\tuple{x} \in \mathcal{X}^{\inInSet \setminus \{\inIn\}}$, let $f \colon \mathcal{X} \to \mathcal{X}, x_{\inIn} \mapsto \sys{s}(\tuple{x} \cup \{(\inIn, x_{\inIn})\})(\outIn)$ and note that the fixed points of~$f$ precisely correspond to the values $x_{\inIn} \in \mathcal{X}$ such that $\sys{s}\bigl(\tuple{x} \cup \{(\inIn, x_{\inIn})\} \bigr)(\outIn) = x_{\inIn}$. Since $\sys{s}$ is monotone, so is~$f$, and therefore $f$ has a least fixed point by \autoref{thm:monotoneFP}.
\end{proof}

\begin{definition}
	We define $\genFunSys{\monoSysSpace}{\connectSet}{\fpChoiceFun}$, the \emph{$\inSet$-system algebra of monotone systems over $\mathcal{X}$}, as follows: For $\sys{s} \in \monoSysIntSpace{\inInSet}{\outInSet}$, $\connectSet(\sys{s}) \coloneqq \bigl\{ \{\inIn, \outIn \} \mid \inIn \in \inInSet, \outIn \in \outInSet \bigr\}$ and $\chosenFixedPoint{\sys{s}}{\inIn}{\outIn}{\tuple{x}}$ is the least element of $\fixedPoints{\sys{s}}{\inIn}{\outIn}{\tuple{x}}$ for all $\inIn \in \inInSet$, $\outIn \in \outInSet$, and $\tuple{x} \in \mathcal{X}^{\inInSet \setminus \{\inIn\}}$.
\end{definition}

For this to actually be a functional system algebra, we need that the parallel composition of systems and connecting interfaces of a system again yield a system in our algebra, i.e., a monotone function. While this is straightforward for parallel composition, it is nontrivial for interface connections. The essence of the proof we provide below is showing that $\chosenFixedPoint{\sys{s}}{\inIn}{\outIn}{\tuple{x}_1} \preceq \chosenFixedPoint{\sys{s}}{\inIn}{\outIn}{\tuple{x}_2}$ whenever $\tuple{x}_1 \preceq \tuple{x}_2$. To this end, we exploit the constructive nature of \autoref{thm:monotoneFP} that can be used to find the fixed points $\chosenFixedPoint{\sys{s}}{\inIn}{\outIn}{\tuple{x}_1}$ and $\chosenFixedPoint{\sys{s}}{\inIn}{\outIn}{\tuple{x}_2}$.
\begin{theorem}\label{thm:monoIsAlg}
	$(\monoSysSpace, \intSetFunct, \parac, \connectSet, \connectFunc) = \genFunSys{\monoSysSpace}{\connectSet}{\fpChoiceFun}$ is a functional $\inSet$-system algebra over $\mathcal{X}$.
\end{theorem}
\begin{proof}
	Since the relation~$\preceq$ is defined componentwise, it is clear from the definition of parallel composition by equation~\eqref{eq:parallelCond} that $\monoSysSpace$ is closed under $\parac$. It remains to show that $\monoSysSpace$ is also closed under connecting interfaces. To this end, let $\inInSet, \outInSet \subseteq \inSet$ be finite disjoint sets and let $\sys{s} \in \monoSysIntSpace{\inInSet}{\outInSet}$. Further let $\inIn \in \inInSet$, $\outIn \in \outInSet$. We have to show that $\connect{\inIn}{\outIn}(\sys{s}) \in \monoSysIntSpace{\inInSet \setminus \{\inIn\}}{\outInSet \setminus \{\outIn\}}$. To show that $\connect{\inIn}{\outIn}(\sys{s})$ is monotone, let $\tuple{x}_1, \tuple{x}_2 \in \mathcal{X}^{\inInSet \setminus \{\inIn\}}$, such that $\tuple{x}_1 \preceq \tuple{x}_2$. We then have $\connect{\inIn}{\outIn}(\sys{s})(\tuple{x}_k) = \sys{s}\bigl(\tuple{x}_k \cup \bigl\{ \bigl(\inIn, \chosenFixedPoint{\sys{s}}{\inIn}{\outIn}{\tuple{x}_k} \bigr) \bigr\} \bigr)|_{\outInSet \setminus \{\outIn\}}$ for $k \in \{1,2\}$, where $\chosenFixedPoint{\sys{s}}{\inIn}{\outIn}{\tuple{x}_k}$ is the least fixed point of the function $x_{\inIn} \mapsto \sys{s}(\tuple{x}_k \cup \{(\inIn, x_{\inIn})\})(\outIn)$. Let $x_k^{\alpha}$ for ordinals~$\alpha$ be as in \autoref{thm:monotoneFP} such that $\chosenFixedPoint{\sys{s}}{\inIn}{\outIn}{\tuple{x}_k} = x_k^{\hat{\alpha}_k}$ for some ordinal~$\hat{\alpha}_k$. We show by transfinite induction over~$\alpha$ that $x_1^{\alpha} \preceq x_2^{\alpha}$ for all ordinals~$\alpha$.
	\begin{enumerate}[label=(\roman*)]
		\item For $\alpha = 0$, we have $x_1^{\alpha} = \min(\mathcal{X}) = x_2^{\alpha}$.
		
		\item Assume $x_1^{\alpha} \preceq x_2^{\alpha}$ for some ordinal~$\alpha$. We then have $\tuple{x}_1 \cup \{(\inIn, x_1^{\alpha})\} \preceq \tuple{x}_2 \cup \{(\inIn, x_2^{\alpha})\}$ and since $\sys{s}$ is monotone,
		\begin{equation*}
			x_1^{\alpha+1} = \sys{s} \bigl(\tuple{x}_1 \cup \{(\inIn, x_1^{\alpha})\} \bigr)(\outIn) \preceq \sys{s} \bigl(\tuple{x}_2 \cup \{(\inIn, x_2^{\alpha})\} \bigr)(\outIn) = x_2^{\alpha+1}.
		\end{equation*}
		
		\item Let $\alpha$ be a nonzero limit ordinal and assume $x_1^{\beta} \preceq x_2^{\beta}$ for all $\beta < \alpha$. We then have
		\begin{equation*}
			x_1^{\alpha} = \sup \bigl\{ x_1^{\beta} \bigm\vert \beta < \alpha \bigr\} \preceq \sup \bigl\{ x_2^{\beta} \bigm\vert \beta < \alpha \bigr\} = x_2^{\alpha}.
		\end{equation*}
	\end{enumerate}
	
	Thus, $x_1^{\alpha} \preceq x_2^{\alpha}$ for all ordinals~$\alpha$. It is easy to see (by transfinite induction) that $x_1^{\alpha}$ = $x_1^{\hat{\alpha}_1}$ for all $\alpha \geq \hat{\alpha}_1$, and $x_2^{\alpha} = x_2^{\hat{\alpha}_2}$ for all $\alpha \geq \hat{\alpha}_2$. We can therefore conclude
	\begin{equation*}
		\chosenFixedPoint{\sys{s}}{\inIn}{\outIn}{\tuple{x}_1} = x_1^{\hat{\alpha}_1} = x_1^{\max\{\hat{\alpha}_1, \hat{\alpha}_2\}} \preceq x_2^{\max\{\hat{\alpha}_1, \hat{\alpha}_2\}} = x_2^{\hat{\alpha}_2} = \chosenFixedPoint{\sys{s}}{\inIn}{\outIn}{\tuple{x}_2}.
	\end{equation*}
	Using monotonicity of $\sys{s}$, this implies
	\begin{equation*}
		\connect{\inIn}{\outIn}(\sys{s})(\tuple{x}_1) = \sys{s}\bigl(\tuple{x}_1 \cup \bigl\{ \bigl(\inIn, \chosenFixedPoint{\sys{s}}{\inIn}{\outIn}{\tuple{x}_1} \bigr) \bigr\} \bigr)|_{\outInSet \setminus \{\outIn\}} \preceq \sys{s}\bigl(\tuple{x}_2 \cup \bigl\{ \bigl(\inIn, \chosenFixedPoint{\sys{s}}{\inIn}{\outIn}{\tuple{x}_2} \bigr) \bigr\} \bigr)|_{\outInSet \setminus \{\outIn\}} = \connect{\inIn}{\outIn}(\sys{s})(\tuple{x}_2).
	\end{equation*}
	Hence, $\connect{\inIn}{\outIn}(\sys{s})$ is monotone and therefore $\connect{\inIn}{\outIn}(\sys{s}) \in \monoSysIntSpace{\inInSet \setminus \{\inIn\}}{\outInSet \setminus \{\outIn\}}$.
\end{proof}

\subsection{Continuous Systems}\label{sec:continuousSys}
Let $\inSet$ be a set and let $(\mathcal{X}, \preceq)$ be an $\omega$-CPO. It is easily verified that $\mathcal{X}^{\inInSet}$ is again an $\omega$-CPO for $\inInSet \subseteq \inSet$, and the supremum of an $\omega$-chain in $\mathcal{X}^{\inInSet}$ corresponds to the tuple of the component-wise suprema. We define for finite disjoint~$\inInSet, \outInSet \subseteq \inSet$,
\begin{equation*}
	\contSysIntSpace{\inInSet}{\outInSet} \coloneqq \bigl\{ \sys{s} \colon \mathcal{X}^{\inInSet} \to \mathcal{X}^{\outInSet} \mid \sys{s} \text{ is $\omega$-continuous} \bigr\}. \footnote{We use the symbol~$\contSysSpace$ because continuity is stronger than monotonicity and $\causalSysSpace$ will be used for the system algebra of causal systems in \autoref{sec:causalSys}.}
\end{equation*}
The following lemma can be proven analogously to \autoref{lem:monoSysFixpoint}.
\begin{lemma} \label{lem:contSysFixpoint}
	Let $\inInSet, \outInSet \subseteq \inSet$ be finite disjoint sets, let $\sys{s} \in \contSysIntSpace{\inInSet}{\outInSet}$, and let $\inIn \in \inInSet$, $\outIn \in \outInSet$. Then, for all $\tuple{x} \in \mathcal{X}^{\inInSet \setminus \{\inIn\}}$, there exists a least $x_{\inIn} \in \mathcal{X}$ such that $\sys{s}\bigl(\tuple{x} \cup \{(\inIn, x_{\inIn})\} \bigr)(\outIn) = x_{\inIn}$.
\end{lemma}

\begin{definition}
	We define $\genFunSys{\contSysSpace}{\connectSet}{\fpChoiceFun}$, the \emph{$\inSet$-system algebra of continuous systems over $\mathcal{X}$}, as follows: For $\sys{s} \in \contSysIntSpace{\inInSet}{\outInSet}$, $\connectSet(\sys{s}) \coloneqq \bigl\{ \{\inIn, \outIn \} \mid \inIn \in \inInSet, \outIn \in \outInSet \bigr\}$ and $\chosenFixedPoint{\sys{s}}{\inIn}{\outIn}{\tuple{x}}$ is the least element of $\fixedPoints{\sys{s}}{\inIn}{\outIn}{\tuple{x}}$ for all $\inIn \in \inInSet$, $\outIn \in \outInSet$, and $\tuple{x} \in \mathcal{X}^{\inInSet \setminus \{\inIn\}}$.
\end{definition}

We again need to show that connecting interfaces of a system yields a system in our algebra, i.e., in this case an $\omega$-continuous function. We do so by proving that for an $\omega$-chain $(\tuple{x}_n)_{n \in \omega}$, $\{ \chosenFixedPoint{\sys{s}}{\inIn}{\outIn}{\tuple{x}_n} \mid n \in \omega \}$ is an $\omega$-chain as well, and $\chosenFixedPoint{\sys{s}}{\inIn}{\outIn}{\sup \{\tuple{x}_n \mid n \in \omega \}} = \sup\{ \chosenFixedPoint{\sys{s}}{\inIn}{\outIn}{\tuple{x}_n} \mid n \in \omega \}$. We then use $\omega$-continuity of~$\sys{s}$ to conclude that $\connect{\inIn}{\outIn}(\sys{s})$ is $\omega$-continuous.
\begin{theorem}\label{thm:contIsAlg}
	$(\contSysSpace, \intSetFunct, \parac, \connectSet, \connectFunc) = \genFunSys{\contSysSpace}{\connectSet}{\fpChoiceFun}$ is a functional $\inSet$-system algebra over $\mathcal{X}$.
\end{theorem}
\begin{proof}
	As in the proof of \autoref{thm:monoIsAlg}, it is straightforward to verify that $\contSysSpace$ is closed under parallel composition. To show that $\contSysSpace$ is closed under connecting interfaces, let $\inInSet, \outInSet \subseteq \inSet$ be finite disjoint sets and let $\sys{s} \in \contSysIntSpace{\inInSet}{\outInSet}$. Further let $\inIn \in \inInSet$, $\outIn \in \outInSet$. To show that $\connect{\inIn}{\outIn}(\sys{s}) \in \contSysIntSpace{\inInSet \setminus \{\inIn\}}{\outInSet \setminus \{\outIn\}}$, let $(\tuple{x}_n)_{n \in \omega}$ be an $\omega$-chain in~$\mathcal{X}^{\inInSet \setminus \{\inIn\}}$. We then have $\connect{\inIn}{\outIn}(\sys{s})(\tuple{x}_n) = \sys{s}\bigl(\tuple{x}_n \cup \bigl\{ \bigl(\inIn, \chosenFixedPoint{\sys{s}}{\inIn}{\outIn}{\tuple{x}_n} \bigr) \bigr\} \bigr)|_{\outInSet \setminus \{\outIn\}}$ for all $n \in \omega$, where $\chosenFixedPoint{\sys{s}}{\inIn}{\outIn}{\tuple{x}_n}$ is the least fixed point of the function $x_{\inIn} \mapsto \sys{s}(\tuple{x}_n \cup \{(\inIn, x_{\inIn})\})(\outIn)$. Let $\hat{\tuple{x}} \coloneqq \sup \{\tuple{x}_n \mid n \in \omega \}$. Since $\sys{s}$ is monotone by \autoref{lem:contImplMon}, one can show as in the proof of \autoref{thm:monoIsAlg} that $\chosenFixedPoint{\sys{s}}{\inIn}{\outIn}{\tuple{x}} \preceq \chosenFixedPoint{\sys{s}}{\inIn}{\outIn}{\tuple{x}'}$ for all $\tuple{x}, \tuple{x}' \in \mathcal{X}^{\inInSet \setminus \{\inIn\}}$ with $\tuple{x} \preceq \tuple{x}'$. Thus, $\{ \chosenFixedPoint{\sys{s}}{\inIn}{\outIn}{\tuple{x}_n} \mid n \in \omega \}$ is an $\omega$-chain. Because we also have $\chosenFixedPoint{\sys{s}}{\inIn}{\outIn}{\tuple{x}_n} \preceq \chosenFixedPoint{\sys{s}}{\inIn}{\outIn}{\hat{\tuple{x}}}$ for all $n \in \omega$, we obtain $\sup \{ \chosenFixedPoint{\sys{s}}{\inIn}{\outIn}{\tuple{x}_n} \mid n \in \omega \} \preceq \chosenFixedPoint{\sys{s}}{\inIn}{\outIn}{\hat{\tuple{x}}}$. Using $\omega$-continuity of~$\sys{s}$, we further have
	\begin{align*}
		\sys{s}\Bigl(\hat{\tuple{x}} \cup \Bigl\{ \bigl(\inIn, \sup \bigl\{ \chosenFixedPoint{\sys{s}}{\inIn}{\outIn}{\tuple{x}_n} \bigm\vert n \in \omega \bigr\} \bigr) \Bigr\} \Bigr) (\outIn) &= \sys{s}\Bigl(\sup \Bigl\{\tuple{x}_n \cup \bigl\{ \bigl(\inIn, \chosenFixedPoint{\sys{s}}{\inIn}{\outIn}{\tuple{x}_n} \bigr) \bigr\} \Bigm\vert n \in \omega \Bigr\} \Bigr)(\outIn) \\
		&= \sup \Bigl\{ \sys{s}\Bigl(\tuple{x}_n \cup \bigl\{ \bigl(\inIn, \chosenFixedPoint{\sys{s}}{\inIn}{\outIn}{\tuple{x}_n} \bigr) \bigr\} \Bigr)(\outIn) \Bigm\vert n \in \omega \Bigr\} \\
		&= \sup \{ \chosenFixedPoint{\sys{s}}{\inIn}{\outIn}{\tuple{x}_n} \mid n \in \omega \}.
	\end{align*}
	This implies that $\sup \{ \chosenFixedPoint{\sys{s}}{\inIn}{\outIn}{\tuple{x}_n} \mid n \in \omega \}$ is a fixed point of $x_{\inIn} \mapsto \sys{s}\bigl(\hat{\tuple{x}} \cup \{(\inIn, x_{\inIn})\} \bigr)(\outIn)$. Since $\chosenFixedPoint{\sys{s}}{\inIn}{\outIn}{\hat{\tuple{x}}}$ is the least fixed point of that function, we have $\chosenFixedPoint{\sys{s}}{\inIn}{\outIn}{\hat{\tuple{x}}} \preceq \sup \{ \chosenFixedPoint{\sys{s}}{\inIn}{\outIn}{\tuple{x}_n} \mid n \in \omega \}$, and thus $\chosenFixedPoint{\sys{s}}{\inIn}{\outIn}{\hat{\tuple{x}}} = \sup\{ \chosenFixedPoint{\sys{s}}{\inIn}{\outIn}{\tuple{x}_n} \mid n \in \omega \}$. Together with $\omega$-continuity of~$\sys{s}$, this implies	
	\begin{align*}
		\connect{\inIn}{\outIn}(\sys{s})(\sup\{\tuple{x}_n \mid n \in \omega\}) &= \connect{\inIn}{\outIn}(\sys{s})\bigl(\hat{\tuple{x}} \bigr) = \sys{s}\bigl(\hat{\tuple{x}} \cup \bigl\{ \bigl(\inIn, \chosenFixedPoint{\sys{s}}{\inIn}{\outIn}{\hat{\tuple{x}}} \bigr) \bigr\} \bigr)|_{\outInSet \setminus \{\outIn\}} \\
		&= \sys{s}\Bigl(\sup \Bigl\{\tuple{x}_n \cup \bigl\{ \bigl(\inIn, \chosenFixedPoint{\sys{s}}{\inIn}{\outIn}{\tuple{x}_n} \bigr) \bigr\} \Bigm\vert n \in \omega \Bigr\} \Bigr)|_{\outInSet \setminus \{\outIn\}} \\
		&= \sup \Bigl\{ \sys{s}\bigl(\tuple{x}_n \cup \bigl\{ \bigl(\inIn, \chosenFixedPoint{\sys{s}}{\inIn}{\outIn}{\tuple{x}_n} \bigr) \bigr\} \bigr)|_{\outInSet \setminus \{\outIn\}} \Bigm\vert n \in \omega \Bigr\} \\
		&= \sup \bigl\{ \connect{\inIn}{\outIn}(\sys{s})(\tuple{x}_n) \bigm\vert n \in \omega \bigr\}.
	\end{align*}
	We conclude that $\connect{\inIn}{\outIn}(\sys{s})$ is $\omega$-continuous and therefore $\connect{\inIn}{\outIn}(\sys{s}) \in \contSysIntSpace{\inInSet \setminus \{\inIn\}}{\outInSet \setminus \{\outIn\}}$.
\end{proof}

\subsection{Composition-Order Invariance of Monotone and Continuous Systems}\label{sec:monoCompInv}
We want to show that the system algebras of monotone and continuous systems are composition-order invariant. We can do so if we additionally require that all nonempty chains in~$\mathcal{X}$ have an infimum. Before we prove the main result of this section, we need the following lemma about fixed points (cf.~\cite[(A.6) in Chapter~I, §2.6]{GD03}).

\begin{lemma}\label{lem:smallerFP}
	Let $(\mathcal{X}, \preceq)$ be a poset such that every nonempty chain in $\mathcal{X}$ has an infimum. Let $f \colon \mathcal{X} \to \mathcal{X}$ be a monotone function and let $x \in \mathcal{X}$ such that $f(x) \preceq x$. Then there exists $\hat{x} \in \mathcal{X}$ such that $f(\hat{x}) = \hat{x} \preceq x$.
\end{lemma}
\begin{proof}
	Let $x_0 \coloneqq x$, $x_{\alpha+1} \coloneqq f(x_{\alpha})$ for any ordinal~$\alpha$, and $x_{\alpha} \coloneqq \inf \{ x_{\beta} \mid \beta < \alpha \}$ for nonzero limit ordinals~$\alpha$ (where we show below that this infimum exists). We claim that we have for all ordinals~$\alpha$
	\begin{equation*}
		\forall \beta \leq \alpha \ \ f(x_{\alpha}) \preceq x_{\alpha} \preceq x_{\beta}.
	\end{equation*}
	We prove this claim by transfinite induction over~$\alpha$. For $\alpha = 0$, the claim follows since $f(x) \preceq x$. Assume the claim is true for some ordinal~$\alpha$. We then have $x_{\alpha+1} = f(x_{\alpha}) \preceq x_{\alpha}$. For $\beta \leq \alpha$, we therefore have by assumption $x_{\alpha+1} \preceq x_{\alpha} \preceq x_{\beta}$. For $\beta = \alpha + 1$, we trivially have $x_{\alpha+1} \preceq x_{\beta}$. Since~$f$ is monotone, we thus have for all $\beta \leq \alpha + 1$,
	\begin{equation*}
		f(x_{\alpha+1}) \preceq f(x_{\alpha}) = x_{\alpha+1} \preceq x_{\beta}.
	\end{equation*}
	Now let $\alpha$ be a nonzero limit ordinal and assume that for all $\alpha' < \alpha$, we have $f(x_{\alpha'}) \preceq x_{\alpha'} \preceq x_{\beta}$ for all $\beta \leq \alpha'$. This implies that $\{ x_{\beta} \mid \beta < \alpha \}$ is a chain and therefore the infimum exists. We then have by definition of $x_{\alpha}$ that $x_{\alpha} \preceq x_{\beta}$ for all $\beta \leq \alpha$. By monotonicity of~$f$, we further have for all $\alpha' < \alpha$,
	\begin{equation*}
		f(x_\alpha) = f(\inf \{ x_{\beta} \mid \beta < \alpha \}) \preceq f(x_{\alpha'}) \preceq x_{\alpha'}.
	\end{equation*}
	Hence, $f(x_\alpha)$ is a lower bound of $\{ x_{\beta} \mid \beta < \alpha \}$, and therefore $f(x_\alpha) \preceq \inf \{ x_{\beta} \mid \beta < \alpha \} = x_{\alpha} \preceq x_{\beta}$ for all $\beta \leq \alpha$. Thus, the claim holds for all ordinals~$\alpha$.
	
	Our claim and \autoref{lem:chainFixHartogs} applied to the poset~$(\mathcal{X}, \succeq)$ yield that there exists an ordinal~$\hat{\alpha}$ such that $x_{\hat{\alpha}} = x_{\hat{\alpha}+1}$. This implies $f(x_{\hat{\alpha}}) = x_{\hat{\alpha}+1} = x_{\hat{\alpha}} \preceq x_0 = x$.
\end{proof}

We now use this lemma to prove composition-order invariance. For showing connection-order invariance, we consider a system~$\sys{s}$ and connecting interfaces~$\inIn$ to $\outIn$, and $\inIn'$ to $\outIn'$. For an input~$\tuple{x}$ at the remaining interfaces, consider the function~$f \colon \mathcal{X}^2 \to \mathcal{X}^2$ that maps $(x_0, x_1)$ to the outputs of $\sys{s}$ at interfaces $\outIn$ and $\outIn'$ on input $\tuple{x} \cup \bigl\{ \bigl(\inIn, x_0 \bigr), \bigl(\inIn', x_1 \bigr) \bigr\}$. When we first connect $\inIn$ to $\outIn$, and then $\inIn'$ to $\outIn'$, the system algebra chooses fixed points $x_{\inIn}$ and $x_{\inIn'}$ for the two connections, where $(x_{\inIn}, x_{\inIn'})$ is a fixed point of~$f$. We show that it is the least fixed point, and that the same is true for connecting the interfaces in the other order, which implies that the two resulting systems are equal.

Note that $(x_{\inIn}, x_{\inIn'})$ being the least fixed point of $f$ is not trivial: We know by definition of the connection operation that $x_{\inIn}$ is the least fixed point of $f( \cdot, x_{\inIn'} )(0)$, but for some~$x_1 \prec x_{\inIn'}$, the function $f( \cdot, x_1)(0)$ could have a fixed point $x_0 \prec x_{\inIn}$ such that $(x_0, x_1)$ is a fixed point of $f$ that does not get chosen for the connections since $x_0$ is not the least fixed point of $f( \cdot, x_1)(0)$. Our proof proceeds as follows: For the least fixed point~$\bigl(\hat{x}_0, \hat{x}_1 \bigr)$ of~$f$, we prove that $\hat{x}_0$ is the least fixed point of $f( \cdot, \hat{x}_1)(0)$ by showing that if there were a fixed point less than~$\hat{x}_0$, we could use the lemma above to obtain a fixed point of~$f$ less than $\bigl(\hat{x}_0, \hat{x}_1 \bigr)$. We then use this to conclude that $(x_\inIn, x_{\inIn'}) = (\hat{x}_0, \hat{x}_1)$.
\begin{theorem}\label{thm:monoCompInv}
	Let $\inSet$ be a set, let $(\mathcal{X}, \preceq)$ be a CPO and let $(\sysSpace, \intSetFunct, \parac, \connectSet, \connectFunc) = \genFunSys{\monoSysSpace}{\connectSet}{\fpChoiceFun}$ be the $\inSet$-system algebra of monotone systems over $\mathcal{X}$, or let $(\mathcal{X}, \preceq)$ be an $\omega$-CPO and let $(\sysSpace, \intSetFunct, \parac, \connectSet, \connectFunc) = \genFunSys{\contSysSpace}{\connectSet}{\fpChoiceFun}$ be the $\inSet$-system algebra of continuous systems over~$\mathcal{X}$. If every nonempty chain in~$\mathcal{X}$ has an infimum, the system algebra is composition-order invariant.
\end{theorem}
\begin{proof}
	Since the least fixed points are chosen, $\genFunSys{\monoSysSpace}{\connectSet}{\fpChoiceFun}$ and $\genFunSys{\contSysSpace}{\connectSet}{\fpChoiceFun}$ have consistently chosen fixed points. Hence, \autoref{lem:coninvCompinvCons} implies that we only have to prove connection-order invariance. Let $\inInSet, \outInSet \subseteq \inSet$ be finite disjoint sets and let $\sys{s} \in \sysIntSpace{\inInSet}{\outInSet}$. Further let $\inIn, \inIn' \in \inInSet$ and $\outIn, \outIn' \in \outInSet$ be distinct elements. We have to show that $\connect{\inIn'}{\outIn'}\bigl( \connect{\inIn}{\outIn}(\sys{s}) \bigr) = \connect{\inIn}{\outIn}\bigl( \connect{\inIn'}{\outIn'}(\sys{s}) \bigr)$. To this end, let $\tuple{x} \in \mathcal{X}^{\inInSet \setminus \{\inIn, \inIn'\}}$. By \autoref{def:funSysAlg}, there exist $x_{\inIn}, x_{\inIn'}, x'_{\inIn}, x'_{\inIn'} \in \mathcal{X}$ such that
	\begin{align*}
		\connect{\inIn'}{\outIn'}\bigl( \connect{\inIn}{\outIn}(\sys{s})\bigr) (\tuple{x}) &= \sys{s}\bigl(\tuple{x} \cup \bigl\{ \bigl(\inIn, x_{\inIn} \bigr), \bigl(\inIn', x_{\inIn'} \bigr) \bigr\} \bigr)|_{\outInSet \setminus \{\outIn, \outIn'\}}, \\
		\connect{\inIn}{\outIn}\bigl( \connect{\inIn'}{\outIn'}(\sys{s})\bigr) (\tuple{x}) &= \sys{s}\bigl(\tuple{x} \cup \bigl\{ \bigl(\inIn, x'_{\inIn} \bigr), \bigl(\inIn', x'_{\inIn'} \bigr) \bigr\} \bigr)|_{\outInSet \setminus \{\outIn, \outIn'\}}.
	\end{align*}
	Let $f \colon \mathcal{X}^2 \to \mathcal{X}^2$ be the function with 
	\begin{equation*}
		f(x_0, x_1) = \Bigl( \sys{s}\bigl(\tuple{x} \cup \bigl\{ \bigl(\inIn, x_0 \bigr), \bigl(\inIn', x_1 \bigr) \bigr\} \bigr) (\outIn), \ \sys{s}\bigl(\tuple{x} \cup \bigl\{ \bigl(\inIn, x_0 \bigr), \bigl(\inIn', x_1 \bigr) \bigr\} \bigr) (\outIn') \Bigr).
	\end{equation*}
	We then have by definition of~$\fpChoiceFun$ and $\connectFunc$ that $x_{\inIn}$ is the least fixed point of $f( \cdot, x_{\inIn'} )(0)$ and that $x'_{\inIn'}$ is the least fixed point of $f( x'_{\inIn}, \cdot)(1)$. We further have that $(x_\inIn, x_{\inIn'})$ and $(x'_\inIn, x'_{\inIn'})$ are fixed points of~$f$. We will show that $(x_\inIn, x_{\inIn'})$ and $(x'_\inIn, x'_{\inIn'})$ are the least fixed point of~$f$ and therefore equal.
	
	Note that if $(\mathcal{X}, \preceq)$ is a CPO or an $\omega$-CPO, then $(\mathcal{X}^2, \preceq)$ is a CPO or an $\omega$-CPO, respectively. Moreover, if~$\sys{s}$ is monotone, so is $f$, and if $\sys{s}$ is $\omega$-continuous, so is $f$; in both cases, $f$ is monotone by \autoref{lem:contImplMon}. Hence, $f$ has a least fixed point~$\bigl(\hat{x}_0, \hat{x}_1 \bigr)$ by \autoref{thm:monotoneFP} and \autoref{thm:continuousFP}. This implies that $\hat{x}_0$ is a fixed point of $f( \cdot, \hat{x}_1)(0)$. We claim that it is the least fixed point of that function. Assume toward a contradiction that $\tilde{x}_0 \prec \hat{x}_0$ is the least fixed point of $f(\cdot, \hat{x}_1)(0)$. We then have by monotonicity of~$f$ for some $y_1 \in \mathcal{X}$,
	\begin{equation*}
		(\tilde{x}_0, y_1) = f( \tilde{x}_0, \hat{x}_1 ) \preceq f( \hat{x}_0, \hat{x}_1 ) = ( \hat{x}_0, \hat{x}_1 ).
	\end{equation*}
	Therefore, $f( \tilde{x}_0, \hat{x}_1 ) = (\tilde{x}_0, y_1 ) \preceq ( \tilde{x}_0, \hat{x}_1 )$. Thus, \autoref{lem:smallerFP} implies that $f$ has a fixed point~$(\bar{x}_0, \bar{x}_1 ) \preceq ( \tilde{x}_0, \hat{x}_1 ) \prec (\hat{x}_0, \hat{x}_1 )$, which contradicts $(\hat{x}_0, \hat{x}_1 )$ being the least fixed point of~$f$.
	
	Since $\hat{x}_0$ is the least fixed point of $f ( \cdot, \hat{x}_1 )(0)$ and $(\hat{x}_0, \hat{x}_1 )$ is a fixed point of~$f$, we have $\hat{x}_1 \in \fixedPoints{\connect{\inIn}{\outIn}(\sys{s})}{\inIn'}{\outIn'}{\tuple{x}}$. We also have that $x_{\inIn'}$ is the least element of $\fixedPoints{\connect{\inIn}{\outIn}(\sys{s})}{\inIn'}{\outIn'}{\tuple{x}}$ and $\hat{x}_1 \preceq x_{\inIn'}$. This implies $\hat{x}_1 = x_{\inIn'}$. As we have seen before, $x_{\inIn}$ is the least fixed point of $f( \cdot, x_{\inIn'} )(0)$ and $\hat{x}_0$ is the least fixed point of $f( \cdot, \hat{x}_1)(0)$. Thus, $x_{\inIn} = \hat{x}_0$. An analogous argument shows that $(x'_\inIn, x'_{\inIn'} ) = ( \hat{x}_0, \hat{x}_1 ) = (x_\inIn, x_{\inIn'})$. We therefore have $\connect{\inIn'}{\outIn'}\bigl( \connect{\inIn}{\outIn}(\sys{s})\bigr) (\tuple{x}) = \connect{\inIn}{\outIn}\bigl( \connect{\inIn'}{\outIn'}(\sys{s})\bigr) (\tuple{x})$.
\end{proof}

\subsection{Kahn Networks}\label{sec:Kahn}
Kahn networks~\cite{Kahn74} are an instantiation of the system algebra of continuous systems. There, the set~$\mathcal{X}$ consists of sequences over some set of values~$\valsp$. The interpretation is that the systems successively take inputs from $\valsp$ at each interface and produce outputs in $\valsp$. An input $(v_1, v_2, v_3) \in \mathcal{X}$ to an interface then corresponds to the input history where $v_1$ was input first, then $v_2$, and finally~$v_3$. Hence, in this model the order of inputs at each interface is relevant, but the order of inputs at different interfaces is not modeled.

More formally, let $\valsp$ be some nonempty set of values, let
\begin{equation*}
	\mathcal{X} \coloneqq \bigl\{ (\val_{\alpha})_{\alpha \in \beta} \in \valsp^{\beta} \mid \beta \leq \omega \bigr\},
\end{equation*}
and let $\sqsubseteq$ be the initial segment relation on~$\mathcal{X}$, i.e,
\begin{equation*}
	\bigl(\val^1_{\alpha} \bigr)_{\alpha \in \beta_1} \sqsubseteq \bigl(\val^2_{\alpha} \bigr)_{\alpha \in \beta_2} \ \vcentcolon\Longleftrightarrow \ \beta_1 \leq \beta_2 \ \wedge \ \forall \alpha \in \beta_1 \ \bigl(\val^1_{\alpha} = \val^2_{\alpha} \bigr).
\end{equation*}
It is straightforward to verify that $(\mathcal{X}, \sqsubseteq)$ is a CPO: For a chain~$\mathcal{C} \subseteq \mathcal{X}$, the supremum of $\mathcal{C}$ is the sequence~$v$ in which for $\alpha \leq \omega$, $v_{\alpha}$ is defined if and only if there is a sequence~$\val'$ in $\mathcal{C}$ in which $v'_{\alpha}$ is defined, and in this case, $v_{\alpha} = v'_{\alpha}$. This is well-defined since all elements in $\mathcal{C}$ are comparable and thus, if two sequences in $\mathcal{C}$ define $v'_{\alpha}$ and $\tilde{v}_{\alpha}$, we have $v'_{\alpha} = \tilde{v}_{\alpha}$. One can therefore consider the system algebra of continuous systems over $\mathcal{X}$. Such systems correspond to Kahn networks~\cite{Kahn74}. By \autoref{thm:monoCompInv}, the system algebra is composition-order invariant since every chain also has an infimum (replace ``there is a sequence in~$\mathcal{C}$'' by ``for all sequences in~$\mathcal{C}$'' in the argument for the supremum).

\paragraph{Example\ifelsarticle\else.\fi}
As a simple example of a system in this model, consider the system~$\sys{s}$ with input interface~$\inIn$ and output interfaces $\outIn_1, \outIn_2$ that forwards all inputs to both $\outIn_1$ and $\outIn_2$. Formally, the system is described by the function~$\sys{s} \colon \mathcal{X}^{\{\inIn\}} \to \mathcal{X}^{\{\outIn_1, \outIn_2\}}$ with
\begin{equation*}
	\sys{s}(\tuple{x})(\outIn_1) = \sys{s}(\tuple{x})(\outIn_2) = \tuple{x}(\inIn).
\end{equation*}
This function is clearly $\omega$-continuous and therefore~$\sys{s}$ is a Kahn network. We can now connect interfaces~$\inIn$ and $\outIn_1$ to obtain the system~$\connect{\inIn}{\outIn_1}(\sys{s})$ with no input interfaces and a single output interface~$\outIn_2$. To determine the output at that interface, we have to find the least fixed point of the function $x_{\inIn} \mapsto \sys{s}(\{(\inIn, x_{\inIn})\})(\outIn_1) = x_{\inIn}$. Since all $x \in \mathcal{X}$ are fixed points of this function, the least fixed point is $\min \mathcal{X}$, i.e., the empty sequence. This also matches the intuition that $\connect{\inIn}{\outIn_1}(\sys{s})$ should not output anything because every other output would have to be caused by itself.

\paragraph{Generalizations\ifelsarticle\else.\fi}
Kahn only considered $\omega$-continuous systems and argued that continuity in contrast to monotonicity prevents a system from only outputting a value after receiving infinitely many inputs \cite{Kahn74}. Note, however, that allowing all monotone systems also yields a valid composition-order invariant system algebra.

Another generalization is to not only consider finite and infinite sequences, but also include transfinite sequences of inputs and outputs. By doing so, one can model systems that output something after infinitely many other outputs. This makes sense if one considers systems that can produce infinitely many outputs in finite time. Conversely, one can also place a finite upper bound on the number of inputs and outputs per interface, i.e., $\mathcal{X} \coloneqq \bigl\{ (\val_{\alpha})_{\alpha \in \beta} \in \valsp^{\beta} \mid \beta \leq n \bigr\}$ for some $n \in \N$. The special case $n = 1$ has been considered by Tackmann~\cite{Tac14}. Note, however, that it is not possible in this model to allow arbitrarily many but only finitely many inputs and outputs, because this does not yield an $\omega$-CPO: The supremum of an $\omega$-chain $C_0 \sqsubseteq C_1 \sqsubseteq C_2 \sqsubseteq \ldots$, where $C_n$ contains $n$ elements, is infinite.

\section{The System Algebra of Causal Systems}\label{sec:causalSys}
\subsection{Causal Systems}
We now develop a system algebra where the inputs and outputs of systems are (partially) ordered, i.e., a single input or output of a system is an element of some poset~$(P, \preceq)$. An input or output history at some interface is then a subset of~$P$. We interpret~$p_1 \prec p_2$ as $p_2$ \emph{could} causally depend on $p_1$, but not vice versa. This can (but does not need to) mean that $p_1$ occurred at an earlier time than $p_2$. For the ease of presentation, we will say that $p_1$ is before $p_2$ and $p_2$ is after $p_1$ if $p_1 \prec p_2$, keeping in mind that $\preceq$ does not necessarily correspond to a relation on time. As a natural instantiation in which time is made explicit and $\prec$ does mean at an earlier time, consider $P = \valsp \times \timesp$ for some set~$\valsp$ and a poset~$(\timesp, \leq)$ and let
\begin{equation}\label{eq:valtimPO}
	(\val_1, \tim_1) \preceq (\val_2, \tim_2) \ \ \vcentcolon\Longleftrightarrow \ \ \tim_1 < \tim_2 \ \vee \ (\val_1, \tim_1) = (\val_2, \tim_2).
\end{equation}
We can then interpret~$(\val, \tim)$ as the value~$\val$ being input or output at time~$\tim$.

We again fix some arbitrary set~$\inSet$ of interface labels and define a functional system algebra where a system is a function that for each input interface takes an input history as an input and outputs the output history for each output interface. That is, we define a functional system algebra over~$\mathcal{X}$, where $\mathcal{X}$ consists of subsets of~$P$. We only allow well-ordered subsets, which means that at each input interface, if there are inputs, there is a first one, and one input is given after another. Inputs at different interfaces (which can be connected to different systems) can however be incomparable (e.g., occur at the same time). Only considering well-ordered inputs per interface seems to be a reasonable restriction if systems in computer science are to be modeled, because such systems are typically discrete and are started at some point, i.e., there is a first output. This restriction is discussed further below. Formally, we define
\begin{equation*}
	\mathcal{X} \coloneqq \{X \subseteq P \mid (X, \preceq) \text { is well-ordered}\}.
\end{equation*}

The systems we consider respect causality, i.e., an output can only depend on inputs that are given before that output. We formalize this by requiring that for every change in the output history, there is an earlier change in the input history.
\begin{definition}\label{def:causalSys}
	For finite disjoint $\inInSet, \outInSet \subseteq \inSet$, a \emph{causal $(\inInSet, \outInSet)$-system over $\mathcal{X}$} is a function $\sys{s} \colon \mathcal{X}^{\inInSet} \to \mathcal{X}^{\outInSet}$ such that
	\begin{equation*}
		\forall \tuple{X}, \tuple{X}' \in \mathcal{X}^{\inInSet} \ \forall \outIn \in \outInSet \ \forall y \in \bigl(\sys{s}(\tuple{X})(\outIn) \symdif \sys{s}(\tuple{X}')(\outIn) \bigr) \ \exists \inIn \in \inInSet \ \exists x \in \bigl(\tuple{X}(\inIn) \symdif \tuple{X}'(\inIn) \bigr) \ \ x \prec y.
	\end{equation*}
	We denote the set of causal $(\inInSet, \outInSet)$-systems over~$\mathcal{X}$ by $\causalSysIntSpace{\inInSet}{\outInSet}$.
\end{definition}
As we show in \autoref{app:equivalence-causal}, our definition corresponds to a generalization of strict causality as considered by Matsikoudis and Lee~\cite{ML15} to systems with several interfaces.

We define the relation~$\sqsubseteq$ on $\mathcal{X}$ such that $X_1 \sqsubseteq X_2$ if $X_1$ is an \emph{initial segment} of $X_2$, i.e., $X_2$ contains all elements from~$X_1$ and all additional elements are greater than all elements in~$X_1$:
\begin{equation*}
	X_1 \sqsubseteq X_2 \ \vcentcolon\Longleftrightarrow \ X_1 \subseteq X_2 \ \wedge \ \forall x_1 \in X_1 \ \forall x_2 \in X_2 \setminus X_1 \ (x_1 \prec x_2).
\end{equation*}
It is easy to verify that $(\mathcal{X}, \sqsubseteq)$ is a poset.


\paragraph{Examples and comparison to Kahn networks\ifelsarticle\else.\fi}
In this paragraph, we consider the intuitive special case $P = \valsp \times \R$ with the partial order as defined in \eqref{eq:valtimPO}. In contrast to Kahn networks, causal systems for this choice of~$P$ directly allow modeling systems that produce outputs depending on the precise time of the inputs. As we will see below, there are also more subtle differences.

If we ignore the concrete times and only keep the order of the elements in an input or output $X \in \mathcal{X}$, we obtain a sequence of values in $\valsp$, i.e., an element of the domain of Kahn networks. Moreover, the relation~$\sqsubseteq$ on~$\mathcal{X}$ then corresponds to the relation~$\sqsubseteq$ on sequences we defined for Kahn networks. However, causality does not imply that the function is monotone with respect to~$\sqsubseteq$. To see this, consider the function~$\sys{s} \colon \mathcal{X}^{\{\inIn\}} \to \mathcal{X}^{\{\outIn\}}$ with
\begin{equation*}
	\sys{s}(\tuple{X})(\outIn) = \begin{cases}
		\{ (1, 1) \}, &(1, 0) \in \tuple{X}(\inIn), \\
		\{ (0, 1) \}, &(1, 0) \notin \tuple{X}(\inIn).
	\end{cases}
\end{equation*}
This function corresponds to a system with one input interface~$\inIn$ and one output interface~$\outIn$ that outputs at time~$1$ a bit indicating whether $1$ has been input at time~$0$. This is a causal system with respect to \autoref{def:causalSys} because outputs can only change at time~$1$, and if they do, the inputs must have changed at time~$0$, which is before time~$1$. On the other hand, $\sys{s}$ is not monotone with respect to~$\sqsubseteq$: We have $\emptyset \sqsubseteq \{ (1, 0) \}$, but on input~$\emptyset$, the system outputs~$\{(0,1)\}$, and on input~$\{ (1, 0) \}$, the system outputs~$\{(1,1)\}$, where $\{(0,1)\} \not\sqsubseteq \{(1,1)\}$.

Note that our definition of causality requires that for every change in the outputs, there is a change in the inputs \emph{strictly} before that. Including systems that produce outputs without any delay makes it impossible to allow arbitrary connections: Consider a system that at one interface takes as input a bit~$b$ and outputs $1-b$ at another interface \emph{at the same time}, and outputs there the bit~$0$ at time~$1$ if no inputs are given at that time. If we now connect these two interfaces, one cannot consistently assign a value to that interface.

This restriction means in particular that the function $\mathcal{X}^{\{\inIn\}} \to \mathcal{X}^{\{\outIn_1, \outIn_2\}}$ that maps the input to itself at both output interface is not a causal system, because changing an input at some point in time results in a change of the outputs at the same time. To model a causal system that corresponds to the Kahn network we gave as an example in \autoref{sec:Kahn}, which forwards all inputs to both output interfaces, we therefore need to add some delay to the inputs. For the sake of simplicity, we here add a fixed delay~$\delta > 0$ to all inputs. We obtain the function~$\sys{s} \colon \mathcal{X}^{\{\inIn\}} \to \mathcal{X}^{\{\outIn_1, \outIn_2\}}$ with
\begin{equation*}
	\sys{s}(\tuple{X})(\outIn_1) = \sys{s}(\tuple{X})(\outIn_2) = \{ (\val, \tim + \delta) \mid (\val, \tim) \in \tuple{X}(\inIn) \}.
\end{equation*}
One can easily verify that this function does satisfy \autoref{def:causalSys}. For connecting interfaces~$\inIn$ and $\outIn_1$, $\emptyset$ is the \emph{unique} fixed point: It clearly is a fixed point, and there is no other fixed point since if there is some input, the output is delayed by~$\delta$ and hence different from the input. Therefore, the resulting system produces no output, which is consistent with our analysis of the corresponding Kahn network. Note, however, that in contrast to the Kahn network, we here have a unique fixed point, so the output~$\emptyset$ is not only intuitive, but also formally the only option. We show in \autoref{sec:causalFP} that causality always guarantees unique fixed points.

\paragraph{On the restriction to well-ordered subsets\ifelsarticle\else.\fi}
One cannot in general define a consistent system algebra of causal systems for $\mathcal{X} = \mathcal{P}(P)$. To see this, let $P = \Q$ and let $\preceq$ be the usual order on~$\Q$. For this poset, an output~$\tim \in \Q$ can be interpreted as an output from a unary set at time~$\tim$. Consider a system~$\sys{s}$ with input interface~$\inIn$ and output interfaces~$\outIn$ and $\outIn'$ that for all $n \in \N \setminus \{0\}$ generates an output at time $1/n$ at both output interfaces if and only if it did not receive an input before time~$1/n$. Formally, $\sys{s}(\{(\inIn, X)\}) = \{(\outIn, Y), (\outIn', Y)\}$ with $Y = \{1/n \mid n \in \mathds{N} \setminus \{0\} \wedge \forall \tim \in X \ (\tim \geq 1/n) \}$. This system clearly respects causality, but, e.g., for $X = \emptyset$, the output is not well-ordered because for each output, there is another output before. When we connect interfaces~$\inIn$ and $\outIn$, the resulting system does not have a well-defined behavior: If it outputs something at time $1/n$ at interface $\outIn'$, it did not receive any input before, i.e., it did not output anything at interface~$\outIn$ before. In particular, it did not output anything before time $1/2n$. But then, it would output something at time $1/2n$, a contradiction. Hence, the system cannot output anything at all. But then, it would output something at time~$1$, also a contradiction.
		
This example shows that we at least have to exclude infinite descending chains, i.e., $\mathcal{X}$ can only contain subsets of $P$ on which the ordering is well-founded. This is implied by only considering well-ordered subsets as we do. One could generalize~$\mathcal{X}$ to include non-well-ordered sets that contain no infinite descending chains, as done by Matsikoudis and Lee~\cite{ML15}. A nonempty subset of such sets contains minimal elements, but not necessarily a unique minimum, allowing for incomparable or concurrent inputs and outputs. Since our systems, in contrast to the functions considered by Matsikoudis and Lee, can have multiple interfaces, we do not lose much generality by restricting ourselves to well-ordered sets: At different interfaces, we can still have inputs or outputs at incomparable or equal times. We also find it natural to have several interfaces for potentially simultaneous inputs or outputs because this makes the assumptions on the concurrency explicit. Moreover, well-ordered sets are easier to handle and allow more intuitive proofs.

\subsection{Fixed Points of Causal Functions} \label{sec:causalFP}
To define interface connection for causal systems, we need to prove a result about fixed points. We show that causal functions, which correspond to a single-interface variant of \autoref{def:causalSys}, have a unique fixed point. We also show how this fixed point can be obtained, i.e., we provide a constructive proof. A similar result was proven non-constructively by Naundorf~\cite{Nau00}; Matsikoudis and Lee have provided a constructive proof of a related result~\cite{ML15}. However, we find that our construction is simpler and better applicable in our setting.

The intuition is similar to that of other constructive fixed point theorems: Start with $X_0 = \emptyset$ and then apply the function~$f$ over and over again until a fixed point is reached. In contrast to monotone functions (see \autoref{thm:monotoneFP} and \autoref{thm:continuousFP}), simply setting $X_{\alpha+1} = f(X_{\alpha})$ does not work because an input can prevent future outputs; consider for example a function~$f$ that on input~$\emptyset$ outputs~$\{x_0, x_1\}$ with $x_0 \prec x_1$, and on input~$\{x_0\}$ outputs~$\{x_0\}$. This does not contradict causality but the fixed point~$\{x_0\}$ is not reached if we set $X_{\alpha+1} = f(X_{\alpha})$ and $f(\{x_0, x_1\}) \neq \{x_0\}$. We therefore only add the least new element in $f(X_{\alpha})$ to $X_{\alpha}$ to obtain $X_{\alpha+1}$. This means for the example above that $X_1 = X_0 \cup \{ \min(f(X_0) \setminus X_0) \} = \{x_0\}$. For nonzero limit ordinals~$\alpha$, we set $X_{\alpha} = \bigcup_{\beta < \alpha} X_{\beta}$.
\begin{theorem} \label{thm:causalFP}
	Let $\mathcal{X}$ be as above and $f \colon \mathcal{X} \to \mathcal{X}$ such that
	\begin{equation} \label{eq:fpthmfAss}
		\forall X, X' \in \mathcal{X} \ \forall y \in \bigl(f(X) \symdif f \bigl(X' \bigr) \bigr) \ \exists x \in \bigl(X \symdif X' \bigr) \ \ x \prec y.
	\end{equation}
	Then, $f$ has a unique fixed point. This fixed point equals~$X_{\hat{\alpha}}$ for some ordinal~$\hat{\alpha}$, where
	\begin{alignat*}{2}
		&&X_0 &= \emptyset, \\
		\text{for any ordinal } \alpha, \qquad &&X_{\alpha+1} &= \begin{cases}
			X_{\alpha}, &f(X_{\alpha}) \setminus X_{\alpha} = \emptyset, \\
			X_{\alpha} \cup \{ \min(f(X_{\alpha}) \setminus X_{\alpha}) \}, &\text{otherwise},
		\end{cases} \\
		\text{and for nonzero limit ordinals } \alpha, \qquad &&X_{\alpha} &= \bigcup_{\beta < \alpha} X_{\beta}.
	\end{alignat*}
\end{theorem}
\begin{proof}
	Note that $X_{\alpha+1}$ is well-defined if $X_{\alpha}$ is well-ordered, which is implied by the following claim.
	\begin{claim} \label{claim:fpthmC1}
		For all ordinals~$\alpha$, $X_{\alpha}$ is well-ordered, we have $X_{\beta} \sqsubseteq X_{\alpha}$ for all $\beta \leq \alpha$, and
		\begin{equation*}
			\forall x_1 \in X_{\alpha} \ \forall x_2 \in f(X_{\alpha}) \setminus X_{\alpha} \ (x_1 \prec x_2).
		\end{equation*}
	\end{claim}
	\begin{claimproof}
		The proof is by transfinite induction over~$\alpha$.
		\begin{enumerate}[label=(\roman*)]
			\item For $\alpha = 0$, there is nothing to show since $X_{\alpha} = \emptyset$.
			
			\item Assume the claim holds for some ordinal~$\alpha$. If $X_{\alpha + 1} = X_{\alpha}$, the claim also holds for $\alpha + 1$. Otherwise, $X_{\alpha + 1} = X_{\alpha} \cup \{ \min(f(X_{\alpha}) \setminus X_{\alpha}) \}$. To see that $X_{\alpha+1}$ is well-ordered, let $S \subseteq X_{\alpha+1}$, $S \neq \emptyset$. If $S \cap X_{\alpha} = \emptyset$, $S$ contains only one element, which is then the least element. Otherwise, the least element of~$S$ is $\min(S \cap X_{\alpha})$ because $X_{\alpha+1} \setminus X_{\alpha}$ only contains the element $\min(f(X_{\alpha}) \setminus X_{\alpha})$, and since the claim holds for $\alpha$, this element is greater than all elements in~$X_{\alpha}$.
			
			Since the claim holds for $\alpha$, we have $\forall x_1 \in X_{\alpha} \ \forall x_2 \in f(X_{\alpha}) \setminus X_{\alpha} \ (x_1 \prec x_2)$.
			This implies $X_{\alpha} \sqsubseteq X_{\alpha+1}$. Since $X_{\beta} \sqsubseteq X_{\alpha}$ for all $\beta \leq \alpha$, we also have $X_{\beta} \sqsubseteq X_{\alpha+1}$ for all $\beta \leq \alpha + 1$.
			
			Let $x_1 \in X_{\alpha+1}$ and $x_2 \in f(X_{\alpha+1}) \setminus X_{\alpha + 1}$. We then have $x_1 \preceq \min(f(X_{\alpha}) \setminus X_{\alpha})$. If $x_2 \in f(X_{\alpha})$, then $x_2 \in (f(X_{\alpha}) \setminus X_{\alpha} ) \setminus \{\min(f(X_{\alpha}) \setminus X_{\alpha}) \}$; hence, $x_1 \prec x_2$. Otherwise, $x_2 \in f(X_{\alpha + 1}) \setminus f(X_{\alpha})$. 
			By~\eqref{eq:fpthmfAss}, we then have $\min(f(X_{\alpha}) \setminus X_{\alpha}) \prec x_2$, and since $x_1 \preceq \min(f(X_{\alpha}) \setminus X_{\alpha})$, we also have $x_1 \prec x_2$. Thus, the claim holds for $\alpha + 1$.
		
			\item Now let $\alpha$ be a nonzero limit ordinal and assume the claim holds for all $\beta < \alpha$. To show that $X_{\alpha}$ is well-ordered, let $S \subseteq X_{\alpha}$, $S \neq \emptyset$. Let $\beta$ be the least ordinal such that $S \cap X_{\beta} \neq \emptyset$. By definition of~$X_{\alpha}$, we have $\beta < \alpha$ and therefore, $s \coloneqq \min(S \cap X_{\beta})$ exists. To see that~$s$ is the least element of~$S$, let $x \in S \setminus X_{\beta}$. Then, $x \in X_{\gamma}$ for some $\gamma$ with $\beta < \gamma < \alpha$. This implies $X_{\beta} \sqsubseteq X_{\gamma}$, and therefore $s \prec x$. Hence, $X_{\alpha}$ is well-ordered.
			
			To show that $X_{\beta} \sqsubseteq X_{\alpha}$ for all $\beta \leq \alpha$, note that $X_{\beta} \subseteq X_{\alpha}$ and let $x_1 \in X_{\beta}$, $x_2 \in X_{\alpha} \setminus X_{\beta}$. Then, there exists some ordinal~$\gamma$ such that $x_2 \in X_{\gamma} \setminus X_{\beta}$ and $\beta < \gamma < \alpha$. Because we assume that the claim holds for all ordinals less than $\alpha$, we obtain $X_{\beta} \sqsubseteq X_{\gamma}$. Hence, $x_1 \prec x_2$ and thus $X_{\beta} \sqsubseteq X_{\alpha}$.
		
			Let $x_1 \in X_{\alpha}$ and $x_2 \in f(X_{\alpha}) \setminus X_{\alpha}$. By definition of~$X_{\alpha}$, there exists some $\beta < \alpha$ such that $x_1 \in X_{\beta}$. As shown above, we have $X_{\beta} \sqsubseteq X_{\alpha}$. Therefore, $x_2 \in f(X_{\alpha}) \setminus X_{\beta}$. If $x_2 \in f(X_{\beta})$, we have $x_1 \prec x_2$ since $x_2 \in f(X_{\beta}) \setminus X_{\beta}$ and the claim holds for $\beta$. Otherwise, we have $x_2 \in f(X_{\alpha}) \setminus f(X_{\beta})$. By~\eqref{eq:fpthmfAss}, this implies $\min(X_{\alpha} \setminus X_{\beta}) \prec x_2$. Because $x_1 \in X_{\beta} \sqsubseteq X_{\alpha}$, we conclude $x_1 \prec \min(X_{\alpha} \setminus X_{\beta}) \prec x_2$. Altogether, the claim holds for $X_{\alpha}$.
		\end{enumerate}
		
		By \autoref{thm:transfInd}, the claim holds for all ordinals~$\alpha$.
	\end{claimproof}
	
	We now have by \autoref{claim:fpthmC1} and \autoref{lem:chainFixHartogs} that there exists an ordinal~$\hat{\alpha}$ such that $X_{\hat{\alpha}} = X_{\hat{\alpha}+1}$. Therefore, $f(X_{\hat{\alpha}}) \setminus X_{\hat{\alpha}} = \emptyset$. To conclude that $X_{\hat{\alpha}}$ is a fixed point of~$f$, we show that $X_{\hat{\alpha}} \subseteq f(X_{\hat{\alpha}})$, and hence $f(X_{\hat{\alpha}}) = X_{\hat{\alpha}}$.
	\begin{claim}\label{claim:fpthmC2}
		For all ordinals~$\alpha$, $X_{\alpha} \subseteq f(X_{\alpha})$.
	\end{claim}
	\begin{claimproof}
		We prove the claim by transfinite induction over~$\alpha$.
		\begin{enumerate}[label=(\roman*)]
			\item For $\alpha = 0$, the claim trivially holds.
			
			\item Assume the claim holds for some ordinal~$\alpha$ and let $x \in X_{\alpha+1}$. Then, $x \in f(X_{\alpha})$ because either $x \in X_{\alpha} \subseteq f(X_{\alpha})$ or $x = \min(f(X_{\alpha}) \setminus X_{\alpha}) \in f(X_{\alpha})$. Assume toward a contradiction that $x \notin f(X_{\alpha + 1})$. We then have $x \in f(X_{\alpha}) \setminus f(X_{\alpha+1})$. By~\eqref{eq:fpthmfAss} and the definition of $X_{\alpha + 1}$, this implies $\min(f(X_{\alpha}) \setminus X_{\alpha}) \prec x$. Since $x \in X_{\alpha + 1}$, we either have $x = \min(f(X_{\alpha}) \setminus X_{\alpha})$ or $x \in X_{\alpha}$ and therefore $x \prec \min(f(X_{\alpha}) \setminus X_{\alpha})$ by \autoref{claim:fpthmC1}. In both cases, we obtain $x \preceq \min(f(X_{\alpha}) \setminus X_{\alpha}) \prec x$, a contradiction. Hence, $x \in f(X_{\alpha + 1})$.
			
			\item Let $\alpha$ be a nonzero limit ordinal and assume $X_{\beta} \subseteq f(X_{\beta})$ for all $\beta < \alpha$. Let $x \in X_{\alpha}$. Then, $x \in X_{\beta} \subseteq f(X_{\beta})$ for some $\beta < \alpha$. By \autoref{claim:fpthmC1}, we have $X_{\beta} \sqsubseteq X_{\alpha}$. Assume toward a contradiction that $x \notin f(X_{\alpha})$. Then, \eqref{eq:fpthmfAss} implies $\min(X_{\alpha} \setminus X_{\beta}) \prec x$. However, since $x \in X_{\beta} \sqsubseteq X_{\alpha}$, we also have $x \prec \min(X_{\alpha} \setminus X_{\beta})$, a contradiction. Thus, $X_{\alpha} \subseteq f(X_{\alpha})$. \qedhere
		\end{enumerate}
	\end{claimproof}
	
	To prove uniqueness, let $Y_1$ and $Y_2$ be fixed points of~$f$. If $Y_1 \neq Y_2$, we can assume without loss of generality that $Y_1 \setminus Y_2 \neq \emptyset$ and
	\begin{equation} \label{eq:fpthmunas}
		\forall y \in Y_2 \setminus Y_1 \ \bigl(y \not\prec \min(Y_1 \setminus Y_2) \bigr)
	\end{equation}
	(otherwise, swap the roles of $Y_1$ and $Y_2$). Since $\min(Y_1 \setminus Y_2) \in Y_1 \setminus Y_2 = f(Y_1) \setminus f(Y_2)$, we have by \eqref{eq:fpthmfAss} that there exists some $y \in Y_1 \symdif Y_2$ such that $y \prec \min(Y_1 \setminus Y_2)$. Then, $y \in Y_1 \setminus Y_2$ implies $y \prec y$, and $y \in Y_2 \setminus Y_1$ contradicts \eqref{eq:fpthmunas}. Therefore, we must have $Y_1 = Y_2$. This concludes the proof.
\end{proof}

While for some functions, $\hat{\alpha}$ will be a small finite number, using ordinals beyond~$\omega$ cannot be avoided in general. To illustrate this, and to show that the iteration can still be carried out in this case, consider the following example. Again let $P = \Q$ and let $\preceq$ be the usual order on~$\Q$. Let $f$ be the function that on the empty input outputs $\N$, and for each input~$x \in \Q$, additionally outputs $(x + \lfloor x+1 \rfloor) / 2 \succ x$. That is
\begin{equation*}
	f(X) = \N \cup \{(x + \lfloor x+1 \rfloor) / 2 \mid x \in X \}.
\end{equation*}
We then have
\begin{equation*}
	X_0 = \emptyset, X_1 = \bigl\{\min f(\emptyset) \bigr\} = \{0\}, X_2 = \bigl\{0, \tfrac{1}{2} \bigr\}, X_3 = \bigl\{0, \tfrac{1}{2}, \tfrac{3}{4} \bigr\}, \ldots, X_n = \bigl\{0, \tfrac{1}{2}, \ldots, \tfrac{2^{n-1}-1}{2^{n-1}} \bigr\}, \ldots
\end{equation*}
Taking the first limit, we obtain $X_\omega = \bigcup_{\beta < \omega} X_{\beta} = \bigl\{ \tfrac{2^{n}-1}{2^{n}} \bigm\vert n \in \N \bigr\}$. We then continue as before:
\begin{equation*}
	X_{\omega + 1} = X_\omega \cup \{1\}, X_{\omega + 2} = X_\omega \cup \{1, 1+\tfrac{1}{2}\}, \ldots, X_{\omega + n} = X_\omega \cup \bigl\{1, 1+\tfrac{1}{2}, \ldots, 1+\tfrac{2^{n-1}-1}{2^{n-1}} \bigr\}, \ldots
\end{equation*}
For the next limit we have $X_{\omega + \omega} = X_{\omega \cdot 2} = \bigl\{ m + \tfrac{2^{n}-1}{2^{n}} \bigm\vert m \in \{0,1\}, n \in \N \bigr\}$. Continuing this process, we arrive at
\begin{equation*}
	X_{\omega \cdot \omega} = \bigl\{ m + \tfrac{2^{n}-1}{2^{n}} \bigm\vert m, n \in \N \bigr\}.
\end{equation*}
It is easy to see that $X_{\omega \cdot \omega}$ is a fixed point of~$f$ and therefore $\hat{\alpha} = \omega \cdot \omega$.

If we again interpret an output~$\tim \in \Q$ as an output from a unary set at time~$\tim$, the system corresponding to the function~$f$ after connecting interfaces generates infinitely many outputs in a finite time. This is often called \emph{Zeno behavior}, named after Zeno's paradoxes. Such behavior can be avoided by requiring \emph{delta causality} instead of just causality, which places a lower bound~$\delta$ on the reaction time of a system \cite{LeeSan98,ML15}. While this seems to be a reasonable restriction in practice, it makes the definitions more cumbersome. Also, the system described above does have a well-defined behavior, and the fixed point is what one expects. Having less restrictions on the type of system can also improve mathematical convenience when modeling systems.

\subsection{Defining the System Algebra}
We can apply \autoref{thm:causalFP} to show that causal systems have a unique fixed point for arbitrary interface connections.
\begin{corollary} \label{cor:sysFixpoint}
	Let $\inInSet, \outInSet \subseteq \inSet$ be finite disjoint sets, let $\sys{s} \in \causalSysIntSpace{\inInSet}{\outInSet}$, and let $\inIn \in \inInSet$, $\outIn \in \outInSet$. Then, for all $\tuple{X} \in \mathcal{X}^{\inInSet \setminus \{\inIn\}}$, there exists a unique $X_{\inIn} \in \mathcal{X}$ such that $\sys{s}\bigl(\tuple{X} \cup \{(\inIn, X_{\inIn})\} \bigr)(\outIn) = X_{\inIn}$.
\end{corollary}
\begin{proof}
	For $\tuple{X} \in \mathcal{X}^{\inInSet \setminus \{\inIn\}}$, let $f \colon \mathcal{X} \to \mathcal{X}, X \mapsto \sys{s}(\tuple{X} \cup \{(\inIn, X)\})(\outIn)$ and note that the fixed points of~$f$ precisely correspond to the values $X_{\inIn} \in \mathcal{X}$ such that $\sys{s}\bigl(\tuple{X} \cup \{(\inIn, X_{\inIn})\} \bigr)(\outIn) = X_{\inIn}$. It is therefore sufficient to show that $f$ has a unique fixed point. To verify that $f$ satisfies the condition of \autoref{thm:causalFP}, let $X, X' \in \mathcal{X}$ and let $y \in \bigl(f(X) \symdif f(X') \bigr)$. Then, $y \in \bigl(\sys{s}(\tuple{X} \cup \{(\inIn, X)\})(\outIn) \symdif \sys{s}(\tuple{X} \cup \{(\inIn, X')\})(\outIn) \bigr)$. We therefore have by \autoref{def:causalSys} that there exists $x \in X \symdif X'$ such that $x \prec y$. Hence, \autoref{thm:causalFP} implies that $f$ has a unique fixed point.
\end{proof}

\begin{definition}\label{def:causalSysAlg}
	We define $\genFunSys{\causalSysSpace}{\connectSet}{\fpChoiceFun}$, the \emph{$\inSet$-system algebra of causal systems over $\mathcal{X}$}, as follows: For $\sys{s} \in \causalSysIntSpace{\inInSet}{\outInSet}$, $\connectSet(\sys{s}) \coloneqq \bigl\{ \{\inIn, \outIn \} \mid \inIn \in \inInSet, \outIn \in \outInSet \bigr\}$ and $\chosenFixedPoint{\sys{s}}{\inIn}{\outIn}{\tuple{X}}$ is the unique element of $\fixedPoints{\sys{s}}{\inIn}{\outIn}{\tuple{X}}$ for all $\inIn \in \inInSet$, $\outIn \in \outInSet$, and $\tuple{X} \in \mathcal{X}^{\inInSet \setminus \{\inIn\}}$.
\end{definition}

To show that causal systems form a functional system algebra, we have to prove that connecting interfaces of a causal system again yields a causal system. To this end, we first show that changing an input~$\tuple{X}$ can only change the fixed point after that change in~$\tuple{X}$.
\begin{lemma}\label{lem:fpmindif}
	Let $\inInSet, \outInSet \subseteq \inSet$ be finite disjoint sets, let $\sys{s} \in \causalSysIntSpace{\inInSet}{\outInSet}$, and let $\inIn \in \inInSet$, $\outIn \in \outInSet$. Further let $\tuple{X}, \tuple{X}' \in \mathcal{X}^{\inInSet \setminus \{\inIn\}}$. We then have
	\begin{equation*}
		\forall y \in \chosenFixedPoint{\sys{s}}{\inIn}{\outIn}{\tuple{X}} \symdif \chosenFixedPoint{\sys{s}}{\inIn}{\outIn}{\tuple{X}'} \ \exists \inIn' \in \inInSet \setminus \{\inIn\} \ \exists x \in \tuple{X}(\inIn') \symdif \tuple{X}'(\inIn') \ (x \prec y).
	\end{equation*}
\end{lemma}
\begin{proof}
	Let
	\begin{equation*}
		E \coloneqq \bigl\{ y \in \chosenFixedPoint{\sys{s}}{\inIn}{\outIn}{\tuple{X}} \symdif \chosenFixedPoint{\sys{s}}{\inIn}{\outIn}{\tuple{X}'} \mid \forall \inIn' \in \inInSet \setminus \{\inIn\} \ \forall x \in \tuple{X}(\inIn') \symdif \tuple{X}'(\inIn') \ (x \not\prec y) \bigr\}
	\end{equation*}
	and assume toward a contradiction that $E \neq \emptyset$. Then, there exists $y_0 \in E$ such that $\forall y \in E \ (y \not\prec y_0)$ because $\chosenFixedPoint{\sys{s}}{\inIn}{\outIn}{\tuple{X}}$ and $\chosenFixedPoint{\sys{s}}{\inIn}{\outIn}{\tuple{X}'}$ are well-ordered and $E \subseteq \chosenFixedPoint{\sys{s}}{\inIn}{\outIn}{\tuple{X}} \cup \chosenFixedPoint{\sys{s}}{\inIn}{\outIn}{\tuple{X}'}$. Since $\chosenFixedPoint{\sys{s}}{\inIn}{\outIn}{\tuple{X}} = \sys{s}\bigl(\tuple{X} \cup \{(\inIn, \chosenFixedPoint{\sys{s}}{\inIn}{\outIn}{\tuple{X}})\} \bigr)(\outIn)$ and $\chosenFixedPoint{\sys{s}}{\inIn}{\outIn}{\tuple{X}'} = \sys{s}\bigl(\tuple{X}' \cup \{(\inIn, \chosenFixedPoint{\sys{s}}{\inIn}{\outIn}{\tuple{X}'})\} \bigr)(\outIn)$, \autoref{def:causalSys} implies that there exists $\inIn' \in \inInSet$ and $x \in \bigl(\tuple{X} \cup \{(\inIn, \chosenFixedPoint{\sys{s}}{\inIn}{\outIn}{\tuple{X}})\} \bigr)(\inIn') \symdif \bigl(\tuple{X}' \cup \{(\inIn, \chosenFixedPoint{\sys{s}}{\inIn}{\outIn}{\tuple{X}'})\} \bigr)(\inIn')$ such that $x \prec y_0$. Because $y_0 \in E$, we have $\inIn' = \inIn$, i.e., $x \in \chosenFixedPoint{\sys{s}}{\inIn}{\outIn}{\tuple{X}} \symdif \chosenFixedPoint{\sys{s}}{\inIn}{\outIn}{\tuple{X}'}$. We further have $x \notin E$ since $x \prec y_0$ and $\forall y \in E \ (y \not\prec y_0)$. Thus, there exist $\inIn'' \in \inInSet \setminus \{\inIn\}$ and $x' \in \tuple{X}(\inIn'') \symdif \tuple{X}'(\inIn'')$ such that $x' \prec x$. This implies $x' \prec y_0$, contradicting $y_0 \in E$. Hence, we have $E = \emptyset$.
\end{proof}

\begin{theorem}\label{thm:causalSysAlg}
	$(\causalSysSpace, \intSetFunct, \parac, \connectSet, \connectFunc) = \genFunSys{\causalSysSpace}{\connectSet}{\fpChoiceFun}$ is a composition-order invariant functional $\inSet$-system algebra over $\mathcal{X}$.
\end{theorem}
\begin{proof}
	Note that $\causalSysSpace$ is closed under parallel composition since the causality condition in \autoref{def:causalSys} is placed on each output interface separately. To show that $\causalSysSpace$ is also closed under connecting interfaces, let $\inInSet, \outInSet \subseteq \inSet$ be finite disjoint sets, let $\sys{s} \in \causalSysIntSpace{\inInSet}{\outInSet}$, $\inIn \in \inInSet$, and let $\outIn \in \outInSet$. Further let $\tuple{X}, \tuple{X}' \in \mathcal{X}^{\inInSet \setminus \{\inIn\}}$, $\outIn' \in \outInSet \setminus \{\outIn\}$, and $y \in \bigl(\connect{\inIn}{\outIn}(\sys{s})(\tuple{X})(\outIn') \symdif \connect{\inIn}{\outIn}(\sys{s})(\tuple{X}')(\outIn') \bigr)$. We have to show that $\exists \inIn' \in \inInSet \setminus \{\inIn\} \ \exists x \in (\tuple{X}(\inIn') \symdif \tuple{X}'(\inIn')) \ x \prec y$. By the definition of $\connect{\inIn}{\outIn}(\sys{s})$ and \autoref{def:causalSys}, we have
	\begin{equation*}
		\exists \inIn' \in \inInSet \ \exists x \in \bigl((\tuple{X} \cup \{(\inIn, \chosenFixedPoint{\sys{s}}{\inIn}{\outIn}{\tuple{X}})\})(\inIn') \symdif (\tuple{X}' \cup \{(\inIn, \chosenFixedPoint{\sys{s}}{\inIn}{\outIn}{\tuple{X}'})\})(\inIn') \bigr) \ \ x \prec y.
	\end{equation*}
	If $\inIn' \neq \inIn$, we are done. Otherwise, let $x \in \chosenFixedPoint{\sys{s}}{\inIn}{\outIn}{\tuple{X}} \symdif \chosenFixedPoint{\sys{s}}{\inIn}{\outIn}{\tuple{X}'}$ such that $x \prec y$. \autoref{lem:fpmindif} then implies that $\exists \inIn'' \in \inInSet \setminus \{\inIn\} \ \exists x' \in \bigl(\tuple{X}(\inIn'') \symdif \tuple{X}'(\inIn'') \bigr) \ x' \prec x \prec y$. Hence, we have $\connect{\inIn}{\outIn}(\sys{s}) \in \causalSysSpace$. This shows that $\genFunSys{\causalSysSpace}{\connectSet}{\fpChoiceFun}$ is a functional system algebra. Using \autoref{cor:sysFixpoint} and \autoref{lem:unifpCompinv}, we conclude that it is composition-order invariant.
\end{proof}

\section{Conclusion and Future Work}
We have introduced the concept of a \emph{system algebra}, which captures the composition of systems as algebraic operations, and identified \emph{composition-order invariance} as an important property of a system algebra. We have then introduced \emph{functional system algebras}, which consist of systems that correspond to functions mapping inputs to outputs. As instantiations of functional system algebras, we have considered monotone and continuous systems, where Kahn networks are an important special case of them. We have shown their composition-order invariance, providing further insights into these well-studied systems. We have finally introduced the system algebra of causal systems, which allows to model systems that depend on time.

Future work includes using the system algebra of causal systems to analyze protocols that deal with time such as clock synchronization protocols. We are further developing a theory of probabilistic system algebras that can be used to describe and study cryptographic protocols.

\paragraph{Acknowledgments\ifelsarticle\else.\fi}
\addcontentsline{toc}{section}{Acknowledgments}
Ueli Maurer was supported by the Swiss National Science Foundation (SNF), project No.~200020-132794. Christopher Portmann and Renato Renner are supported by the European Commission FP7 Project RAQUEL (grant No.~323970), US Air Force Office of Scientific Research (AFOSR) via grant~FA9550-16-1-0245, the Swiss National Science Foundation (via the National Centre of Competence in Research `Quantum Science and Technology'), and the European Research Council -- ERC (grant No.~258932). Björn Tackmann was supported by the Swiss National Science Foundation (SNF) via Fellowship No.~P2EZP2\_155566 and in part by the NSF grants CNS-1228890 and CNS-1116800.

\appendix

\section{Comparison with Strict Causality by Matsikoudis and Lee}\label{app:equivalence-causal}
We show below that the definition of strict causality by Matsikoudis and Lee~\cite{ML15} and our definition of causality are equivalent up to syntactical differences for functions that are compatible with both formalisms. Beforehand, we recall the relevant definitions from \cite{ML15}. The domains of the functions considered there consist of so-called \emph{signals}. A signal is a partial function~$\sigma \colon \tilde{\timesp} \to \valsp$, for some poset~$(\tilde{\timesp}, \preceq)$ and some set~$\valsp$ \cite[Definition~2.2]{ML15}. We denote the set of all signals by~$\Sigma$. Strict causality is then defined as follows.
\begin{definition}\label{def:strictly-causal-ML15}
	Let $F \colon \Sigma \to \Sigma$ be a partial function. Then, $F$ is \emph{strictly causal} if there is a partial function~$f \colon \Sigma \times \tilde{\timesp} \to \Sigma$ such that for all $\sigma$ in the domain of~$F$ and for all $\tau \in \tilde{\timesp}$,
	\begin{equation*}
		F(\sigma)(\tau) = f\bigl(\sigma|_{\{\tau' \in \tilde{\timesp} \mid \tau' \prec \tau\}}, \tau \bigr),
	\end{equation*}
	where equality here means that either both terms are undefined, or they are both defined and equal.
\end{definition}

To be compatible with these definitions, we only consider the special case of our system algebra where $\mathcal{X} = \{X \subseteq P \mid (X, \preceq) \text { is well-ordered}\}$ for a poset~$(P, \preceq)$ with $P = \valsp \times \timesp$ for some set~$\valsp$ and a poset~$(\timesp, \leq)$, and
$(\val_1, \tim_1) \preceq (\val_2, \tim_2)$ if and only if $\tim_1 < \tim_2$ or $(\val_1, \tim_1) = (\val_2, \tim_2)$.
Let $\inSet$ be a set and consider the poset~$\bigl(\tilde{\timesp}, \preceq \bigr)$, where $\tilde{\timesp} \coloneqq \inSet \times \timesp$ and $(\inIn_1, \tim_1) \preceq (\inIn_2, \tim_2)$ if and only if $\tim_1 < \tim_2$ or $(\inIn_1, \tim_1) = (\inIn_2, \tim_2)$. This means that, intuitively, the interface identifier $\inIn \in \inSet$ in our model becomes part of the ``time'' $\tau \in \tilde{\timesp}$ during the translation to the model of~\cite{ML15}. For $\inInSet \subseteq \inSet$, a tuple~$\tuple{X} \in \mathcal{X}^\inInSet$ can then be viewed as a signal via the injection
\begin{equation*}
	\varphi_{\inInSet} \colon \mathcal{X}^\inInSet \to \Sigma, \ \tuple{X} \mapsto \bigl\{\bigl((\inIn, \tim), \val \bigr) \in (\inInSet \times \timesp) \times \valsp \mid \inIn \in \inInSet \wedge (\val, \tim) \in \tuple{X}(\inIn) \bigr\}.
\end{equation*}
Note that $\varphi_{\inInSet}(\tuple{X})$ is indeed a partial function $\tilde{\timesp} \to \valsp$ because for each $(\inIn,\tim) \in \inInSet \times \timesp$, there is at most one~$\val \in \valsp$ such that $(\val, \tim) \in \tuple{X}(\inIn)$ since $\tuple{X}(\inIn)$ is well-ordered.

The following lemma shows that the Definitions~\ref{def:causalSys} and~\ref{def:strictly-causal-ML15} are essentially equivalent for the functions considered here.
\begin{lemma}\label{lem:causality-equivalence}
	Let $\inInSet$ and $\outInSet$ be finite disjoint sets and let $\sys{s} \colon \mathcal{X}^\inInSet \to \mathcal{X}^\outInSet$. Then, $\sys{s}$ is a causal $(\inInSet,\outInSet)$-system over~$\mathcal{X}$ in the sense of \autoref{def:causalSys} if and only if $\varphi_{\outInSet} \circ \sys{s} \circ \varphi_{\inInSet}^{-1}$ is strictly causal in the sense of \autoref{def:strictly-causal-ML15}.\footnote{The domain of $\varphi_{\outInSet} \circ \sys{s} \circ \varphi_{\inInSet}^{-1}$ is the set of all $\sigma \in \Sigma$ such that there exists $\tuple{X} \in \mathcal{X}^{\inInSet}$ with $\varphi_{\inInSet}(\tuple{X}) = \sigma$.}
\end{lemma}
\begin{proof}
	First assume that $\sys{s}$ is a causal $(\inInSet,\outInSet)$-system over~$\mathcal{X}$ and consider the partial function~$f \colon \Sigma \times \timesp \to \Sigma$ defined by
	\begin{equation*}
		f(\sigma, \tau) \coloneqq \varphi_{\outInSet} \circ \sys{s} \circ \varphi_{\inInSet}^{-1}(\sigma)(\tau).
	\end{equation*}
	Now let $\sigma \in \Sigma$ and $\tuple{X} \in \mathcal{X}^\inInSet$ such that $\varphi_{\inInSet}(\tuple{X}) = \sigma$. Further let $\tau = (\outIn, \tim) \in \outInSet \times \timesp$. We then have by \autoref{def:causalSys} for $\tuple{X}' \coloneqq \varphi_{\inInSet}^{-1}\bigl(\sigma|_{\{\tau' \in \tilde{\timesp} \mid \tau' \prec \tau\}} \bigr)$,
	\begin{equation*}
		\forall y \in \bigl(\sys{s}(\tuple{X})(\outIn) \symdif \sys{s}(\tuple{X}')(\outIn) \bigr) \ \exists \inIn \in \inInSet \ \exists x \in \bigl(\tuple{X}(\inIn) \symdif \tuple{X}'(\inIn)\bigr) \ \ x \prec y.
	\end{equation*}
	By definition of $\tuple{X}'$, there are no $\inIn \in \inInSet$ and $(\val', \tim') \in \bigl(\tuple{X}(\inIn) \symdif \tuple{X}'(\inIn) \bigr)$ with $\tim' < \tim$. This implies that there is no $\val \in \valsp$ with $(\val, \tim) \in \bigl(\sys{s}(\tuple{X})(\outIn) \symdif \sys{s}(\tuple{X}')(\outIn) \bigr)$, and therefore
	\begin{equation*}
		\begin{split}
			\varphi_{\outInSet} \circ \sys{s} \circ \varphi_{\inInSet}^{-1}(\sigma)(\tau) &= \varphi_{\outInSet} \circ \sys{s}(\tuple{X})(\tau) = \varphi_{\outInSet} \circ \sys{s}(\tuple{X}')(\tau) = \varphi_{\outInSet} \circ \sys{s} \circ \varphi_{\inInSet}^{-1}\bigl(\sigma|_{\{\tau' \in \tilde{\timesp} \mid \tau' \prec \tau\}} \bigr)(\tau) \\
			&= f \bigl(\sigma|_{\{\tau' \in \tilde{\timesp} \mid \tau' \prec \tau\}}, \tau \bigr).
		\end{split}
	\end{equation*}
	Hence, $\varphi_{\outInSet} \circ \sys{s} \circ \varphi_{\inInSet}^{-1}$ is strictly causal.
	
	To prove the opposite direction, assume $\varphi_{\outInSet} \circ \sys{s} \circ \varphi_{\inInSet}^{-1}$ is strictly causal and let $f \colon \Sigma \times \timesp \to \Sigma$ be a partial function such that for all $\sigma$ in the domain of~$\varphi_{\outInSet} \circ \sys{s} \circ \varphi_{\inInSet}^{-1}$ and for all $\tau \in \tilde{\timesp}$, we have $\varphi_{\outInSet} \circ \sys{s} \circ \varphi_{\inInSet}^{-1}(\sigma)(\tau) = f\bigl(\sigma|_{\{\tau' \in \tilde{\timesp} \mid \tau' \prec \tau\}}, \tau \bigr)$. Now let $\tuple{X}, \tuple{X}' \in \mathcal{X}^{\inInSet}$, $\outIn \in \outInSet$, and $(\val, \tim) = y \in \bigl(\sys{s}(\tuple{X})(\outIn) \symdif \sys{s}(\tuple{X}')(\outIn) \bigr)$. We then have $\varphi_{\outInSet}(\sys{s}(\tuple{X}))(\outIn, \tim) \neq \varphi_{\outInSet}(\sys{s}(\tuple{X}'))(\outIn, \tim)$ and thus for $\tau \coloneqq (\outIn, \tim)$,
	\begin{equation*}
		f\bigl(\varphi_{\inInSet}(\tuple{X})|_{\{\tau' \in \tilde{\timesp} \mid \tau' \prec \tau\}}, \tau \bigr) = \varphi_{\outInSet} \circ \sys{s} (\tuple{X})(\tau) \neq \varphi_{\outInSet} \circ \sys{s} (\tuple{X}')(\tau) = f\bigl(\varphi_{\inInSet}(\tuple{X}')|_{\{\tau' \in \tilde{\timesp} \mid \tau' \prec \tau\}}, \tau \bigr).
	\end{equation*}
	Hence, $\varphi_{\inInSet}(\tuple{X})|_{\{\tau' \in \tilde{\timesp} \mid \tau' \prec \tau\}} \neq \varphi_{\inInSet}(\tuple{X}')|_{\{\tau' \in \tilde{\timesp} \mid \tau' \prec \tau\}}$. Therefore, there exist $\inIn \in \inInSet$ and $x \in \bigl(\tuple{X}(\inIn) \symdif \tuple{X}'(\inIn)\bigr)$ with $x \prec y$. This shows that $\sys{s}$ satisfies \autoref{def:causalSys} and concludes the proof.
\end{proof}

Note that for all finite~$\inInSet \subseteq \inSet$ and for all $\tuple{X} \in \mathcal{X}^{\inInSet}$, the relation~$\prec$ is well-founded on the domain of $\varphi_{\inInSet}(\tuple{X})$: If it contained an infinite descending chain, then there would also be an infinite descending chain in $\tuple{X}(\inIn)$ for some $\inIn \in \inInSet$ (since $\inInSet$ is finite), contradicting that $\tuple{X}(\inIn)$ is well-ordered. Thus, \cite[Theorem~4.10]{ML15} can be applied, which together with \autoref{lem:causality-equivalence} implies that for all causal $(\inInSet,\outInSet)$-systems~$\sys{s}$ over~$\mathcal{X}$, the partial function $\varphi_{\outInSet} \circ \sys{s} \circ \varphi_{\inInSet}^{-1}$ is \emph{strictly contracting}. Hence, all results in \cite{ML15} about strictly contracting partial functions can be used.

\ifelsarticle
	\section*{References}
	\bibliography{funsys}
\else
	\printbibliography
\fi

\end{document}